\documentclass[reprint,aps,prx,superscriptaddress,nofootinbib]{revtex4-2}

\usepackage[dvipsnames]{xcolor}
\usepackage[colorlinks=true,linkcolor=blue,citecolor=blue,urlcolor=blue]{hyperref}

\usepackage{orcidlink}

\usepackage{amsmath,amsfonts,amssymb,amsthm}
\usepackage{physics,mathtools,bbm}

\usepackage{graphicx}
\graphicspath{{./figures/}}
\usepackage[normalem]{ulem}

\newtheorem{Definition}{Definition}
\newtheorem{Lemma}{Lemma}
\newtheorem*{Lemma*}{Lemma}
\newtheorem{Theorem}{Theorem}
\newtheorem*{Theorem*}{Theorem}
\newtheorem{Corollary}{Corollary}
\newtheorem{Conjecture}{Conjecture}

\newcommand{\Z}{\mathbb{Z}}
\newcommand{\Cl}{\mathcal{C}\ell}
\newcommand{\Herm}{\mathrm{Herm}}
\newcommand{\Conv}{\mathrm{Conv}}
\newcommand{\STAB}{\mathrm{STAB}}
\newcommand{\ASTAB}{\mathrm{ASTAB}}
\newcommand{\CNC}{\mathrm{CNC}}
\newcommand{\ACNC}{\mathrm{ACNC}}

\newcommand{\WP}{\mathrm{WP}}
\newcommand{\AWP}{\mathrm{AWP}}
\newcommand{\PSD}{\mathrm{PSD}}

\usepackage{bbold}
\DeclareMathAlphabet{\mathbbold}{U}{bbold}{m}{n}
\DeclareMathAlphabet{\mathbb}{U}{msb}{m}{n}
\newcommand{\bbzero}{\mathbbold{0}}

\begin{document}

\title{Basis-independent stabilizerness and maximally noisy magic states}

\author{Michael Zurel\,\orcidlink{0000-0003-0333-1174}}
\email{mzurel@sfu.ca}
\affiliation{Department of Mathematics, Simon Fraser University, Burnaby, Canada}

\author{Jack Davis\,\orcidlink{0000-0003-1705-1798}}
\email{jack.davis@inria.fr}
\affiliation{DIENS, \'Ecole normale sup\'erieure, PSL University, CNRS, INRIA, 45 rue d'Ulm, Paris, 75005, France}

\date{\today}

\begin{abstract}
    Absolutely stabilizer states are those that remain convex mixtures of stabilizer states after conjugation by any unitary. Here we give a characterization of such states for multiple qudits of all prime dimensions by introducing a polytope of their allowed spectra. We illustrate this through the examples of one qubit, two qubits, and one qutrit. In particular, the set of absolutely stabilizer states for a single qubit is a ball inscribed in the stabilizer octahedron, but for higher dimensions the geometry is more complicated. For odd-prime-dimensional qudits, we also give a complete characterization of absolutely Wigner-positive states, i.e., states whose Wigner function remains nonnegative after conjugation by any unitary. In so doing, we show there are absolutely Wigner-positive states that are not absolutely stabilizer, which can be seen as a unitarily invariant version of bound magic. We then study the radii of the largest balls contained in the sets of absolutely stabilizer states and absolutely Wigner-positive states. These radii respectively tell us the lowest possible purity of nonstabilizer and Wigner-negative states. Conversely, we also find the radius of the smallest ball containing the set of absolutely Wigner-positive states, giving a tight purity-based necessary condition thereof. Finally, we apply our methods to magic witnessing and find simple sufficient conditions for determining if an $n$-qubit state is reducible to an $m$-qubit magic state for any $m<n$.
\end{abstract}

\maketitle

\section{Introduction}

Stabilizer states are of fundamental importance in quantum computation. They provide the structural foundation upon which much of the field rests, and are indispensable in the sub-fields of quantum error correction and classical simulation~\cite{Gottesman1997, Gottesman1999, AaronsonGottesman2004, CampbellBrowne2009, BravyiSmolin2016, BravyiGosset2016, BravyiHoward2019, SeddonCampbell2021}. Such states define the complementary set of non-stabilizer, or magic states, as those inexpressible as a convex mixture of pure stabilizer states. Classifying, manipulating, and distilling such magic states is essential for achieving universality in the gate-injection model of fault-tolerant quantum computation~\cite{GottesmanChuang1999, BravyiKitaev2005, HowardEmerson2014, VeitchEmerson2014, HowardCampbell2017, SeddonCampbell2019, HeimendahlGross2022}. Consequently, characterizing the structure and geometry of mixed stabilizer states is central for understanding the onset of quantum advantage.

The property of a state being stabilizer or non-stabilizer is inherently basis-dependent: a stabilizer state with respect to one computational basis can be a magic state in another basis. This is an immediate consequence of the unitary group acting transitively on pure states, with a prototypical example being the non-Clifford $T$ gate mapping the stabilizer state $\ket{+}$ into the magic state $\ket{T}$. While the fault-tolerant implementation of such a unitary is highly constrained by the Eastin-Knill theorem~\cite{EastinKnill2009}, this observation nonetheless raises an important question: which mixed quantum states remain stabilizer under arbitrary unitary conjugation? Equivalently, which sets of eigenvalues are so mixed that no choice of eigenbasis will produce a magic state?  This is analogous to the notion of absolute separability in entanglement theory~\cite{ZyczkowskiLewenstein1998,KusZyczkowski2001} where one asks whether a state remains separable under all global unitaries.  In both settings, the purpose is to characterize whether a mixed state contains any latent resourcefulness that can be unitarily activated.  States that do not have such potential can be thought of as being universally non-resourceful, a property worth characterizing in any resource theory, including in particular that of magic.  

Indeed, the idea of such \textit{absolutely stabilizer} states was a byproduct of the recent work of Liu et.\ al.~\cite{LiuRoth2025} that introduced the Triangle Criterion, a magic witnessing protocol analogous to the positive partial transpose (PPT) test for entanglement detection~\cite{Peres1996,HorodeckiHorodecki1996}.  A key feature of the protocol is its operational interpretation: a state in violation of the Triangle Criterion is equivalent to its reducibility to a single-qubit magic state using stabilizer protocols, including post-selection.  This naturally led to the question of maximally noisy magic states and to the existence of a Hilbert-Schmidt ball of absolutely stabilizer states.  This ball however only serves as a purity-based sufficient condition.  A full understanding of the shape of the absolutely stabilizer region, and its relation to magic witnessing and stabilizer reductions, remains open.

Here we develop a systematic framework for characterizing the set of absolutely stabilizer states for multi-partite systems with any prime local dimension.  Our main result shows that membership is governed by linear inequalities on the spectrum of a given state.  In other words, absolute stabilizerness is completely characterized by a polytope within the simplex of eigenvalue distributions, which we call the \emph{absolutely stabilizer spectral polytope}, or \textit{muggle polytope} for short.  This description also identifies the noisiest possible multi-partite magic states as those with spectra lying just outside the absolutely stabilizer spectral polytope. Relatedly, we show that beyond the single-qubit case, the set of absolutely stabilizer states is not the inscribed ball of the stabilizer polytope, nor any other Hilbert-Schmidt ball in the space of Hermitian operators. Even for the cases of a single qutrit and two qubits, counterexamples can be constructed. This is because purity no longer determines the spectrum of a state uniquely, implying that the subset of fixed-purity states decomposes into disjoint unitary group orbits. Consequently, any unitarily invariant set, including in particular the set of absolutely stabilizer states, need not coincide with a Hilbert–Schmidt ball. 

Our main technical approach is to analyze the polar dual of the stabilizer polytope, called the $\Lambda$ polytope, where vertices of the former correspond to facets of the latter, and vice-versa~\cite{Heimendahl2019, ZurelRaussendorf2020, ZurelHeimendahl2024}.  In particular, we show that the inequalities defining the muggle polytope are directly related to the spectra of the vertices of $\Lambda$.
In the context of multi-qubit systems, we then exploit the finer structure of $\Lambda$ using the closed and noncontextual (CNC) framework~\cite{RaussendorfZurel2020} to dramatically reduce the number of vertices that must be evaluated to test if a given state is absolutely stabilizer.  This leads to a conjecture involving a distinguished family of $\Lambda$ vertices, namely those of \textit{CNC-type}; see main text for details.  We analytically verify this conjecture for low-dimensional examples and give supporting evidence for it in higher dimensions. In the two-qubit case, this refinement is provably optimal and realizes the muggle polytope as being generated by only two inequalities, a reduction in redundancy from 22{,}320 vertices of $\Lambda$.  Moreover, we show that the same distinguished family of CNC-type vertices contains precisely the magic witnesses that define the Triangle Criterion.  
Using this observation we generalize the Triangle Criterion to the multi-qubit magic setting and show it has a similar operational interpretation: a restricted family of magic witnesses whose violations certify reducibility to magic under stabilizer protocols.  This generalization, while strongly related, is independent of the validity of our conjecture.

We also use our framework to analyze the Hilbert-Schmidt inradius of the stabilizer polytope, i.e., the radius of the largest Hilbert-Schmidt ball contained in the polytope.  This radius is found to depend strongly on whether the local dimension is even or odd, a consequence of the existence of state-independent proofs of contextuality using multiqubit Pauli operators~\cite{Mermin1993}, which don't exist in odd dimensions~\cite{OkayRaussendorf2017,RaussendorfFeldmann2023}. For odd-prime dimensions we derive the inradius exactly, while for qubits the radius we find relies on a conjecture about the structure of the $\Lambda$ polytope. Our conjecture exactly matches that of Conjecture~1 in the recent work~\cite{LiuRoth2025}, offering a different perspective via the framework of closed and noncontextual sets of Pauli operators.

Finally, we discuss the related notion of absolute Wigner positivity~(AWP) in odd-prime dimensions~\cite{SalazarJunior2022}. Such states are those whose Wigner function~\cite{Wootters1987, Gross2006, WangSu2019} remains nonnegative under arbitrary unitary conjugation. Our techniques for characterizing the spectra carry over directly to the AWP setting and yield the \emph{AWP polytopes} in the simplex of spectra.  Similar polytopes have been constructed for other kinds of Wigner functions~\cite{AbgaryanTorosyan2020, DenisMartin2024}, but here we study them for the Wigner function most prominent in quantum computation~\cite{Wootters1987, Gross2006, WangSu2019}. This object is more tractable than the muggle polytope due to relative simplicity of the Wigner polytope (i.e.~the convex hull of the phase space point operators) as compared to the stabilizer polytope. We find that the AWP set is similarly not the Hilbert-Schmidt inball of the Wigner polytope, and moreover that in the case of a single qutrit the muggle polytope is a strict subset of the AWP polytope. That is, there exist mixed qutrit states that are non-stabilizer yet no global unitary can generate negativity. This inclusion of spectral polytopes can be seen as a unitarily-invariant version of bound magic~\cite{VeitchEmerson2012}; see also~\cite{CampbellBrowne2010, PrakashGupta2020}. Finally, we show that the inball of the Wigner polytope matches the inball of the stabilizer polytope in odd-prime dimensions.

The rest of the article is organized as follows. In Section~\ref{Section:Background} we recall the relevant background, including the stabilizer subtheory, the Wigner function, and the framework of closed and noncontextual sets of Pauli operators. Section~\ref{Section:MuggleMain} then presents our spectral characterization of absolutely stabilizer states, and we illustrate this characterization through some examples. In Section~\ref{Section:AWP}, we fully characterize the set of absolutely Wigner positive states. In Section~\ref{Section:StabilizerInradius} we find the radii of the largest ball contained in the sets of absolutely stabilizer and absolutely Wigner positive states, and we compare these radii against other related radii. In Section~\ref{Section:ACNC} we generalize the Triangle Criterion by introducing witnesses of multi-qubit magic, and connect it to our developed framework. Section~\ref{Section:Discussion} concludes with a discussion of our results.

\section{Background}\label{Section:Background}

\subsection{The stabilizer formalism}

Here we recall certain aspects of the stabilizer formalism~\cite{Gottesman1997,Gottesman1999,Appleby2005,Gross2006,DeBeaudrap2013,GrossWalter2013,SingalHsieh2023,ZurelHeimendahl2024}. While some of the following definitions and statements are valid in all $n$-qudit Hilbert spaces $\mathcal{H} \simeq (\mathbb C^d)^{\otimes n}$, we restrict attention throughout to prime qudit dimension $d$.

The qudit Pauli $X$ and $Z$ operators for a single qudit ($n=1$) are defined by their action on the computational basis $\ket{j}$,
\begin{equation}
    X\ket{j}=\ket{j+1}, \quad Z \ket{j} = \omega^j \ket{j}, \quad \forall j\in \Z_d := \Z/d\Z,
\end{equation}
where the arithmetic inside a ket is modulo $d$ and $\omega = e^{2\pi i/d}$ is the first primitive $d^{\text{th}}$ root of unity. For a system of $n$ qudits these are used to build the set of discrete-variable displacement operators, here called generalized Pauli operators.  These are labeled by vectors of the vector space $\Z_d^{n} \times \Z_d^n \simeq \Z_d^{2n}$. In particular, for $u=(u_z, u_x)\in\Z_d^n\times\Z_d^n$, the generalized Pauli operator is
\begin{equation}\label{Weyl_def_tau}
    T_{u}:=\tau^{- u_z \cdot u_x}\bigotimes_{k=1}^n Z^{[u_z]_k} X^{[u_x]_k},
\end{equation}
where $\tau = (-1)^d e^{i\frac{\pi}{d}}$~\cite{DeBeaudrap2013}. The \emph{Pauli group} is the group generated by the generalized Pauli operators, $\mathcal{P} = \langle \{ T_u \;|\; u \in \Z_d^{2n} \}\rangle$, and the \emph{Clifford group}, $\Cl$, is the normalizer of the Pauli group with respect to the unitary group $\mathrm{U}(d^n)$ modulo phases. The centre of the Pauli group $\mathcal{Z}(\mathcal P)$ is isomorphic to the additive group $\Z_d$ for odd $d$ and to $\Z_4$ for $d=2$, but in all cases $\mathcal{P} / \mathcal{Z}(\mathcal P)\simeq\Z_d^{2n}$.

The vector space $\Z_d^{2n}$ is equipped with the standard symplectic form
\begin{equation}
    [u,v] := u_z \cdot v_x - u_x \cdot v_z \mod d,
\end{equation}
which computes the group commutator between pairs of Pauli operators,
\begin{equation}\label{DV_CCR}
    [T_{u},T_{v}] := T_{u}T_{v}T_{u}^{-1}T_{v}^{-1} = \omega^{[u,v]}\mathbb{I}.
\end{equation}
A subspace $I \subset \Z_d^{2n}$ over which the symplectic form vanishes is called \emph{isotropic}, while an isotropic subspace of dimension $n$, which is the maximal dimension possible for such subspaces, is called \emph{Lagrangian}.  Isotropic subspaces are in one-to-one correspondence with equivalence classes of abelian subgroups of $\mathcal{P}$ up to phases. Bases of a fixed isotropic subspace are in correspondence with generating sets of the associated abelian subgroup, again up to phases.

We also define a $\Z_d$-valued function $\beta$ that tracks how \emph{commuting} Pauli operators compose through the relation
\begin{equation}\label{Equation:BetaFunction}
    T_u T_v=\omega^{-\beta(u,v)}T_{u+v},
\end{equation}
for all $u,v\in \Z_d^{2n}$ such that $[u,v]=0$. With the phase convention chosen above in Eq.~\eqref{Weyl_def_tau}, the odd-dimensional case reduces to $T_u T_v = \omega^{[u,v]\cdot2^{-1}}T_{u+v}$, and so $\beta\equiv 0$. For qubits this is no longer the case, and, in fact, no phase prefactor replacing the one in Eq.~\eqref{Weyl_def_tau} exists that can make $\beta\equiv0$ for qubits~\cite{RaussendorfFeldmann2023}.

This has important consequences for sets of commuting Pauli operators, which can be measured simultaneously. In particular, if a measurement of commuting operators $T_u$ and $T_v$ is performed yielding outcomes $\omega^{a}$ and $\omega^{b}$ respectively, the outcome $\omega^c$ for the measurement of observable $T_{u+v}$ is uniquely determined via the operator relation
\begin{equation}\label{Equation:NoncontextualityOperatorRelation}
    \omega^{-a-b}T_uT_v=\omega^{-a-b-\beta(u,v)}T_{u+v}=\omega^{-c}T_{u+v}.
\end{equation}
More generally, if a set of commuting Paulis $\{T_{u_i}\}_{u_i}$ are measured yielding outcomes $\{\omega^{r(u_i)}\}_{u_i}$, the outcomes for the isotropic subspace $I:=\mathrm{span}\{u_1,u_2,\dots\}$ generated by these elements are fully determined. Therefore, we can label a measurement by an isotropic subspace $I$ and a set of measurement outcomes via a function $r:I\to\Z_d$ satisfying consistency conditions like Eq.~\eqref{Equation:NoncontextualityOperatorRelation} for all pairs in $I$.  Equivalently, with Eq.~\eqref{Equation:BetaFunction}, $r$ satisfies
\begin{equation}\label{stab_consistent}
    r(u)+r(v)=r(u+v)-\beta(u,v)\quad\forall u,v\in I.
\end{equation}
Hence, an assignment of measurement outcomes to a collection of commuting Pauli operators is not always linear over an isotropic basis.
For such a pair $(I,r)$, the corresponding projector is
\begin{equation}\label{Equation:PauliProjector}
    \Pi_I^r=\frac{1}{|I|}\sum\limits_{u\in I}\omega^{-r(u)}T_u.
\end{equation}
In the case where $I=L$ is a Lagrangian subspace, the projector $\Pi_L^r$ is rank one and can be identified with a pure quantum state, this is a pure stabilizer state.  In fact, any pure stabilizer state can be uniquely identified with such a pair $(L,r)$.

\subsection{The stabilizer and \texorpdfstring{$\Lambda$}{Lambda} polytopes}

We write $\Herm(\mathcal{H})$ for the real vector space of Hermitian operators on the Hilbert space $\mathcal{H} \simeq (\mathbb C^d)^{\otimes n}$, while $\Herm_c(\mathcal{H})$ denotes the affine subspace of $\Herm(\mathcal{H})$ defined by $\Herm_c(\mathcal{H}):=\left\{X\in\Herm(\mathcal{H})\;\big|\;\Tr(X)=c\right\}$. Let $\mathcal{S}$ denote the set of pure $n$-qudit stabilizer states. The \emph{stabilizer polytope} is the convex hull of pure stabilizer state projectors,
\begin{equation}
    \STAB:=\Conv(\{\,\ketbra{\sigma}{\sigma}\;\big|\;\ket{\sigma}\in\mathcal{S}\}).
\end{equation}
Any element of $\STAB$ is often referred to as a stabilizer state, and the abstract property of membership is sometimes called stabilizerness. The $n$-qudit $\Lambda$ \emph{polytope} is
\begin{equation}\label{Equation:LambdaDefinition}
    \Lambda := \left\{X\in\Herm_1(\mathcal{H}) \, \big|\,\Tr(X\ketbra{\sigma}{\sigma})\ge0\;\,\forall\ket{\sigma}\in\mathcal{S}\right\}.
\end{equation}
Equivalently, $\Lambda$ is the restriction of the \textit{dual cone} of $\STAB$ in $\Herm(\mathcal H)$ to the affine space $\Herm_1(\mathcal{H})$ of unit-trace Hermitian matrices. These polytopes have many applications, including being used to define classical simulation algorithms for quantum computation~\cite{HowardCampbell2017, ZurelRaussendorf2020, ZurelHeimendahl2024} and to characterize magic state distillation~\cite{Reichardt2009}.

Both $\STAB$ and $\Lambda$ live in the affine space $\Herm_1(\mathcal{H})$, but it will sometimes be convenient to translate these polytopes to the traceless subspace $\Herm_0(\mathcal{H})$, the generalization of Bloch space:
\begin{equation}
    \STAB_0 := \left\{X-\frac{1}{d^n}\mathbb{I}\;\bigg|\;X\in\STAB\right\},
\end{equation}
and
\begin{equation}
    \Lambda_0 := \left\{X-\frac{1}{d^n}\mathbb{I}\;\bigg|\;X\in\Lambda\right\}.
\end{equation}
With respect to the Hilbert-Schmidt inner product, $\langle X,Y\rangle=\Tr(XY)$, $\Lambda_0$ is the polar dual of $-d^n\STAB_0$. That is,
\begin{equation}\label{lambda_polar_dual_dnSTAB}
\begin{aligned}
    \Lambda_0 =& (-d^n\STAB_0)^\circ\\
    =& \{X\in\Herm_0(\mathcal{H})\;|\;\Tr(XY)\le1\;\forall Y\in-d^n\STAB_0\}.
\end{aligned}
\end{equation}

\subsection{The Wigner function}

The Wigner function is a quasiprobability representation that describes quantum states, transformations, and measurements as real-valued functions over $\Z_d^{2n}$.  Here we use the well-studied variant, valid in at least odd-prime dimensions, where negativity therein is linked to quantum contextuality and Wigner-positive circuits can be efficiently simulated on a classical computer~\cite{Wootters1987, Gross2006, VeitchEmerson2012, MariEisert2012, PashayanBartlett2015, DelfosseRaussendorf2017, WangSu2019}.  It arises via an operator basis known as the displaced parity basis or the set of \emph{phase space point operators}.  First define the multi-qudit parity operator as an expansion in the computational basis
\begin{equation}\label{Equation:ParityOperator}
    A_0 := \sum_{j \in \Z_d^n} \ketbra{-j}{j} = \frac{1}{d^n}\sum_{u \in \Z_d^{2n}} T_u,
\end{equation}
where again arithmetic inside kets is modulo $d$.  The remaining phase space point operators are displaced versions thereof:
\begin{equation}
    A_u = T_u A_0 T^\dagger_u \in \Herm_1(\mathcal H),\quad\forall u\in\Z_d^{2n}.
\end{equation}
These operators are Hermitian, have unit trace, and form an orthogonal basis in the space of operators, satisfying the orthogonality relation $\Tr(A_uA_v)=d^n\delta_{u,v}$.

The Wigner function, $W_\rho:\Z_d^{2n}\to\mathbb{R}$, of a state $\rho$ is given by the expansion coefficients of a state in this basis,
\begin{equation}\label{Equation:WignerDefinition1}
    \rho=\sum\limits_{v\in\Z_d^{2n}}W_\rho(v)A_v,
\end{equation}
which, by orthogonality, can be extracted via
\begin{equation}\label{Equation:WignerDefinition2}
    W_\rho(u)=\frac{1}{d^n}\Tr(\rho A_u).
\end{equation}

The \emph{Wigner polytope}, denoted $\WP$, is the set of unit-trace Hermitian operators with nonnegative Wigner function.  By orthogonality of the phase space point operators and the two descriptions of the Wigner function, Eq.~\eqref{Equation:WignerDefinition1} and Eq.~\eqref{Equation:WignerDefinition2}, $\WP$ has dual descriptions: (1)~as the convex hull of the phase space point operators
\begin{equation}\label{Equation:WignerPolytope1}
    \WP=\Conv(\{A_u\;|\;u\in\Z_d^{2n}\}),
\end{equation}
and (2)~as an intersection of halfspaces
\begin{equation}\label{Equation:WignerPolytope2}
    \WP=\{X\in\Herm_1(\mathcal{H})\;|\;\Tr(XA_u)\ge0\;\forall u\in\Z_d^{2n}\}.
\end{equation}
Using a similar argument as in the previous subsection, this makes $\WP_0 := \WP-\mathbb{I}/d^n$ the polar dual of $-d^n\WP_0$ with respect to the Hilbert-Schmidt inner product.

A \emph{Wigner-positive} state is one whose Wigner function is everywhere nonnegative: $W_\rho(u) \geq 0$ for all $u \in \Z_d^{2n}$ (we follow convention and do not refer to such states as Wigner-nonnegative).  The set of Wigner-positive states is the intersection of the Wigner polytope with the set of quantum states.  This is known to be a strict superset of the stabilizer polytope, a phenomenon known as \emph{bound magic}~\cite{VeitchEmerson2012}.  See also~\cite{CampbellBrowne2010, PrakashGupta2020} for related notions of bound magic.

\subsection{CNC operators}\label{sec:CNC_background}

Closed and noncontextual (CNC) operators were introduced in Ref.~\cite{RaussendorfZurel2020} in order to define a quasiprobability function for qubits that is useful for simulating quantum computations, similar to the Wigner function for odd-dimensional qudits~\cite{MariEisert2012,VeitchEmerson2012}. They are defined with respect to certain sets of Pauli observables which do not contain state-independent proofs of contextuality like the Mermin square~\cite{Mermin1993}.

\begin{Definition}
    A set $\Omega\subset\Z_d^{2n}$ is \emph{closed under inference} if for all $u,v\in\Omega$,
    \begin{equation}
        [u,v]=0 \;\Longrightarrow\; u+v\in\Omega.
    \end{equation}
\end{Definition}

In terms of observables, a set $\Omega$ is closed under inference if for any two commuting Pauli operators $T_u$ and $T_v$ with $u,v\in\Omega$, the vector $u+v$ corresponding to the product of the Pauli operators, $T_uT_v\propto T_{u+v}$, is also in $\Omega$. The idea is, if $T_u$ and $T_v$ commute, and measurement outcomes of $T_u$ and $T_v$ are known, then the outcome of a measurement of $T_{u+v}$ could be inferred. Closure under inference is also called partial closure in Ref.~\cite{KimAbramsky2023}.

\begin{Definition}
    Let $\Omega\subset\Z_d^{2n}$ be a set that is closed under inference.  A \emph{noncontextual value assignment} for $\Omega$ is a function $\gamma:\Omega\to\Z_d$ such that for all $u,v\in\Omega$ with $[u,v]=0$,
    \begin{equation}\label{Equation:NoncontextualityCondition}
        \gamma(u)+\gamma(v) = \gamma(u+v) - \beta(u,v).
    \end{equation}
\end{Definition}

Such a value assignment can be seen as a generalization of the above stabilizer consistency relation Eq.~\eqref{stab_consistent} from isotropic subspaces to sets that are closed under inference.

\begin{Definition}
    A set $\Omega\subset\Z_d^{2n}$ is \emph{noncontextual} if there exists a noncontextual value assignment function $\gamma:\Omega\to\Z_d$.
\end{Definition}

A set $\Omega\subset\Z_d^{2n}$ that is both closed under inference and noncontextual is called \emph{CNC}. A CNC operator is defined as follows.

\begin{Definition}
    For a pair $(\Omega,\gamma)$ where $\Omega\subset\Z_d^{2n}$ is CNC and $\gamma:\Omega\to\Z_d$ is a noncontextual value assignment for $\Omega$, the corresponding \emph{CNC operator} is
    \begin{equation}\label{Equation:CNCOperators}
        A_\Omega^\gamma=\frac{1}{d^n}\sum\limits_{v\in\Omega}\omega^{-\gamma(v)}T_v.
    \end{equation}
\end{Definition}

The CNC sets $\Omega\subset\Z_d^{2n}$ and noncontextual value assignments $\gamma:\Omega\to\Z_d$ have been characterized for qudits of all prime dimensions~\cite{RaussendorfZurel2020,ZurelHeimendahl2024b}. The characterization looks different for qubits than for odd-prime-dimensional qudits.

\subsubsection{Qubits}

Not every set of multiqubit Pauli observables is noncontextual because of the existence of Mermin square-style proofs of contextuality~\cite{Mermin1993}. Intuitively, a CNC set is simply a collection of multiqubit Pauli observables that does not contain a state-independent proof of contextuality. In fact, if a set is closed under inference, then for it to be noncontextual, it suffices that it does not contain a set of operators with the same commutation relations as the nine Mermin square observables~\cite[Supplementary Material, Lemma~3]{ZurelRaussendorf2020} (see Figure~\ref{Figure:CNCMerminSquare} for some examples). With this condition, the multiqubit CNC sets can be fully characterized and classified. This is achieved by the following theorem.

\begin{figure}
    \centering
    \includegraphics[width=0.45\linewidth]{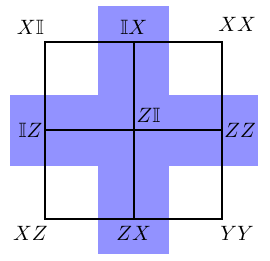}
    \hspace{0.05\linewidth}
    \includegraphics[width=0.45\linewidth]{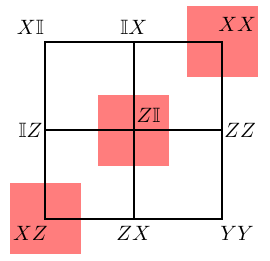}
    \caption{Two Mermin squares each with a CNC subset.  By construction, such subsets are chosen to avoid a Mermin-square style proof of contextuality~\cite{Mermin1993}. See Theorem~\ref{Theorem:CNCClassification} for the full characterization of CNC sets.}
    \label{Figure:CNCMerminSquare}
\end{figure}

\begin{Theorem}[\cite{RaussendorfZurel2020}; \S IV]\label{Theorem:CNCClassification}
    A set $\Omega\subset\Z_2^{2n}$ is CNC if and only if it has the form
    \begin{equation}\label{Equation:QubitCNCSets}
        \Omega=\bigcup\limits_{k=1}^\xi\langle a_k,I\rangle,
    \end{equation}
    where $I\subset\Z_2^{2n}$ is an $(n-m)$-dimensional isotropic subspace for some $m\in\{1,\dots,n\}$, $1\le\xi\le 2m+1$, $\{a_k\}_k$ pairwise anticommute, and all $a_k$ commute with every element of $I$. Further, a CNC set is inclusion-maximal if and only if $\xi=2m+1$.
\end{Theorem}

To get any noncontextual value assignments $\gamma$ on a CNC set $\Omega\subset\Z_2^{2n}$, we can choose the values of $\gamma$ freely on a basis for the subspace $I$, and on the elements $a_1,\dots,a_\xi$. Then the values assigned to the remaining elements of $\Omega$ can be inferred using Eq.~\eqref{Equation:NoncontextualityCondition}.
This characterization of the multiqubit CNC pairs $(\Omega,\gamma)$ is equivalent to the following characterization of the multiqubit CNC operators.

\begin{Corollary}\label{Corollary:CNCStructure}
    For any CNC set $\Omega\subset\Z_2^{2n}$ and any noncontextual value assignment $\gamma:\Omega\to\Z_2$, there exists an $m\in\{0,\dots,n\}$, an $(n-m)$-qubit stabilizer state $\ket{\sigma}$, an $n$-qubit Clifford group element $g$, and another CNC pair $(\tilde{\Omega},\tilde{\gamma})$ with $\tilde{\Omega}\subset\Z_2^{2m}$ such that
    \begin{equation}\label{CNC_normal_form}
    A_\Omega^\gamma = g \left( A_{\tilde{\Omega}}^{\tilde{\gamma}} \otimes \ketbra{\sigma}{\sigma} \right)g^\dagger,
    \end{equation}
    and the nonzero vectors in the CNC set $\tilde{\Omega}$ pair-wise anticommute. The set $\Omega$ is maximal if and only if $m\ge1$ and $|\tilde{\Omega}|=2m+2$.
\end{Corollary}
Here, up to the conjugation by $g$, the isotropic subspace $I \subset \Z_{2}^{2n}$ in Eq.~\eqref{Equation:QubitCNCSets} together with the value assignment $\gamma|_I$ specify the stabilizer state $\ket{\sigma}$, and the elements $a_k$ in Eq.~\eqref{Equation:QubitCNCSets} specify the nonzero elements of $\tilde{\Omega}$. Since there are no nontrivial constraints imposed by Eq.~\eqref{Equation:NoncontextualityCondition} on the value assignment $\tilde{\gamma}:\tilde{\Omega}\to\Z_2$, the value of $\tilde{\gamma}$ can be chosen freely on the nonzero elements of $\tilde{\Omega}$.

For every number $n$ of qubits, CNC operators $A_\Omega^\gamma$ corresponding to maximal CNC sets $\Omega$ are vertices of the $\Lambda$ polytopes~\cite[Supplementary Material, Sec.~IV]{ZurelRaussendorf2020}. On $n$ qubits, there are $n$ orbits of CNC-type vertices of $\Lambda$ under the action of the Clifford group. They are labeled by the parameter $m$, which takes integer values between $1$ and $n$ and determines the size and commutation structure of the CNC set via Theorem~\ref{Theorem:CNCClassification}.

\subsubsection{Odd-prime-dimensional qudits}

In contrast to the qubit case, with odd-prime-dimensional qudits there are no state-independent proofs of contextuality among the Pauli operators~\cite{OkayRaussendorf2017,RaussendorfFeldmann2023}, so a set $\Omega\subset\Z_d^{2n}$ being CNC is equivalent to it being closed under inference. The CNC pairs $(\Omega,\gamma)$ have also been characterized for odd-prime dimensional qudits.

\begin{Theorem}[\cite{ZurelHeimendahl2024b}; Theorem~1]
    A set $\Omega\subset\Z_d^{2n}$ is closed under inference if and only if
    \begin{itemize}
        \item[(i)] $\Omega$ is a subspace of $\Z_d^{2n}$ (i.e., $\Omega$ is closed under addition), or
        \item[(ii)] $\Omega$ has the form
        \begin{equation}\label{Equation:QuditCNCSets}
            \Omega=\bigcup\limits_{k=1}^\xi\langle a_k,I\rangle
        \end{equation}
        where $I$ is an $(n-m)$-dimensional isotropic subspace of $\Z_d^{2n}$ for some $m\in\{1,\dots,n\}$, $\{a_k\}_k$ pair-wise do not commute, and all $a_k$ commute with every element of $I$.
    \end{itemize}
\end{Theorem}
Unlike the qubit case, here the maximum value of $\xi$ in Eq.~\eqref{Equation:QuditCNCSets} is not known~\cite{Chin2005}.

One example of a maximal CNC set is the set $\Omega=\Z_d^{2n}$ of all Pauli operators. For multiple qudits ($n\ge2$), the noncontextuality condition, Eq.~\eqref{Equation:NoncontextualityCondition}, is enough to force the noncontextual value assignments, $\gamma:\Z_d^{2n}\to\Z_d$, to be linear functions~\cite{DelfosseRaussendorf2017}. The corresponding CNC operators, Eq.~\eqref{Equation:CNCOperators}, are exactly the phase space point operators of the Wigner function~\cite{Wootters1987, Gross2006, Gross2008, DelfosseRaussendorf2017}. These operators are vertices of the multi-qudit $\Lambda$ polytopes for every number of qudits~\cite[Theorem~3]{ZurelHeimendahl2024}.

For a single qudit ($n=1$), the situation is different. Since the noncontextuality condition does not impose linearity of the value assignment on all of $\Z_d^{2n}$, but only on commuting pairs of elements in $\Z_d^{2n}$, there exist CNC pairs $(\Z_d^{2n},\gamma)$ where $\gamma$ is non-linear. There are $d^2$ CNC operators corresponding to Wigner function phase space points, but there are another $d^{d+1}-d^2$ CNC operators corresponding to nonlinear value assignments on $\Z_d^{2n}$. For a single-qudit, these are also vertices of $\Lambda$, and in fact, all of the single-qudit $\Lambda$-polytope vertices are CNC operators with $\Omega=\Z_d^{2n}$. There are also CNC pairs $(\Omega,\gamma)$ with non-linear $\gamma$ for multiple qudits, but only when $\Omega\subsetneq\Z_d^{2n}$~\cite{ZurelHeimendahl2024b}. These are not known to be vertices of $\Lambda$ in general.

\section{Absolutely stabilizer states}\label{Section:MuggleMain}

We now discuss our central topic of absolutely stabilizer states, beginning with their formal definition.  We then derive our spectral characterization via a study of the $\Lambda$ polytope.  After separately specializing to qubits and odd-prime qudits, which are structurally different, we discuss the explicit cases of a single qubit, two qubits, and a single qutrit.

\subsection{Definition}
\label{sec:absolutely_stabilizer_definition}

\begin{Definition}\label{Definition:AbsolutelyStabilizerStates}
    A state $\rho$ is \emph{absolutely stabilizer} if it remains a convex mixture of pure stabilizer states under arbitrary unitary conjugation,
    \begin{equation}
        U \rho U^\dagger \in \STAB, \quad \forall \, U \in \mathrm{U}(d^n).
    \end{equation}
\end{Definition}

It follows immediately from Definition~\ref{Definition:AbsolutelyStabilizerStates} that the set of such states is the intersection of all unitarily transformed stabilizer polytopes,
\begin{equation}\label{Equation:MuggleStates}
    \begin{aligned}
        \ASTAB := \!\! \bigcap\limits_{U\in \,\mathrm{U}(d^n)} \!\! U \,\STAB \, U^\dagger.
    \end{aligned}
\end{equation}
We also call $\ASTAB$ the \emph{muggle set} and absolutely stabilizer states \emph{muggle states}, due to  the equivalent characterization of such states as those that cannot be transformed into any magic state with any unitary\footnote{Note that our usage of unitary invariance differs from~\cite{WagnerGalvao2025}, which focuses on sets of states rather than individual ones.}.  The name \emph{absolutely stabilizer} is analogous to the notion of absolute separability from entanglement theory, i.e.\ states that remain separable under any unitary~\cite{ZyczkowskiLewenstein1998,KusZyczkowski2001}. The affinely shifted version of $\ASTAB$ is denoted
\begin{equation}\label{Equation:MuggleStates_Bloch}
\begin{aligned}
    \ASTAB_0  &=  \ASTAB - \frac{1}{d^n} \mathbb I \\
    &= \! \!\!\bigcap\limits_{U\in \mathrm{U}(d^n)} \!\!\! U \, \STAB_0 \, U^\dagger.
\end{aligned}
\end{equation}

Note that the term ``absolutely stabilizer'' is also used in Ref.~\cite{LiuRoth2025}, but with a slightly different definition from ours.  There, an absolutely stabilizer state is one that cannot be transformed into a magic state via any unitary evolution followed by any destructive post-selection of a subsystem, including those onto magic states\footnote{Both here and in certain discussions in~\cite{LiuRoth2025}, only the reduced state of a system after a post-selection is considered, rather than the global one.  That is, the post-selection of a subsystem is here modelled as destroying the associated register, and so maps larger Hilbert spaces to smaller ones.}.  Our Definition~\ref{Definition:AbsolutelyStabilizerStates} only requires unitary invariance, however a short argument shows that $\ASTAB$ is in fact closed under such post-selections. 

\begin{Theorem}\label{Theorem:ASTABPostselection}
    Let $\ASTAB$ denote the $n$-qudit set of absolutely stabilizer states where $d$ is any prime.  Then $\ASTAB$ is closed under destructive subsystem post-selections onto arbitrary pure states.  
\end{Theorem}

\begin{proof}[Proof of Theorem~\ref{Theorem:ASTABPostselection}]
    Associate to the post-selection of $n-k$ qudits onto the pure state $\ket{\phi} \in \mathbb C^{d^{n-k}}$ the isometry
    \begin{equation}\label{stabilizer_isometry}
    \begin{aligned}
        V: \mathbb C^{d^k} &\rightarrow \mathbb C^{d^n}, \qquad\ket{\psi} \mapsto \ket{\psi} \otimes \ket{\phi},
    \end{aligned}
    \end{equation}
    where $k>0$.  Assuming non-zero collapse probability, $\Tr[V^\dagger \rho V] >0$, a global state $\rho \in \Herm_1(\mathbb C^{d^n})$ under this post-selection is
    \begin{equation}
        \rho \mapsto \rho' = \frac{V^\dagger \rho V}{\Tr[V^\dagger \rho V]} \in \Herm_1(\mathbb C^{d^k}).
    \end{equation}
    Now choose the lift of any $U \in \mathrm{U}(d^k)$ to be $\widetilde{U} :=  (U \otimes I) \in \mathrm{U}(d^n)$.  This is done so that
    \begin{equation}
        \widetilde{U}V = V U.
    \end{equation}
    For a given $n$-qudit state $\rho \in \ASTAB$ and any unitary $U \in \mathrm{U}(d^k)$, the state $U \rho' U^\dagger$ then becomes
    \begin{equation}\label{local_to_lifted_unitary}
        \frac{ U V^\dagger \rho V  U^\dagger}{\Tr[V^\dagger \rho V]} = \frac{ U V^\dagger \rho V  U^\dagger}{\Tr[UV^\dagger \rho V U^\dagger]} = \frac{ V^\dagger \widetilde{U} \rho \widetilde{U}^\dagger V }{\Tr[V^\dagger \widetilde{U} \rho \widetilde{U}^\dagger V]}.
    \end{equation}
    As written, this does not appear to imply that $\rho'$ will be stabilizer independently of $U$, however the unitary invariance of $\ASTAB$ can be exploited.  Without loss of generality, let $G \in \mathrm{U}(d^{n-k})$ be some unitary on the discarded system that maps $\ket{\phi}$ onto some pure $(n-k)$-qudit stabilizer state: $G\ket{\phi} = \ket{\sigma}$.  Then the post-selection onto $\ket{\phi}$ can be expressed as
    \begin{align}\label{rotated_postselection}
        U\rho'U^\dagger &= \frac{ V^\dagger (U \otimes G^\dagger G) \rho (U^\dagger \otimes G^\dagger G) V }{\Tr[ V^\dagger (U \otimes G^\dagger G) \rho (U^\dagger \otimes G^\dagger G) V ]} \\
        &= \frac{ V^\dagger_{\sigma} (U \otimes G) \rho (U^\dagger \otimes G^\dagger ) V_{\sigma} }{\Tr[V^\dagger_{\sigma} (U \otimes G) \rho (U^\dagger \otimes G^\dagger ) V_{\sigma} ]},
    \end{align}
    where
    \begin{equation}
        V_{\sigma} := (I\otimes G) V, \qquad \ket{\psi} \mapsto \ket{\psi} \otimes \ket{\sigma}
    \end{equation}
    is now an isometry associated with a post-selection onto a stabilizer state.  And since $(U \otimes G) \rho (U^\dagger \otimes G^\dagger ) \in \STAB$ by Definition~\ref{Definition:AbsolutelyStabilizerStates}, Eq.~\eqref{rotated_postselection} amounts to taking an $n$-qudit stabilizer state and post-selecting $n-k$ qudits onto a stabilizer state, which is guaranteed to result in a $k$-qudit stabilizer state.  That is, $U\rho'U^\dagger \in \STAB$.  And since $U$ was arbitrary, $\rho' \in \ASTAB$ for $k$ qudits.  In other words, $\ASTAB$ is closed under arbitrary subsystem post-selection and the two definitions of ``absolutely stabilizer'' coincide.
\end{proof}

Moreover, $\ASTAB$ is in fact closed under adaptive unitaries conditioned on the outcomes of projective measurements.  This follows immediately from Theorem~\ref{Theorem:ASTABPostselection}, as at each stage in a successive chain of measurements, modelled with a collection of post-selection isometries $\{V_i\}$, the resulting state is in $\ASTAB$ and so will remain so under a unitary.  Finally, being an intersection of convex sets, $\ASTAB$ is itself convex.  This implies that performing an unread destructive subsystem measurement, again modelled with a collection of isometries, the resulting average will also be absolutely stabilizer.
Hence, in summary, $\ASTAB$ can also be characterized as the set of states that cannot be transformed into a magic state using operations built from arbitrary unitaries, arbitrary destructive post-selections, discarding of subsystems (i.e.\ partial trace), and adaptive actions thereof.

\subsection{Non-triviality of ASTAB}

Before characterizing the set of absolutely stabilizer states, $\ASTAB$, we comment on a more general situation to illustrate why the problem is non-trivial. The shifted stabilizer polytope $\STAB_0$ lives in the generalized Bloch space $\Herm_0(\mathcal{H})$, which is a real Hilbert space with respect to the Hilbert-Schmidt inner product. Consider a subgroup $G$ of the associated orthogonal group, $G \subset \mathrm{O}(\Herm_0(\mathcal H))$, and denote its action on matrices as $g(\cdot)$. The generalization of  Eq.~\eqref{Equation:MuggleStates_Bloch} becomes
\begin{equation}\label{radial_function}
    \mathcal{O}_G := \bigcap\limits_{g\in G}g(\STAB_0).
\end{equation}
If this group action was transitive on the Hilbert-Schmidt unit sphere, then clearly the intersection $\mathcal O_G$ would be a Hilbert-Schmidt ball. In particular, it would be the inscribed ball of the stabilizer polytope. This happens for the single-qubit case because the adjoint representation of the unitary group contains all Hilbert-Schmidt rotations. But for $d^n > 2$ this is no longer true: unitary conjugations preserve the entire spectrum of a Hermitian operator while Hilbert-Schmidt rotations only preserve the Euclidean norm of the spectrum (i.e.~$\norm{A}_2 = \sqrt{\Tr(A^2)}$). More precisely, the set of unitary conjugations is a subgroup of the group of orthogonal transformations,
\begin{equation}
\begin{aligned}
    \mathrm{Ad}(\mathrm{SU}(d^n)) \subset \mathrm{O}(\Herm_0(\mathcal{H})) \simeq \mathrm{O}(d^{2n}-1),
\end{aligned}
\end{equation}
where $G=\mathrm{Ad}(\mathrm{SU}(d^n))$ is the image of the adjoint representation of the special unitary group $\mathrm{SU}(d^n)$ over $\Herm_0(\mathcal H) \simeq \mathfrak{su}(d^n)$, seen as the Lie algebra $\mathfrak{su}(d^n)$.\footnote{We are using the physicist's version, $i \, \mathfrak{su}(d^n)$, without loss of generality.} Hence the action of the unitary group by conjugation is not transitive on the Hilbert-Schmidt sphere, and so we cannot expect $\mathcal O_{\mathrm{Ad}(\mathrm{SU}(d^n))}$ to be a rescaled version thereof. In fact, since we are focusing on a subgroup of the rotation group we expect the muggle set to be a superset of the stabilizer inball:
\begin{equation}
    \mathcal{O}_{\mathrm{O}(\Herm_0(\mathcal{H}))} \subset \mathcal{O}_{\mathrm{Ad}(\mathrm{SU}(d^n))} = \ASTAB_0.
\end{equation}

A useful (but slightly dangerous) visual analogy for seeing how lower symmetry implies a larger invariant set can be made for a single qubit by considering the intersection of all rotations of $\STAB_0$ (Fig.~\ref{Figure:OneQubit}) versus just the subgroup of $z$-axis rotations (Fig.~\ref{Figure:cone_in_bloch}).  The latter is clearly a superset of the former.

\begin{figure}
    \centering
    \includegraphics[width=0.9\linewidth]{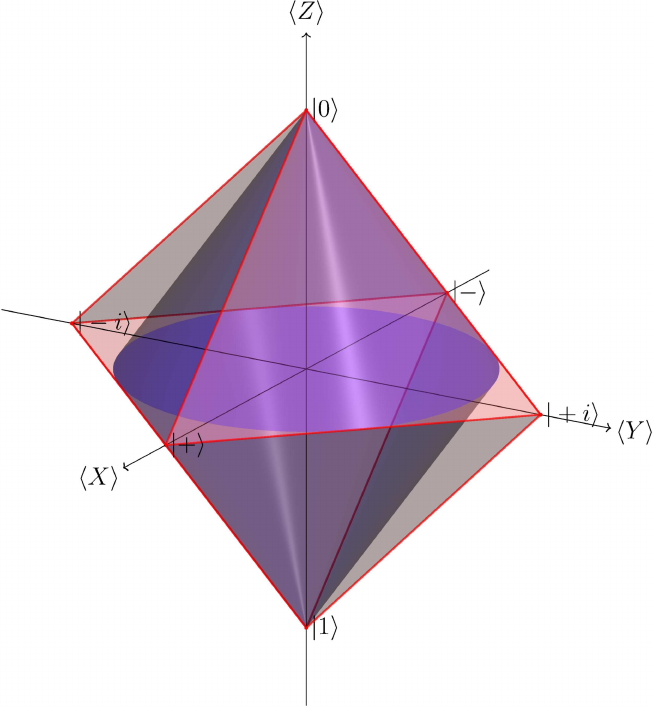}
    \caption{Intersection of all single-qubit stabilizer polytopes rotated about the $z$-axis to form a bicone within the Bloch sphere.  This bicone is a strict superset  of the stabilizer polytope's inball (see Fig.~\ref{Figure:OneQubit}) and is used here as a low-dimensional example to illustrate that a smaller group in Eq.~\eqref{radial_function} produces a larger invariant set.}
    \label{Figure:cone_in_bloch}
\end{figure}

\subsection{Spectral characterization of absolutely stabilizer states}

The set of absolutely stabilizer states is by definition unitarily-invariant. Membership of $\ASTAB$ is therefore determined only by the spectrum of a given state.  In this section we show that $\ASTAB$ can be characterized as a polytope in the space of state spectra (i.e.\ in the simplex of probability distributions). We call such convex sets \emph{absolutely stabilizer spectral polytopes}, or \emph{muggle polytopes}. This characterization requires the following lemma relating the facets of the stabilizer polytope to the vertices of the $\Lambda$ polytope.

\begin{Lemma}\label{Lemma:STABHRepresentation}
    For any $X\in\Herm_1(\mathcal{H})$, $X\in\STAB$ if and only if $\Tr(XY)\ge0$ for all vertices $Y\in\mathrm{vert}(\Lambda)$.
\end{Lemma}

\begin{proof}[Proof of Lemma~\ref{Lemma:STABHRepresentation}]
    For $X,Y\in\Herm_1(\mathcal{H})$, we write
    \begin{equation}
        X=\frac{1}{d^n}\mathbb{I}+X_0 \quad\text{and}\quad Y=\frac{1}{d^n}\mathbb{I}+Y_0,
    \end{equation}
    where $X_0,Y_0\in\Herm_0(\mathcal{H})$. Then the Hilbert-Schmidt inner product of $X$ and $Y$ is
    \begin{equation}
        \Tr(XY)=\frac{1}{d^n}+\Tr(X_0Y_0).
    \end{equation}
    Therefore, we have
    \begin{equation}
    \begin{aligned}
        &\{X\in\Herm_1(\mathcal{H})\;|\;\Tr(XY)\ge0\;\forall Y\in\mathrm{vert}(\Lambda)\}\\
        =& \frac{1}{d^n}\mathbb{I} + \left\{X_0\in\Herm_0(\mathcal{H}) \bigg| \begin{array}{r} -d^n\Tr(X_0Y_0)\le1 \\ \forall \, Y_0\in\mathrm{vert}(\Lambda_0)\end{array}\right\}.\label{Equation:STABDoubleDual}
    \end{aligned}
    \end{equation}
    That this set is equal to $\STAB$ follows from the polar duality between $\Lambda_0$ and $-d^n\STAB_0$, or equivalently from the polar duality between $\STAB_0$ and $-d^n\Lambda_0$. More explicitly, the second term on the right hand side of Eq.~\eqref{Equation:STABDoubleDual} is the polar dual of the polar dual of $\STAB_0$, and so with the fact that $\STAB_0$ contains the origin (since $\STAB$ contains the maximally mixed state), by a standard result in polytope theory this is equal to $\STAB_0$~\cite[Theorem~2.11]{Ziegler1995}.
\end{proof}

We now arrive at the first result of this section: a characterization of the spectra of absolutely stabilizer states.

\begin{Theorem}\label{Theorem:FirstSpectralCharacterization}
    For any number of qudits $n$ of any prime dimension $d$, an $n$-qudit state $\rho$ is absolutely stabilizer if and only if
    \begin{equation}\label{Equation:SpectralMuggleInequalities}
        \sum_{k=1}^{d^n}\lambda_k^{\uparrow}(\rho)\lambda_k^{\downarrow}(A)\ge0\qquad\forall A\in\mathrm{vert}(\Lambda),
    \end{equation}
    where $\lambda^\uparrow(X)$ (resp.\ $\lambda^\downarrow(X)$) denotes the vector of eigenvalues of the operator $X$ in non-decreasing (resp.\ non-increasing) order.
\end{Theorem}

Intuitively, each vertex of $\Lambda$ is associated with a facet of $\STAB$, and Eq.~\eqref{Equation:SpectralMuggleInequalities} forces the entire unitary group orbit of $\rho$ to be under the facet.
The proof relies on the so-called Ky Fan Minimum/Maximum Principle~\cite{KyFan1949,KyFan1950} via the following corollary thereof.

\begin{Corollary}[Ref.~\cite{MarshallArnold2011}, Proposition~20.A.2.a]\label{Corollary:KyFan}
    If $G$ is an $n\times n$ Hermitian matrix with eigenvalues $\lambda_1\ge\cdots\ge\lambda_n$, and $H$ is an $n\times n$ Hermitian matrix with eigenvalues $\mu_1\ge\cdots\ge\mu_n$, then
    \begin{align}
        \max_{U\in\mathrm{U}(n)}\Tr(UGU^\dagger H) =& \sum\limits_{i=1}^n\lambda_i\mu_i,\quad\text{and}\\
        \min_{U\in\mathrm{U}(n)}\Tr(UGU^\dagger H) =& \sum\limits_{i=1}^n\lambda_i\mu_{n-i+1}.
    \end{align}
\end{Corollary}

\begin{proof}[Proof of Theorem~\ref{Theorem:FirstSpectralCharacterization}]
    Using Lemma~\ref{Lemma:STABHRepresentation} together with the definition of $\ASTAB$, for a state $\rho$ we know that $\rho \in \ASTAB$ if and only if
    \begin{equation}\label{Equation:MuggleHalfspaces}
        \Tr(U\rho U^\dagger A)\ge0,\quad\forall A\in\mathrm{vert}(\Lambda),\;\forall U\in\mathrm{U}(d^n).
    \end{equation}
    For any given vertex $A$ of $\Lambda$, the condition Eq~\eqref{Equation:MuggleHalfspaces} must hold for all unitaries $U\in\mathrm{U}(d^n)$. This is equivalent to the condition,
    \begin{equation}\label{Equation:MuggleHalfspaces2}
        \min_{U\in\mathrm{U}(d^n)}\Tr(U\rho U^\dagger A)\ge0,\quad\forall A\in\mathrm{vert}(\Lambda).
    \end{equation}
    Then the theorem follows from Corollary~\ref{Corollary:KyFan} with $G=\rho$ a candidate absolutely stabilizer state and $H=A$ a vertex of $\Lambda$.
\end{proof}

The conclusion of Theorem~\ref{Theorem:FirstSpectralCharacterization} is that the set of spectra of absolutely stabilizer states is defined by a finite number of linear inequalities, and thus forms a polyhedral set. This set is bounded since it is contained in the simplex of non-negative and normalized spectra of quantum states. With these two facts, the Minkowski-Weyl theorem (also called the fundamental theorem of polytopes~\cite{Ziegler1995, Grunbaum2003}) implies that it is a polytope with a finite number of vertices. The set of all absolutely stabilizer states is then the unitary group orbit of the diagonal matrices corresponding to these allowed spectra. This characterization however contains some redundancy. Namely, not all the vertices of $\Lambda$ are needed. To determine which vertices of $\Lambda$ contribute, we can use the theory of majorization~\cite{MarshallArnold2011}. We make the following definition.

\begin{Definition}\label{Definition:SpectrallyGeneratingVertices}
    A subset $\mathcal{V}_{\mathrm{sGen}}\subseteq \mathrm{vert}(\Lambda)$ is \emph{spectrally generating} if for every vertex $A\in\mathrm{vert}(\Lambda)$ there exists a probability distribution
    $p(B\,|\,A)$ on $\mathcal{V}_{\mathrm{sGen}}$ such that
    \begin{equation}\label{Equation:SpectralGeneratingCondition}
        \lambda^\downarrow(A)\prec\sum_{B\in\mathcal{V}_{\mathrm{sGen}}}p(B\,|\,A)\,\lambda^\downarrow(B).
    \end{equation}
\end{Definition}

That is, a subset of vertices is here said to be spectrally generating if the spectrum of any other vertex is majorized by a convex combination of the spectra of the given subset.  To determine the allowed spectra of absolutely stabilizer states, we need only the constraints of the form Eq.~\eqref{Equation:SpectralMuggleInequalities} from a spectrally generating set of vertices.

\begin{Theorem}\label{Theorem:MainSpectralCharacterization}
    Let $\mathcal{V}_{\mathrm{sGen}}\subseteq\mathrm{vert}(\Lambda)$ be spectrally generating in the sense of Definition~\ref{Definition:SpectrallyGeneratingVertices}. Then a state $\rho$ is absolutely stabilizer if and only if
    \begin{equation}\label{Equation:ReducedMuggleInequalities}
        \sum_{k=1}^{d^n}\lambda_k^\uparrow(\rho)\lambda_k^\downarrow(A) \ge 0 \qquad \forall A\in\mathcal{V}_{\mathrm{sGen}}.
    \end{equation}
\end{Theorem}

\begin{proof}[Proof of Theorem~\ref{Theorem:MainSpectralCharacterization}]
    Fix a spectrally generating set $\mathcal{V}_{\mathrm{sGen}}\subseteq\mathrm{vert}(\Lambda)$. We want to show that the set of states defined by Eq.~\eqref{Equation:ReducedMuggleInequalities} is equal to the set of states defined by Eq.~\eqref{Equation:SpectralMuggleInequalities}. That Eq.~\eqref{Equation:SpectralMuggleInequalities} implies Eq.~\eqref{Equation:ReducedMuggleInequalities} for any state $\rho$ follows immediately from the inclusion $\mathcal{V}_{\mathrm{sGen}}\subseteq\mathrm{vert}(\Lambda)$.

    For the opposite direction, suppose a state $\rho$ satisfies Eq.~\eqref{Equation:ReducedMuggleInequalities}. Consider an arbitrary vertex $A$ of $\Lambda$. By assumption, there exists a probability distribution $p(B\,|\,A)$ on $\mathcal{V}_{\mathrm{sGen}}$ such that
    \begin{equation}
        \lambda^\downarrow(A) \prec \sum\limits_{B\in\mathcal{V}_{\mathrm{sGen}}}p(B\,|\,A)\lambda^\downarrow(B).
    \end{equation}
    Let $\mu^\downarrow=\sum_B p(B\,|\,A)\lambda^\downarrow(B)$. Since $\rho$ satisfies Eq.~\eqref{Equation:ReducedMuggleInequalities}, we have
    \begin{equation}
        \sum\limits_{k=1}^{d^n}\lambda_k^\uparrow(\rho)\mu_k^\downarrow=\sum\limits_{B\in\mathcal{V}_{\mathrm{sGen}}}p(B\,|\,A)\sum\limits_{k=1}^{d^n}\lambda_k^\uparrow(\rho)\lambda_k^\downarrow(B)\ge0.
    \end{equation}

    Then, using summation by parts,
    \begin{widetext}
    \begin{align}
        \sum\limits_{k=1}^{d^n}\lambda_k^\uparrow(\rho)\lambda_k^\downarrow(A) =& \lambda_{d^n}^\uparrow(\rho)\underbrace{\left(\sum\limits_{\ell=1}^{d^n}\lambda_\ell^\downarrow(A)\right)}_{=1}-\sum\limits_{k=1}^{d^n-1}\underbrace{\left(\sum\limits_{\ell=1}^k\lambda_\ell^\downarrow(A)\right)}_{\le\sum_{\ell=1}^k\mu_\ell^\downarrow}\underbrace{\left(\lambda_{k+1}^\uparrow(\rho)-\lambda_k^\uparrow(\rho)\right)}_{\ge0}\\
        \ge& \lambda_{d^n}^\uparrow(\rho)\left(\sum\limits_{\ell=1}^{d^n}\mu_\ell^\downarrow\right)-\sum\limits_{k=1}^{d^n-1}\left(\sum\limits_{\ell=1}^k\mu_\ell^\downarrow\right)\left(\lambda_{k+1}^\uparrow(\rho)-\lambda_k^\uparrow(\rho)\right) = \sum\limits_{k=1}^{d^n}\lambda_k^\uparrow(\rho)\mu_k^\downarrow\ge0.
    \end{align}
    \end{widetext} 
    Here the first line is the summation by parts formula, the second line follows from the assumption $\lambda^\downarrow(A)\prec\mu^\downarrow$ together with the relations indicated in the underbraces, and finally we apply the summation by parts formula in reverse. The last inequality was shown above. Therefore, for any state $\rho$, Eq.~\eqref{Equation:SpectralMuggleInequalities} follows from Eq.~\eqref{Equation:ReducedMuggleInequalities}.
\end{proof}

The last piece of the characterization of the set of absolutely stabilizer states is a characterization of a set of spectrally generating vertices $\mathcal{V}_{sGen}$. Since the spectra of the vertices of $\Lambda$ are different for qubits than for odd-prime-dimensional qudits, here we split into two cases.

\subsubsection{Qubits}

The spectrum of a CNC-type vertex of $\Lambda$ is determined by the parameter $m$ described in Theorem~\ref{Theorem:CNCClassification}.  This is described by the following lemma.

\begin{Lemma}\label{Lemma:CNCEigenvaluesQubits}
    For any number of qubits $n$ and any $m\in\{1,\dots,n\}$, every $n$-qubit CNC operator $A_\Omega^\gamma$ corresponding to a maximal CNC set $\Omega$ of type $m$ has eigenvalues
    \begin{align}
        \begin{cases}
            \frac{1+\sqrt{2m+1}}{2^m} & \text{with multiplicity $2^{m-1}$},\\
            \frac{1-\sqrt{2m+1}}{2^m} & \text{with multiplicity $2^{m-1}$},\\
            0 & \text{with multiplicity $2^n-2^m$}.
        \end{cases}
    \end{align}
\end{Lemma}

\begin{proof}[Proof of Lemma~\ref{Lemma:CNCEigenvaluesQubits}]
    A CNC operator $A_\Omega^\gamma$ of type $m$ has the form
    \begin{equation}
        A_\Omega^\gamma=g(A_{\tilde{\Omega}}^{\tilde{\gamma}}\otimes\ketbra{\sigma}{\sigma})g^\dagger,
    \end{equation}
    where $g$ is an $n$-qubit Clifford group element, $\ket{\sigma}$ is an $(n-m)$-qubit stabilizer state, and $A_{\tilde{\Omega}}^{\tilde{\gamma}}$ is a CNC operator such that all Pauli operators labeled by $\Omega\setminus\{0\}$ pair-wise anti-commute. Conjugation by unitary operators do not change the spectrum, so it suffices to characterize CNC operators of the form $A_{\tilde{\Omega}}^{\tilde{\gamma}}\otimes\ketbra{\sigma}{\sigma}$. The eigenvalues are all products of the eigenvalues of $A_{\tilde{\Omega}}^{\tilde{\gamma}}$ and the eigenvalues of $\ketbra{\sigma}{\sigma}$. But $\ketbra{\sigma}{\sigma}$ is a projector onto a pure $(n-m)$-qubit state, so it has eigenvalue $1$ with multiplicity $1$ and eigenvalue $0$ with multiplicity $2^{n-m}-1$. Thus, $A_{\tilde{\Omega}}^{\tilde{\gamma}}\otimes\ketbra{\sigma}{\sigma}$ has eigenvalue $0$ with multiplicity $2^m(2^{n-m}-1)$, and the remaining eigenvalues are exactly the eigenvalues of $A_{\tilde{\Omega}}^{\tilde{\gamma}}$.

    The operator $2^mA_{\tilde{\Omega}}^{\tilde{\gamma}}-\mathbb{I}$ is a sum of pair-wise anticommuting Pauli operators (the non-identity Pauli operators in $\tilde{\Omega}$). We have
    \begin{align}
        (2^mA_{\tilde{\Omega}}^{\tilde{\gamma}}-\mathbb{I})^2 =& \left(\sum\limits_{v\in\tilde{\Omega}\setminus\{0\}}(-1)^{\tilde{\gamma}(v)}T_v\right)^2\\
        =& \sum\limits_{v\in\tilde{\Omega}\setminus\{0\}}(-1)^{2\tilde{\gamma}(v)}T_v^2\\&+\sum\limits_{\substack{u,v\in\tilde{\Omega}\setminus\{0\}\\u\ne v}}(-1)^{\tilde{\gamma}(u)+\tilde{\gamma}(v)}T_uT_v\\
        =& |\tilde{\Omega}\setminus\{0\}|\mathbb{I}=(2m+1)\mathbb{I}.
    \end{align}
    The last term in the second line above vanishes as a result of the pair-wise anti-commutation of the operators in the set. Therefore, the eigenvalues of $2^mA_{\tilde{\Omega}}^{\tilde{\gamma}}-\mathbb{I}$ are in $\{-\sqrt{2m+1},\sqrt{2m+1}\}$. Since $\Tr(2^mA_{\tilde{\Omega}}^{\tilde{\gamma}}-\mathbb{I})=0$, each eigenvalue must appear with equal multiplicity. The result follows.
\end{proof}

Not all vertices of $\Lambda$ are of CNC type, but there is evidence that vertices of CNC type form a spectrally generating set in the sense of Definition~\ref{Definition:SpectrallyGeneratingVertices}. Therefore, we make the following conjecture.

\begin{Conjecture}\label{Conjecture:LambdaMajorization}
    For any number of qubits and any vertex $A\in\mathrm{vert}(\Lambda)$, the eigenvalues $\lambda^\downarrow(A)$ of $A$ are majorized by a convex combination of the eigenvalues $\lambda^\downarrow(A_\Omega^\gamma)$ of CNC-type vertices of $\Lambda$.
\end{Conjecture}

By directly enumerating the vertices of $\Lambda$ for one and two qubits, we confirm the conjecture for these cases. This is shown in Appendix~\ref{Section:ConjectureEvidence}, together with evidence for this conjecture for higher numbers of qubits. If the conjecture is true, we obtain an especially clean spectral characterization of the multi-qubit muggle polytope.

\begin{Theorem}\label{Theorem:QubitConjecturedCharacterization}
    If Conjecture~\ref{Conjecture:LambdaMajorization} is true, then a multiqubit state $\rho$ is absolutely stabilizer if and only if, for all $m\in\{1,\dots,n\}$,
    \begin{equation}
        \frac{1+\sqrt{2m+1}}{2^m}\sum\limits_{k=1}^{2^{m-1}}\lambda_k^\uparrow(\rho) + \frac{1-\sqrt{2m+1}}{2^m}\sum\limits_{k=1}^{2^{m-1}}\lambda_k^\downarrow(\rho)\ge0.
    \end{equation}
\end{Theorem}

\begin{proof}[Proof of Theorem~\ref{Theorem:QubitConjecturedCharacterization}]
    This is the immediate conclusion of Theorem~\ref{Theorem:MainSpectralCharacterization}, Conjecture~\ref{Conjecture:LambdaMajorization}, and Lemma~\ref{Lemma:CNCEigenvaluesQubits}.
\end{proof}

\subsubsection{Odd-prime-dimensional qudits}

There are two different classes of CNC operators for odd-prime-dimensional qudits, those with linear value assignments and those with nonlinear value assignments. The eigenvalues of the first class are described by the following lemma.

\begin{Lemma}\label{Lemma:PhasePointEigenvalues}
    For any number $n$ of qudits with odd-prime dimension $d$, the phase space point operators $A_u$ of the Wigner function have eigenvalues
    \begin{align}
        \begin{cases}
            \;\;\,1 & \text{with multiplicity $\frac{1}{2}(d^n+1)$},\\
            -1 & \text{with multiplicity $\frac{1}{2}(d^n-1)$}.
        \end{cases}
    \end{align}
\end{Lemma}

\begin{proof}[Proof of Lemma~\ref{Lemma:PhasePointEigenvalues}]
    The phase space point operator $A_u, u\in\Z_d^{2n}$ can be obtained from $A_0$ by $A_u=T_uA_0T_u^\dagger$. Conjugation by unitary operators doesn't change the eigenvalues, so it suffices to characterize the eigenvalues of $A_0$. It can easily be shown using Eq.~\eqref{Equation:ParityOperator} that the parity operator $A_0$ satisfies $A_0^2=\mathbb{I}$. Therefore, the eigenvalues of $A_u$ are in $\{-1,1\}$. Since $\Tr(A_u)=1$, the multiplicity of eigenvalue $+1$ must be one more than the multiplicity $-1$. The result follows.
\end{proof}

Through Eq.~\eqref{Equation:SpectralMuggleInequalities}, this gives a necessary condition for an $n$-qudit state $\rho$ to be an absolutely stabilizer state, namely,
\begin{equation}
    \sum\limits_{k=1}^{(d^n+1)/2}\lambda_k^\uparrow(\rho) - \sum\limits_{k=1}^{(d^n-1)/2}\lambda_k^\downarrow(\rho)\ge0.
\end{equation}

Since the noncontextuality condition Eq.~\eqref{Equation:NoncontextualityCondition} enforces linearity only on \emph{commuting} triples $u,v,u+v\in\Omega$, at least for the case $n=1$ and odd-prime $d$, there are CNC operators which are vertices of $\Lambda$ but are not phase space point operators~\cite{OkayRaussendorf2017,ZurelHeimendahl2024,ZurelHeimendahl2024b}. This can be seen already for a single qutrit. There, there are nine phase space point operators, but the $\Lambda$ polytope has $81$ vertices. The remaining $72$ vertices are CNC operators with $\Omega=\Z_3^2$ and with non-linear value assignments $\gamma$~\cite{DelfosseRaussendorf2017}. 

We could make the same conjecture about CNC operators being a spectrally generating set, Conjecture~\ref{Conjecture:LambdaMajorization}, for odd-prime-dimensional qudits, and this would give us a similarly conditional characterization of the muggle set as Theorem~\ref{Theorem:QubitConjecturedCharacterization}. However, a complete classification of the spectra of multiqudit CNC operators is currently not known. Even for a single-qudit, there are many possible spectra of $\Lambda$ polytope vertices. In this case, all vertices of $\Lambda$ are CNC type, but there are many unitary group orbits of CNC-type operators. They are enumerated up to $d=7$ in Ref.~\cite{ApplebyChaturvedi2008}. They find two orbits (i.e., two distinct spectra) for $d=3$, $9$ for $d=5$, and $210$ for $d=7$. A complete characterization of CNC spectra on odd-prime-dimensional qudits is currently out of reach.

\subsection{Examples}

In this section, we illustrate the set of muggle states through three examples: (1)~one qubit, (2)~two qubits, and (3)~one qutrit.

\subsubsection{One qubit}

The one-qubit case is easy to visualize. The set of physical states is the Bloch ball, a Hilbert-Schmidt ball in the $3$-dimensional real affine space $\Herm_1(\mathbb{C}^2)$. There are six pure stabilizer states, two eigenstates for each of the $X$, $Y$, and $Z$ Pauli operators. These are the vertices of $\STAB \subset \Herm_1(\mathbb C^2)$, which forms a regular octahedron.

On a single-qubit, $\Lambda$ is a cube with vertices $A_{\pm\pm\pm}=\frac{1}{2}(I\pm X\pm Y\pm Z)$. These vertices are all unitarily equivalent with eigenvalues $(1\pm\sqrt{3})/2$. The set of absolutely stabilizer states consists of all states with eigenvalues $(1-\lambda,\lambda)$, where without loss of generality we assume $1-\lambda\ge\lambda$, satisfying
\begin{equation}
    \begin{cases}
        \lambda,1-\lambda\ge0\\
        \frac{1+\sqrt{3}}{2}\lambda+\frac{1-\sqrt{3}}{2}(1-\lambda)\ge0.
    \end{cases}
\end{equation}
These constraints simplify to
\begin{equation}
    \frac{1}{2}-\frac{1}{2\sqrt{3}}\le\lambda\le\frac{1}{2},
\end{equation}
corresponding to diagonal matrices $(\mathbb{I}+zZ)/2$ with $0\le z\le 1/\sqrt{3}$.
The set of all absolutely stabilizer states is the unitary orbit of these diagonal states. These are states of the form $\rho(x,y,z)=\frac{1}{2}(\mathbb{I}+xX+yY+zZ)$ with $x^2+y^2+z^2\le\frac{1}{3}$. This is simply a three-dimensional Hilbert-Schmidt ball in $\Herm_0(\mathbb{C}^2)$, shifted to be centred at the maximally mixed state and inscribed in the stabilizer octahedron, the largest ball inscribed in the stabilizer polytope. This is illustrated in Figure~\ref{Figure:OneQubit}.

\begin{figure}
    \centering
    \includegraphics[width=\linewidth]{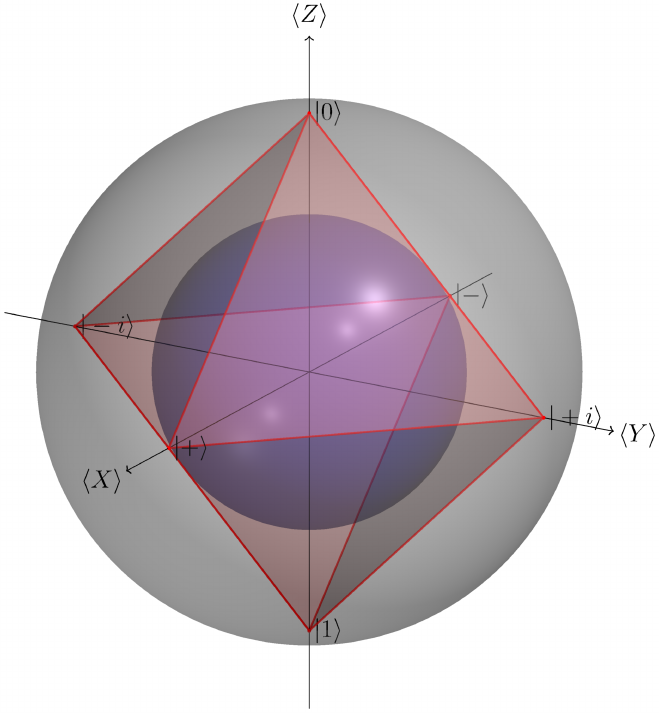}
    \caption{\label{Figure:OneQubit}The single-qubit Bloch ball is shown in gray containing the stabilizer octahedron in red and the set of single-qubit absolutely stabilizer states in blue. In the single-qubit case, the set of absolutely stabilizer states is a Hilbert-Schmidt ball inscribed in the stabilizer octahedron.}
\end{figure}

\subsubsection{Two qubits}

The two-qubit stabilizer polytope is a $15$-dimensional polytope with $60$ vertices, corresponding to the $60$ pure two-qubit stabilizer states. $\Lambda$ has $22{,}320$ vertices which fall into $8$ orbits under the action of the Clifford group with $8$ distinct spectra. Conjecture~\ref{Conjecture:LambdaMajorization} holds in this case (this is shown in Appendix~\ref{Section:ConjectureEvidence}), and so by Theorem~\ref{Theorem:QubitConjecturedCharacterization} we only need the constraints from the spectra of the two unitary group orbits of CNC-type vertices. The spectra are
\begin{equation}
    \lambda^\downarrow(A_\Omega^\gamma)=\begin{cases}
        \left(\frac{1+\sqrt{3}}{2},0,0,\frac{1-\sqrt{3}}{2}\right) & \text{for $m=1$},\\
        \left(\frac{1+\sqrt{5}}{4},\frac{1+\sqrt{5}}{4},\frac{1-\sqrt{5}}{4},\frac{1-\sqrt{5}}{4}\right) & \text{for $m=2$.}
    \end{cases}
\end{equation}

Suppose the eigenvalues of a candidate absolutely stabilizer state $\rho$ are $\lambda_1,\lambda_2,\lambda_3,\lambda_4$. Without loss of generality, we can assume the eigenvalues satisfy
\begin{equation}
    \begin{cases}
        \lambda_1\ge\lambda_2\ge\lambda_3\ge\lambda_4\ge0\\
        \lambda_1+\lambda_2+\lambda_3+\lambda_4=1.
    \end{cases}
\end{equation}
These constraints define the Weyl chamber: the subpolytope in the of the three dimensional unit simplex of nonnegative spectra corresponding to spectra sorted in non-increasing order. Then, to get the allowed (sorted) spectra of absolutely stabilizer states, we impose the two additional constraints
\begin{equation}
    \begin{cases}
        \frac{1-\sqrt{3}}{2}\lambda_1 + \frac{1+\sqrt{3}}{2}\lambda_4 \ge 0\\
        \frac{1-\sqrt{5}}{4}(\lambda_1+\lambda_2) + \frac{1+\sqrt{5}}{4}(\lambda_3+\lambda_4) \ge 0.
    \end{cases}
\end{equation}
These constraints give us a three-dimensional polytope of allowed spectra with vertices.
\begin{equation}
    \begin{cases}
        \left(\frac{1}{4}, \frac{1}{4}, \frac{1}{4}, \frac{1}{4}\right)\\
        \left(\frac{5+\sqrt{3}}{22}, \frac{5+\sqrt{3}}{22}, \frac{5+\sqrt{3}}{22}, \frac{7-3\sqrt{3}}{22}\right)\\
        \left(\frac{7+3\sqrt{3}}{22}, \frac{5-\sqrt{3}}{22}, \frac{5-\sqrt{3}}{22}, \frac{5-\sqrt{3}}{22}\right)\\
        \left(\frac{5+\sqrt{5}}{20}, \frac{5+\sqrt{5}}{20}, \frac{5-\sqrt{5}}{20}, \frac{5-\sqrt{5}}{20}\right)\\
        \left(\frac{5+\sqrt{5}}{20}, \frac{5+\sqrt{5}}{20}, \frac{5\sqrt{3}-4\sqrt{5}+\sqrt{15}}{20}, \frac{10-5\sqrt{3}+2\sqrt{5}-\sqrt{15}}{20}\right)\\
        \left(\frac{10+5\sqrt{3}-2\sqrt{5}-\sqrt{15}}{20}, \frac{-5\sqrt{3}+4\sqrt{5}+\sqrt{15}}{20}, \frac{5-\sqrt{5}}{20}, \frac{5-\sqrt{5}}{20}\right)
    \end{cases}
\end{equation}
The set of all unordered spectra is obtained by taking the convex hull of these $6$ vertices together with all of their coordinate permutations. The result is a three-dimensional polytope of spectra with $32$ vertices and $18$ facets. The set of all absolutely stabilizer states is the set of all unitary conjugations of the diagonal matrices corresponding to these allowed spectra. A cross section of the $2$-qubit state space, together with the stabilizer polytope and muggle set, is shown in Figure~\ref{Figure:TwoQubits}.

\begin{figure}
    \centering
    \includegraphics[width=\linewidth]{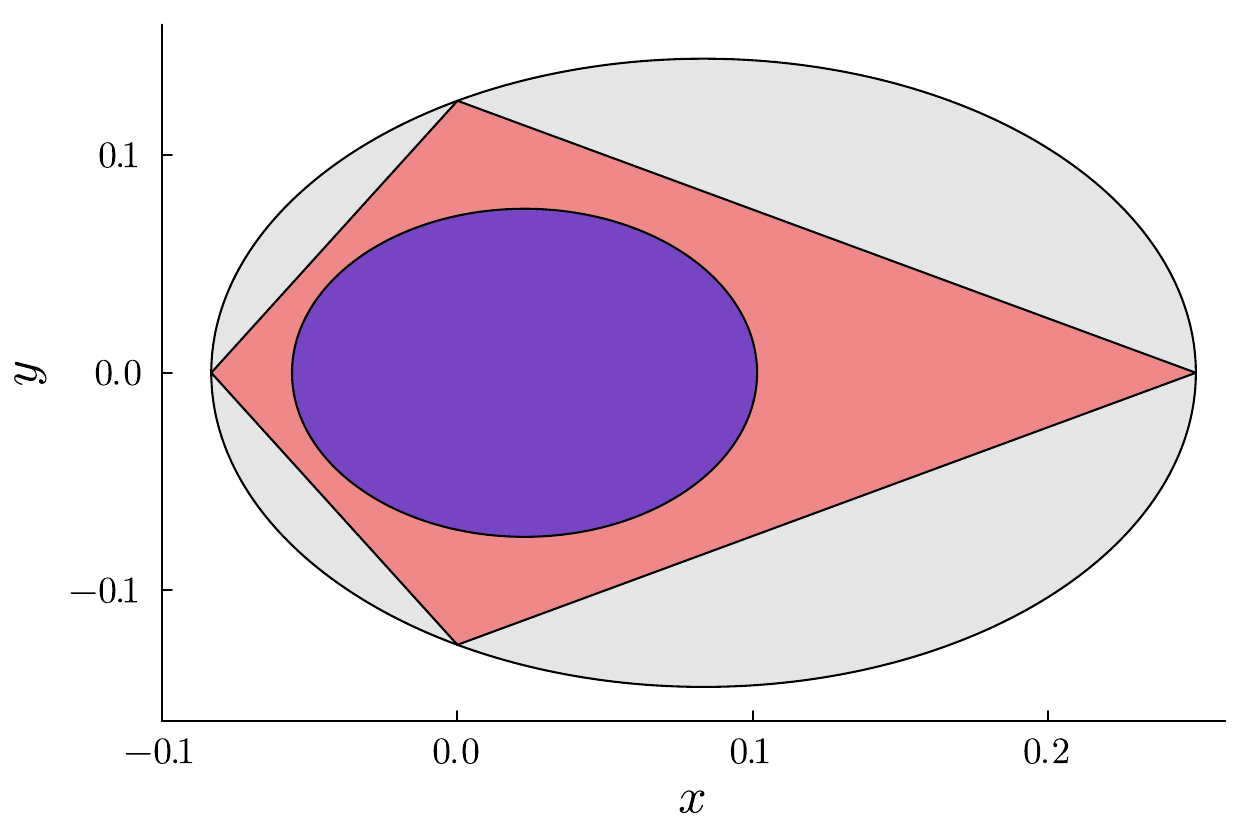}
    \caption{\label{Figure:TwoQubits}Cross section of the set of $2$-qubit states parametrized by $\rho(x,y)=\frac{1}{4}I+x(ZZ+XX-YY)+y(Z_1+Z_2)$. The set of physical density matrices is shown in gray, $\STAB_0$ in red, and $\ASTAB_0$ in blue.}
\end{figure}

\subsubsection{One qutrit}

The single-qutrit stabilizer polytope has $12$ vertices and $81$ facets. The facets, which correspond to the $81$ vertices of $\Lambda$, are all CNC-type, but they fall into two orbits under the action of the unitary group. $9$ of them correspond to phase space point operators of the Wigner function (CNC operators $A_{\Z_d^2}^\gamma$ with linear value assignment $\gamma$), and the other $72$ correspond to CNC operators with non-linear value assignments~\cite{DelfosseRaussendorf2017}. The eigenvalues of these operators are
\begin{equation}
    \lambda^\downarrow\left(A_{\mathbb{Z}_3^2}^\gamma\right)=
    \begin{cases}
        (1,1,-1) & \text{for linear $\gamma$,}\\
        \left(\frac{1+\sqrt{5}}{2},0,\frac{1-\sqrt{5}}{2}\right) & \text{for nonlinear $\gamma$.}
    \end{cases}
\end{equation}

Let $\lambda_1,\lambda_2,\lambda_3$ with $\lambda_1\ge\lambda_2\ge\lambda_3\ge0$ be the sorted eigenvalues of a candidate single-qutrit muggle state $\rho$. With non-negativity and normalization, $\lambda_1+\lambda_2+\lambda_3=1$, together with the constraints
\begin{equation}\label{Equation:OneQutritMuggleInequalities}
    \begin{cases}
        -\lambda_1+\lambda_2+\lambda_3\ge0,\\
        \frac{1-\sqrt{5}}{2}\lambda_1+\frac{1+\sqrt{5}}{2}\lambda_3\ge0,
    \end{cases}
\end{equation}
from Theorem~\ref{Theorem:MainSpectralCharacterization}, we get a two-dimensional polygon of allowed ordered spectra for absolutely stabilizer states living in the Weyl chamber. Allowing unsorted spectra (i.e., coordinate permutations of the spectra in the Weyl chamber), we get the full polygon of allowed spectra living in the simplex of normalized nonnegative spectra. These absolutely stabilizer state spectra are shown in Figure~\ref{Figure:OneQutrit}. For comparison, Figure~\ref{Figure:OneQutrit} also shows the allowed spectra of absolutely Wigner positive states~\cite{SalazarJunior2022} described in the next section.

\begin{figure*}
    \centering
    \includegraphics[width=0.375\linewidth]{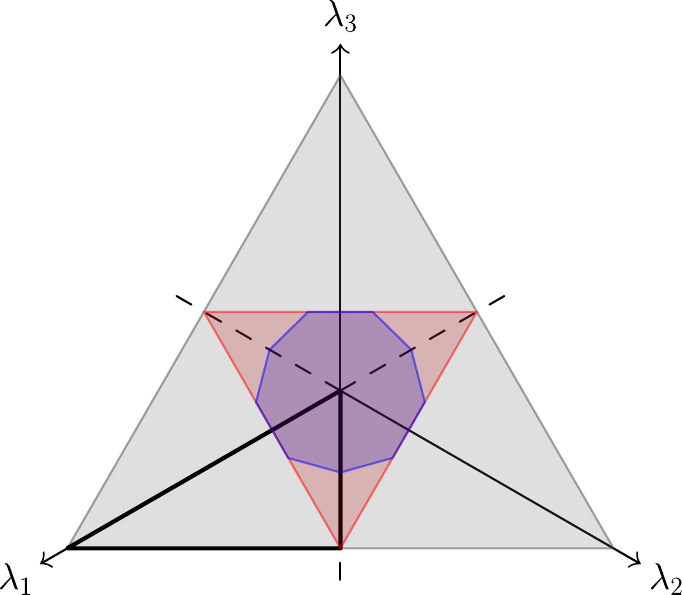}
    \hspace{0.05\linewidth}
    \includegraphics[width=0.525\linewidth]{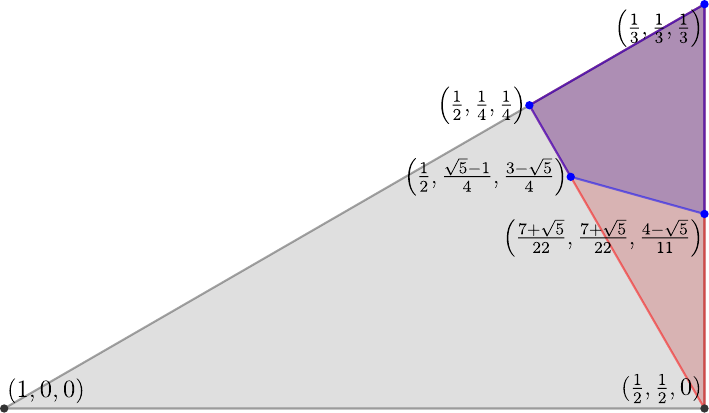}
    \caption{\label{Figure:OneQutrit} Left: Ternary plot of the spectra of single-qutrit absolutely stabilizer states (blue) and absolutely Wigner-positive states (red).  Right: the associated Weyl chamber (i.e.~the subset of ordered spectra) with special points of intersection displayed.}
\end{figure*}

\section{Absolute Wigner positivity}\label{Section:AWP}

On odd-prime-dimensional qudits, it is known that there are non-stabilizer states that cannot be distilled to useful magic states, and on which stabilizer circuits can be efficiently simulated classically. These so-called bound magic states are non-stabilizer states with positive Wigner function ~\cite{VeitchEmerson2012}. Measures of negativity in the Wigner function serve as monotones with respect to the resource theory of magic~\cite{VeitchEmerson2014}. Therefore, with a similar motivation as absolutely stabilizer states, we can define absolutely Wigner-positive states as those for which the Wigner function remains nonnegative after conjugation by any unitary group element~\cite{SalazarJunior2022}.  In this section, we get a complete spectral characterization of the set of AWP states using similar methods as those we used to characterize absolutely stabilizer states in Section~\ref{Section:MuggleMain}.

\begin{Definition}
    An $n$-qudit state $\rho$ is absolutely Wigner-positive (or AWP) if the Wigner function of $U\rho U^\dagger$ is nonnegative for all unitary operators $U\in\mathrm{U}(d^n)$.
\end{Definition}

Similar to the description of the set of absolutely stabilizer states in Eq.~\eqref{Equation:MuggleStates}, the set of absolutely Wigner-positive states, denoted $\AWP$, is
\begin{equation}
    \mathrm{AWP} = \bigcap\limits_{U\in\mathrm{U}(d^n)} U\, \WP \, U^\dagger.
\end{equation}
This set is also unitarily-invariant and thus has a spectral characterization\footnote{Note that due to the size of the Wigner polytope, the set $\AWP$ as defined could \textit{a priori} include non-positive-semi-definite operators. In fact, this is not the case: every operator in $\AWP$ has nonnegative eigenvalues. This follows immediately from the normalization of the spectra of operators in $\Herm_1(d^n)$ together with Eq.~\eqref{Equation:AWPSpectralCharacterization}.}. With the following theorem, we get a complete characterization of the spectra of AWP states.

\begin{Theorem}\label{Theorem:SpectralAWPCharacterization}
    For any number of qudits $n$ with any odd-prime dimension $d$, an $n$-qudit state $\rho$ is absolutely Wigner-positive if and only if
    \begin{equation}\label{Equation:AWPSpectralCharacterization}
        \sum\limits_{k=1}^{(d^n-1)/2}\lambda_k^\downarrow(\rho)\le\frac{1}{2},
    \end{equation}
    i.e., the sum of the largest $(d^n-1)/2$ eigenvalues of $\rho$ is no more than $1/2$. Equivalently, a state $\rho$ is absolutely Wigner-positive if and only if its spectrum is in the $(d^n-1)$-dimensional polytope formed by the convex hull of the points
    \begin{equation}\label{Equation:AWPSpectralVertex}
        \begin{cases}
            \bigg(\underbrace{\frac{1}{d^n-1},\dots,\frac{1}{d^n-1}}_{d^n-1\text{ times}},0\bigg)\\
            \bigg(\frac{2}{d^n+1},\underbrace{\frac{1}{d^n+1},\dots,\frac{1}{d^n+1}}_{d^n-1\text{ times}}\bigg)
        \end{cases}
    \end{equation}
    and all their coordinate permutations.
\end{Theorem}

\begin{proof}[Proof of Theorem~\ref{Theorem:SpectralAWPCharacterization}]
    With the description of the Wigner polytope Eq.~\eqref{Equation:WignerPolytope2}, state $\rho$ is AWP if and only if
    \begin{equation}
        \Tr(U\rho U^\dagger A_u)\ge0,\quad\forall u\in\Z_d^{2n},\;\forall U\in\mathrm{U}(d^n),
    \end{equation}
    or equivalently,
    \begin{equation}
        \min_{U\in\mathrm{U}(d^n)}\Tr(U\rho U^\dagger A_u)\ge0,\quad\forall u\in\Z_d^{2n}.
    \end{equation}
    Using Corollary~\ref{Corollary:KyFan}, together with Lemma~\ref{Lemma:PhasePointEigenvalues} characterizing the eigenvalues of phase space point operators, we get the constraint that a state $\rho$ is AWP if and only if
    \begin{equation}\label{Equation:AWPSpectralConstraint}
        \sum\limits_{k=1}^{(d^n+1)/2}\lambda_k^\uparrow(\rho)-\sum\limits_{k=1}^{(d^n-1)/2}\lambda_k^\downarrow(\rho)\ge0.
    \end{equation}
    With normalization,
    \begin{equation}
        \sum\limits_{k=1}^{d^n}\lambda_k^\downarrow(\rho)=\sum\limits_{k=1}^{(d^n-1)/2}\lambda_k^\downarrow(\rho) + \sum\limits_{k=1}^{(d^n+1)/2}\lambda_k^\uparrow(\rho)=1,
    \end{equation}
    Eq.~\eqref{Equation:AWPSpectralConstraint} is equivalent to Eq.~\eqref{Equation:AWPSpectralCharacterization}. This gives us the first statement of the theorem.

    For the second statement of the theorem, let $\lambda_1,\dots,\lambda_{d^n}$ be the eigenvalues of an $n$-qudit candidate AWP state $\rho$. Without loss of generality, we can assume
    \begin{equation}\label{Equation:AWPWeylChamber}
        \begin{cases}
            \lambda_1\ge\lambda_2\ge\cdots\ge\lambda_{d^n}\ge0,\\
            \lambda_1+\cdots+\lambda_{d^n}=1.
        \end{cases}
    \end{equation}
    This gives us a Weyl chamber in the simplex of allowed spectra of physical quantum states. This Weyl chamber is a $d^n-1$-dimensional simplex in coordinates $(\lambda_1,\dots,\lambda_{d^n})$ defined by the constraints Eq.~\eqref{Equation:AWPWeylChamber} and with vertices
    \begin{equation}\label{Equation:AWPWeylChamberVertices}
        \begin{cases}
            v_1=(1,0,\dots,0),\\
            v_2=(\frac{1}{2},\frac{1}{2},0,\dots,0)\\
            \vdots\\
            v_{d^n}=(\frac{1}{d^n},\dots,\frac{1}{d^n}).
        \end{cases}
    \end{equation}

    The allowed spectra of AWP states in this Weyl chamber are obtained by imposing the additional inequality Eq.~\eqref{Equation:AWPSpectralCharacterization}. Then, through one step of the double-description method for polytope vertex enumeration~\cite{FukudaProdon1996}, we get the allowed AWP spectra. We will find the vertices of the subpolytope of the Weyl chamber obtained by imposing Eq.~\eqref{Equation:AWPSpectralCharacterization}, and then show that after allowing unordered spectra (coordinate permutations of the spectra in the Weyl chamber), most of these vertices are redundant.

    According to the double description method, the vertices of the new polytope obtained after imposing Eq.~\eqref{Equation:AWPSpectralCharacterization} on the Weyl chamber simplex have vertices consisting of (1)~the Weyl chamber vertices which are consistent with this inequality, and (2)~the intersection points of the hyperplane
    \begin{equation}\label{Equation:AWPHyperplane}
        H:=\left\{\lambda\in\mathbb{R}^{d^n}\;\bigg|\;\sum_{i=1}^{(d^n-1)/2}\lambda_i=\frac{1}{2}\right\}
    \end{equation}
    with edges connecting a vertex that violates Eq.~\eqref{Equation:AWPSpectralCharacterization} to a vertex that satisfies it.

    The only vertices of the Weyl chamber that are compatible with Eq.~\eqref{Equation:AWPSpectralCharacterization} are $v_{d^n}$ (the spectrum of the maximally mixed state) and $v_{d^n-1}$. $v_{d^n-1}$ lies on the hyperplane, so there are no new vertices obtained by intersecting $H$ with line segments of the form $[v_{d^n-1},v_k]$ for $k=1,\dots,d^n-2$.

    Since the Weyl chamber is a simplex, all pairs of vertices Eq.~\eqref{Equation:AWPWeylChamberVertices} are connected by edges. Therefore, we need to check for intersections of $H$ with all line segments connecting $v_{d^n}$ to all of $v_1,\dots,v_{d^n-2}$. We now compute the intersections $v_k':=H\cap[v_{d^n},v_k]$ explicitly. Parameterize the edge $[v_k,v_{d^n}]$ by $\nu_k(t)=(1-t)v_{d^n}+t v_k$, $t\in[0,1]$. Then
    \begin{equation}\label{Equation:EdgeParam}
        [\nu_k(t)]_i=
        \begin{cases}
            \frac{1}{d^n}+t\left(\frac{1}{k}-\frac{1}{d^n}\right), & 1\le i\le k,\\
            \frac{1-t}{d^n}, & k<i\le d^n.
        \end{cases}
    \end{equation}
    Then we enforce the linear equation defining the hyperplane $H$, $\sum_{i=1}^{(d^n-1)/2}[\nu_k(t)]_i=\frac{1}{2}$, and solve for $t$ to find the intersection point $v_k':=H\cap[v_{d^n},v_k]$. Here we have two cases. For $k=1,\dots,(d^n-1)/2-1$, we have
    \begin{equation}\label{Equation:AWPSpectrumVertexCase1}
        v_k'=\bigg(\underbrace{\frac{k+1}{k(d^n+1)},\dots}_{k\text{ times}},\underbrace{\frac{1}{d^n+1},\dots}_{d^n-k\text{ times}}\bigg),
    \end{equation}
    and for $k=(d^n-1)/2,\dots,d^n-2$, we have
    \begin{equation}\label{Equation:AWPSpectrumVertexCase2}
        v_k'=
        \bigg(\underbrace{\frac{1}{d^n-1},\dots}_{k\text{ times}},
        \underbrace{\frac{d^n-1-k}{(d^n-1)(d^n-k)},\dots}_{d^n-k\text{ times}}\bigg),
    \end{equation}

    To summarize, the vertices of the polytope of allowed ordered spectra of AWP states are $(v_1',\dots,v_{d^n-2}',v_{d^n-1},v_{d^n})$, where $v_k'$ takes the form Eq.~\eqref{Equation:AWPSpectrumVertexCase1} or Eq.~\eqref{Equation:AWPSpectrumVertexCase2} depending on whether $k<(d^n-1)/2$. The set of all (unordered) spectra of AWP states is the polytope formed from the convex hull of these vertices and all of their coordinate permutations. Now we show that after allowing coordinate permutations, all of these vertices except $v_1'$ and $v_{d^n-1}$ are redundant.

    For $k=1$, we have
    \begin{equation}\label{Equation:pDef}
        v_1':= \bigg(\frac{2}{d^n+1},\underbrace{\frac{1}{d^n+1},\dots,\frac{1}{d^n+1}}_{d^n-1\text{ times}}\bigg).
    \end{equation}
    Let $v_1'^{(j)}$ denote the coordinate permutation of $v_1'$ with the entry $2/(d^n+1)$ in position $j$. For $2\le k<(d^n-1)/2$, let $T$ be the set of indices of the $k$ larger components in Eq.~\eqref{Equation:AWPSpectrumVertexCase1}, so $|T|=k$. Then
    \begin{equation}\label{Equation:Case2ConvexComb}
        v_k'=\frac{1}{|T|}\sum_{j\in T} v_1'^{(j)}.
    \end{equation}
    Hence, every $v_k'$ with $2\le k<(d^n-1)/2$ lies in the convex hull of permutations of $v_1'$, and is redundant.
    
    Now, for $(d^n-1)/2\le k<d^n-1$, let $T:=\{k+1,\dots,d^n\}$ be the set of indices of the smaller entries in $v_k'$, so $|T|=d^n-k$. For each $j\in\{1,\dots,d^n\}$ let $v_{d^n-1}^{(j)}$ be the coordinate permutation of $v_{d^n-1}$ with a $0$ in the $j^{th}$ position. Then we can check that
    \begin{equation}\label{Equation:Case1ConvexComb}
        v_k'=\frac{1}{|T|}\sum_{j\in T}v_{d^n-1}^{(j)}.
    \end{equation}
    Therefore, every $v_k'$ with $(d^n-1)/2\le k\le d^n-2$ lies in the convex hull of the permutations of $v_{d^n-1}$, and so these vertices are also redundant.

    Finally, note that $v_{d^n}$ is also redundant after coordinate permutations, since
    \begin{equation}
        v_{d^n}=\frac{1}{d^n}\sum\limits_{j=1}^{d^n}v_{d^n-1}^{(j)}.
    \end{equation}

    Therefore, after allowing unordered spectra, all of the vertices $v_2',\dots,v_{d^n-2}'$ and the vertex $v_{d^n}$ are redundant, and so the set of AWP spectral polytope is the convex hull of $v_1'$ and $v_{d^n-1}$ and all of their coordinate permutations.
\end{proof}

For the case, $d=3$ and $n=1$, $v_1'$ is a convex mixture of $v_{d^n-1}$ and its coordinate permutations (in particular, $v_1'=(v_2^{(2)}+v_2^{(3)})/2$ using the notation from the proof above), so the polytope of AWP spectra is a simplex spanned by $v_2$ and its coordinate permutations, but this is not true in general. For all Hilbert space dimensions larger than $3$, $v_1'$ is not a mixture of $v_{d^n-1}$ and its coordinate permutations, so the AWP spectra are a convex hull of two simplexes, one generated by $v_1'$ and the other generated by $v_{d^n-1}$.

An illustration of the set of qutrit AWP spectra is shown in Figure~\ref{Figure:OneQutrit}.  Here we see that the AWP polytope (red) has vertices on spectra with only two non-vanishing eigenvalues equal to $\frac12$.  This implies that for qutrit states, any equal mixture of two orthogonal pure states is already absolutely Wigner-positive.  More generally, Theorem~\ref{Theorem:SpectralAWPCharacterization} says there exist AWP states with rank $d^n - 1$.

Relatedly, the muggle polytope is always strictly contained in the AWP polytope, which implies the existence of mixed states that are non-stabilizer (i.e.~magic) yet are absolutely Wigner-positive.  This can be seen as a strengthening of the notion of bound magic introduced in~\cite{VeitchEmerson2012}, where it was shown that the Wigner polytope was a superset of the stabilizer polytope.

\section{Purity bounds}\label{Section:StabilizerInradius}

Here we discuss the inscribed Hilbert-Schmidt ball of the stabilizer and Wigner polytopes, which amount to sufficient conditions for inclusion based on purity.  We also discuss the circumradius of the Wigner polytope, which is a necessary condition.  We then do a comparison to the balls of positive semi-definite matrices and absolutely separable states.

\subsection{Inradius of the stabilizer polytope}

In this section we find the inradius of the stabilizer polytope, i.e., the radius of the largest ball which is contained in the stabilizer polytope. By unitary invariance, this is also the largest ball contained in the set of absolutely stabilizer states. A state having Bloch radius less than this inradius thus acts as a tight sufficient condition for a state to be a mixture of stabilizer states, as well as being absolutely stabilizer. Conversely, this radius also characterizes how impure magic states can be.  That is, there exist non-stabilizer states with Bloch radius only infinitesimally larger than the inradius of the stabilizer polytope. Interestingly, we find a drastic difference between the inradius of the stabilizer polytope for even versus odd qudit dimension $d$.

\bigskip

For a polytope $P$, let $r(P)$ denote the inradius of $P$, i.e., the radius of the largest ball of the same dimension as $P$ which is contained in $P$. Similarly, let $R(P)$ denote the circumradius of $P$, the radius of the smallest ball containing $P$. For a general polytope, a ball contained in the polytope that achieves the inradius is not necessarily centred at the origin. Therefore, to start, we need the following lemma.

\begin{Lemma}\label{Lemma:StabRadiusOriginCentred}
    The stabilizer polytope contains a $d^{2n}-1$-dimensional ball of radius $r(\STAB_0)$ centred at the maximally mixed state.
\end{Lemma}

\begin{proof}[Proof of Lemma~\ref{Lemma:StabRadiusOriginCentred}]
    It suffices to show that $\STAB_0$ contains a $d^{2n}-1$-dimensional ball of radius $r(\STAB_0)$ which is centred at the origin in $\Herm_0(\mathcal{H})$, i.e., the zero operator $\bbzero$.
    
    Let $B(c,r)$ denote a ball of radius $r$ centred at $c$. By definition of $r(\STAB_0)$, $\STAB_0$ contains a ball of radius $r(\STAB_0)$. Let $B(A,r(\STAB_0))=A+B(\bbzero,r(\STAB_0))$ be such a ball. If $A=\bbzero$, then we are done. If $A\ne\bbzero$, then consider the action of the Clifford group on this ball. We get an orbit of balls,
    \begin{equation}
        \left\{g(A+B(\bbzero,r(\STAB_0)))g^\dagger\;\bigg|\;g\in\Cl\right\}.
    \end{equation}
    Since the stabilizer polytope is closed under the action of the Clifford group, each of these balls is contained in $\STAB_0$.

    By convexity, the stabilizer polytope also contains a uniform mixture of these balls with respect to the Minkowski sum,
    \begin{equation}
        \frac{1}{|\Cl|}\sum\limits_{g\in\Cl}g\left(A+B(\bbzero,r(\STAB_0))\right)g^\dagger\subset\STAB_0.
    \end{equation}
    But since conjugation by unitaries preserves the Hilbert-Schmidt inner product,
    \begin{equation}
        gB(\bbzero,r(\STAB_0))g^\dagger=B(\bbzero,r(\STAB_0)),\;\forall g\in\Cl.
    \end{equation}
    Further, the Clifford group on prime dimensional qudits is always a unitary $2$-design~\cite{GrossEisert2007}, and for qubits it is a unitary $3$-design~\cite{Webb2016,Zhu2017}, so
    \begin{align}
        \frac{1}{|\Cl|}&\sum\limits_{g\in\Cl}g\left(A+B(\bbzero,r(\STAB_0))\right)g^\dagger\\
        =& \left[\frac{1}{|\Cl|}\sum\limits_{g\in\Cl}gAg^\dagger\right]+B(\bbzero,r(\STAB_0))\\
        =& B(\bbzero,r(\STAB_0))\subset\STAB_0.
    \end{align}
    This completes the proof.
\end{proof}

Note that the only property of the stabilizer polytope we use in the proof of Lemma~\ref{Lemma:StabRadiusOriginCentred} is that its automorphism group contains the Clifford group. Therefore, the same statement applies to any polytope in $\Herm_1(\mathcal{H})$ with this property. In particular, it applies to the $\Lambda$ polytope and to the Wigner polytope. Essentially the same argument shows that $\Lambda$ is contained in $d^{2n}-1$-dimensional ball of radius $R(\Lambda_0)$ which is centred at the maximally mixed state. To find the inradius of the stabilizer polytope, we use the following lemma, a variant of one of the cases in Ref.~\cite[\S1.2]{GritzmannKlee1992}.

\begin{Lemma}\label{Lemma:InradiusPolarCircumradius}
    Let $P$ be a polytope. If $P$ contains a ball of radius $r(P)$ centred at the origin, and $P^\circ$ is contained in a ball of radius $R(P^\circ)$ centred at the origin, then $r(P)R(P^\circ)=1$.
\end{Lemma}

\begin{proof}[Proof of Lemma~\ref{Lemma:InradiusPolarCircumradius}]
    Let $B:=B(0,1)$, $r:=r(P)$, and $R:=R(P^\circ)$. By assumption, $rB\subseteq P$ and
    $P^\circ\subseteq RB$. We use three facts about polarity~\cite{Ziegler1995,NariciBeckenstein2010}: (1)~taking polar duals reverses inclusions, (2)~since $P$ contains the origin, $(P^\circ)^\circ=P$, and (3)~for any $t>0$, $(tB)^\circ=\frac{1}{t}B$. Then, we have
    \begin{equation}
        P^\circ \subseteq (rB)^\circ=\frac1r\,B
    \end{equation}
    and
    \begin{equation}
        \frac1R\,B=(RB)^\circ\subseteq (P^\circ)^\circ=P.
    \end{equation}
    Therefore, $R\le \frac1r$ and $r\ge \frac1R$, so $rR=1$.
\end{proof}    

Using this Lemma~\ref{Lemma:InradiusPolarCircumradius}, our strategy for finding the inradius of the multiqudit stabilizer polytope proceeds in two steps. First, we relate the inradius of $\STAB_0$ to the circumradius of $\Lambda_0$ via $r(\STAB_0)=1/(d^nR(\Lambda_0))$, and second we find the circumradius of $\Lambda_0$. Since the vertices of the $\Lambda$ polytopes have a different structure for $d=2$ versus odd-prime $d$, here we split into two cases.

\subsubsection{Qubits}

The following result gives us a bound on the inradius.

\begin{Lemma}\label{Lemma:CNCRadius}
    For any CNC operator $A_\Omega^\gamma$, we have
    \begin{equation}
        \Tr({A_\Omega^\gamma}^2)=\frac{|\Omega|}{d^n}.
    \end{equation}
\end{Lemma}

\begin{proof}[Proof of lemma~\ref{Lemma:CNCRadius}]
    By direct calculation we obtain
    \begin{equation}
    \begin{aligned}
        \Tr({A_\Omega^\gamma}^2) =& \Tr({A_\Omega^\gamma}^\dagger A_\Omega^\gamma)\\
        =& \frac{1}{d^{2n}}\sum\limits_{u,v\in\Omega}\omega^{\gamma(u)-\gamma(v)}\Tr(T_u^\dagger T_v)\\
        =& \frac{1}{d^n}\sum\limits_{u,v\in\Omega}\omega^{\gamma(u)-\gamma(v)}\delta_{u,v}\\
        =&\frac{1}{d^n}\sum\limits_{u\in\Omega}1=\frac{|\Omega|}{d^n}.
    \end{aligned}
    \end{equation}
    In the first line we use the fact that CNC operators are Hermitian~\cite{RaussendorfZurel2020,ZurelHeimendahl2024b}, in the second line we use the definition of CNC operators Eq.~\eqref{Equation:CNCOperators}, and the third line we use the orthogonality of Pauli operators, $\Tr(T_u^\dagger T_v)=d^n\delta_{u,v}$.
\end{proof}

Using the classification of CNC sets Theorem~\ref{Theorem:CNCClassification}, we can see that the size of a maximal CNC set of the form Eq.~\eqref{Equation:QubitCNCSets} with parameters $m$ is $|\Omega|=(2m+2)2^{n-m}$. Therefore, among CNC operators $A_\Omega^\gamma$, the value $\Tr({A_\Omega^\gamma}^2)$ is maximized for the case $m=1$, with $\Tr({A_\Omega^\gamma}^2)=2$.

We make the following conjecture.

\begin{Conjecture}\label{Conjecture:LambdaOutradius}
    For any number of qubits, for any $X\in\Herm_1(\mathcal{H})$, if $\Tr(X^2)>2$ then $X\notin\Lambda$.
\end{Conjecture}

In Appendix~\ref{Section:ConjectureEvidence}, we show that this conjecture holds for $n=1$ and $n=2$, and we provide evidence that it holds generally. Since CNC operators are in $\Lambda$, if this conjecture is true, then the maximum value of $\Tr(X^2)$ over elements of $\Lambda$ is equal to $2$, and it is achieved by the CNC operators of type $m=1$. Now, with the following theorem, we find a formula for the inradius of the stabilizer polytope conditional on Conjecture~\ref{Conjecture:LambdaOutradius}.

\begin{Theorem}\label{Theorem:StabilizerInradius}
    If Conjecture~\ref{Conjecture:LambdaOutradius} is true, then the Hilbert-Schmidt inradius of the $n$-qubit stabilizer polytope is
    \begin{equation}\label{qubit_MMS_HS_radius}
       r(\ASTAB_0) = r(\STAB_0)=\frac{1}{\sqrt{2^n(2^{n+1}-1)}}.
    \end{equation}
\end{Theorem}

\begin{proof}[Proof of Theorem~\ref{Theorem:StabilizerInradius}]
    The first equality holds by unitary invariance.  Assuming Conjecture~\ref{Conjecture:LambdaOutradius} is true, we have $\max_{X\in\Lambda}\Tr(X^2)=2$, or equivalently,
    \begin{equation}
        R(\Lambda_0)=\sqrt{\max_{X_0\in\Lambda_0}\Tr(X_0^2)}=\sqrt{2-\frac{1}{2^n}}.
    \end{equation}
    Therefore,
    \begin{equation}
        R(-2^n\Lambda_0)=2^nR(\Lambda_0)=\sqrt{2^n(2^{n+1}-1)}.
    \end{equation}
    Then using Lemma~\ref{Lemma:InradiusPolarCircumradius} and using polar duality of $\STAB_0$ and $-2^n\Lambda_0$, the result follows.
\end{proof}

\begin{Corollary}
    Assuming Conjecture~\ref{Conjecture:LambdaOutradius} is true, then every state $\rho$ with purity
    \begin{equation}\label{insphere_purity_qubits}
        \Tr(\rho^2)\le\frac{1}{2^n}+\frac{1}{2^n(2^{n+1}-1)}=\frac{1}{2^{n}-\frac12}
    \end{equation}
    is a mixture of stabilizer states.
\end{Corollary}

This corollary demonstrates that our Conjecture~\ref{Conjecture:LambdaOutradius} is equivalent to Conjecture~1 in Ref.~\cite{LiuRoth2025}, which constructed a PPT-like analog for the resource theory of magic. This connection offers an additional geometric interpretation (conditioned on the validity of our conjecture) of the purity in Ref.~\cite{LiuRoth2025} as corresponding to the inradius of the stabilizer polytope.

\subsubsection{Odd-prime-dimensional qudits}

The structure of the argument is similar in the case of odd-prime-dimensional qudits, the main difference is that in this case, we know that the vertices of $\Lambda$ with the largest Hilbert-Schmidt norm are the phase space point operators of the Wigner function~\cite{Gross2006, Gross2008, ZurelHeimendahl2024}. This gives us the following (unconditional) theorem.

\begin{Theorem}\label{Theorem:StabilizerInradiusQudits}
    For any number of qudits $n$ of any odd-prime dimension $d$, the Hilbert-Schmidt inradius of the $n$-qudit stabilizer polytope is
    \begin{equation}\label{qudit_MMS_HS_radius}
        r(\ASTAB_0)=r(\STAB_0)=\frac{1}{\sqrt{d^n(d^{2n}-1)}}.
    \end{equation}
\end{Theorem}

\begin{proof}[Proof of Theorem~\ref{Theorem:StabilizerInradiusQudits}]
    The first equality holds by unitary invariance.  Phase space point operators of the Wigner function are vertices of $\Lambda$~\cite{ZurelHeimendahl2024}. They are CNC-type operators, $A_\Omega^\gamma$, with $\Omega=\Z_d^{2n}$ and $\gamma:\Z_d^{2n}\to\Z_d$ a linear function. By Lemma~\ref{Lemma:CNCRadius}, these operators have
    \begin{equation}
        \Tr({A_{\Z_d^{2n}}^\gamma}^2)=\frac{|\Z_d^{2n}|}{d^n}=d^n.
    \end{equation}
    
    With Lemma~3 from Ref.~\cite{ZurelHeimendahl2024}, coefficients of the Pauli basis expansion for any vertex $A$ of $\Lambda$ satisfy
    \begin{equation}
        |\Tr(T_uA)|\le1,\quad\forall u\in\Z_d^{2n}.
    \end{equation}
    Therefore,
    \begin{align}
        \Tr(A^2)=&\frac{1}{d^n}\sum\limits_{v\in\Z_d^{2n}}|\Tr(T_vA)|^2\\
        \le&\frac{1}{d^n}\sum_{v\in\Z_d^{2n}}|\Tr(T_vA)|\le d^n.
    \end{align}
    That is, no vertex $A$ of $\Lambda$ can have $\Tr({A}^2)>d^n$, and so $\max_{X\in\Lambda}\Tr(X^2)=d^n$.

    Equivalently,
    \begin{equation}
        R(\Lambda_0)=\sqrt{\max_{X_0\in\Lambda_0}\Tr(X_0^2)}=\sqrt{d^n-\frac{1}{d^n}}
    \end{equation}
    Therefore,
    \begin{equation}
        R(-d^n\Lambda_0)=d^nR(\Lambda_0)=\sqrt{d^n(d^{2n}-1)}.
    \end{equation}
    Finally, using Lemma~\ref{Lemma:InradiusPolarCircumradius} and polar duality of $\STAB_0$ and $-d^n\Lambda_0$, we get the stated result.
\end{proof}

\begin{Corollary}
    Every state $\rho$ with purity
    \begin{equation}
        \Tr(\rho^2)\le\frac{1}{d^n}+\frac{1}{d^n(d^{2n}-1)}=\frac{1}{d^{n}-\frac{1}{d^n}}
    \end{equation}
    is a mixture of stabilizer states.
\end{Corollary}

\begin{figure*}
    \centering
    \includegraphics[width=0.375\linewidth]{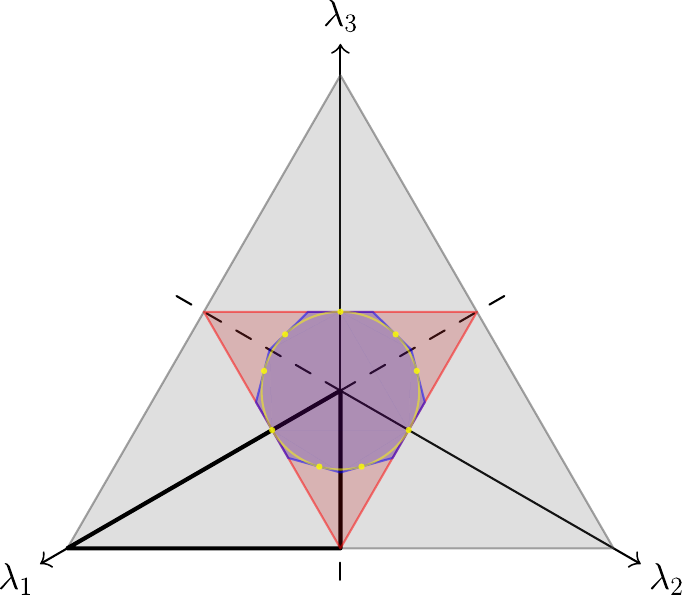}
    \hspace{0.05\linewidth}
    \includegraphics[width=0.525\linewidth]{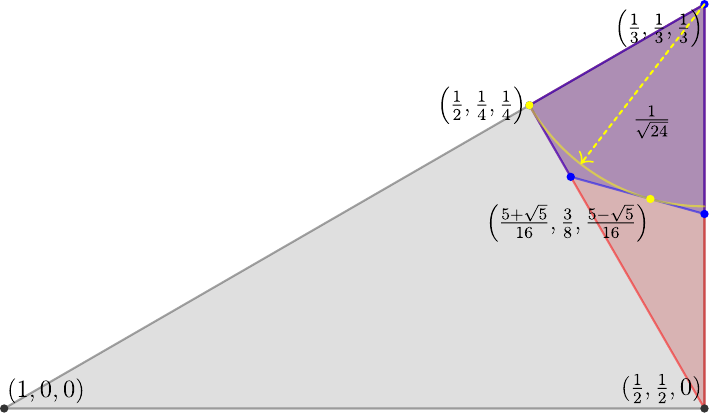}
    \caption{Left: Ternary plot of the spectra of single-qutrit absolutely stabilizer states (blue), absolutely Wigner-positive states (red), together with the inradius (yellow) common to both.  Yellow dots denote absolutely stabilizer states on the boundary of both $\ASTAB$ and its inradius.  Right: the associated Weyl chamber with the spectra of these special points displayed together with the Hilbert-Schmidt radius $r(\ASTAB)=\frac{1}{\sqrt{24}}$.} \label{Figure:OneQutritInradius}
\end{figure*}

\subsection{Inradius and circumradius of AWP states}

Using similar methods as above, we can get the inradius of the Wigner polytope.

\begin{Theorem}\label{Theorem:AWPInradiusOutradius}
    For any number of qudits $n$ of any odd-prime Hilbert space dimension $d$, the Hilbert-Schmidt inradius of $\WP_0$ and $\AWP_0$ is
    \begin{equation}
        r(\AWP_0)=r(\WP_0)=\frac{1}{\sqrt{d^n(d^{2n}-1)}},
    \end{equation}
    and the circumradius of $\AWP_0$ is
    \begin{equation}\label{AWP_outsphere_radius}
        R(\AWP_0)=\frac{1}{\sqrt{d^n(d^n-1)}}.
    \end{equation}
\end{Theorem}

\begin{proof}[Proof of Theorem~\ref{Theorem:AWPInradiusOutradius}]
    The Wigner function is Clifford covariant~\cite{Gross2006,Gross2008,RaussendorfFeldmann2023}. Therefore, the Wigner polytope is closed under the action of the Clifford group, and so using the same argument as in the proof of Lemma~\ref{Lemma:StabRadiusOriginCentred}, we can assume that the ball contained in the Wigner polytope that achieves the inradius is centred at the maximally mixed state, or equivalently, the ball contained in $\WP_0$ that achieves the inradius is centred at $\bbzero$.

    Then the proof of the first statement (inradius) is essentially the same as the proof of Theorem~\ref{Theorem:StabilizerInradiusQudits}, except the Wigner polytope and its dual each have only one unitary group orbit of vertices, and they're both proportional to phase space point operators.

    For the second statement (circumradius), by unitary invariance of the ball centred at the maximally mixed state with radius $R(\AWP_0)$, it suffices to find the largest possible Hilbert-Schmidt distance from the maximally mixed state to a diagonal AWP state. The possible entries of a diagonal AWP state forms a polytope characterized in Theorem~\ref{Theorem:SpectralAWPCharacterization}. The vertices of this polytope have diagonal entries of the form Eq.~\eqref{Equation:AWPSpectralVertex}. The Hilbert-Schmidt distance from the maximally mixed state to the first of these vertices is
    \begin{equation}
        \bigg|\bigg|v_{d^n-1}-\left(\frac{1}{d^n},\frac{1}{d^n},\dots\right)\bigg|\bigg|_2 = \frac{1}{\sqrt{d^n(d^n-1)}},
    \end{equation}
    and to the second of the vertices is
    \begin{equation}
        \bigg|\bigg|v_1'-\left(\frac{1}{d^n},\frac{1}{d^n},\dots\right)\bigg|\bigg|_2 = \sqrt{\frac{d^n-1}{d^n(d^n+
        1)^2}},
    \end{equation}
    where $v_{d^n-1}$ and $v_1'$ are defined in the proof of Theorem~\ref{Theorem:SpectralAWPCharacterization}. The first value is always the larger of the two. This proves the second statement of the theorem.
\end{proof}
See Fig.~\ref{Figure:OneQutritInradius} for an illustration of the AWP inradius in the case of a single qutrit, together with the spectra of states that are both on the $\ASTAB$ boundary and its inradius.

\begin{Corollary}
    Every state $\rho$ with purity
    \begin{equation}
        \Tr(\rho^2)\le\frac{1}{d^n}+\frac{1}{d^n(d^{2n}-1)}=\frac{1}{d^n-\frac{1}{d^n}}
    \end{equation}
    is AWP, and every state $\rho$ with purity
    \begin{equation}
        \Tr(\rho^2)>\frac{1}{d^n}+\frac{1}{d^n(d^n-1)}=\frac{1}{d^n-1}.
    \end{equation}
    is not AWP.
\end{Corollary}

That is, if $\Tr(\rho^2)\le1/(d^n-1/d^n)$, then for every unitary $U\in\mathrm{U}(d^n)$, the Wigner function of $U\rho U^\dagger$ is nonnegative, and if $\Tr(\rho^2)>1/(d^n-1)$ then there exists a unitary $U\in\mathrm{U}(d^n)$ such that the Wigner function of $U\rho U^\dagger$ takes negative values.

\subsection{Comparison to other balls}

Here we briefly compare the Hilbert-Schmidt balls derived above with other balls of interest.  We emphasize that our $r(\STAB_0)$ for qubits relies on a conjecture, while for odd-prime-qudits it is a theorem.

In dimensions $d^n>2$, it is no longer true that all Hermitian operators with unit Hilbert-Schmidt norm correspond to positive semi-definite (PSD) operators, i.e.~quantum states. The radius of the largest inscribed ball of PSD matrices centred at the maximally mixed state $\frac{1}{d^n}\mathbb I$ is well-known~\cite{Harriman1978, BengtssonZyczkowski2006} to be
\begin{equation}
    r_{\PSD} = \frac{1}{\sqrt{d^n(d^n-1)}} > r(\STAB_0).
\end{equation}
Hence all traceless Hermitian operators within $r(\STAB_0)$ can be interpreted as (shifted) quantum states. This is true regardless of $d$ being $2$ or an odd-prime, however for $d=2$ it is notable how close the two radii are.  For odd-prime dimensions, this radius exactly matches that of the circumradius of the set of absolutely Wigner-positive states~\eqref{AWP_outsphere_radius}:
\begin{equation}
    r_{\PSD} = R(\AWP_0).
\end{equation}
Hence every traceless Hermitian operator inside $\AWP_0$ also corresponds to a valid quantum state.

Another notable ball is the Gurvits-Barnum ball of absolutely separable states, centred at the maximally mixed state~\cite{GurvitsBarnum2003} (see also~\cite{KusZyczkowski2001}) with radius
\begin{equation}
    r_{\text{GB}} = \frac{2}{(\sqrt{2}d)^n},
\end{equation}
where separability is defined with respect to the usual partition $(\mathbb C^d)^{\otimes n}$. Note that $r_{\text{GB}} < r_{\PSD}$. In comparison to the stabilizer polytope we have
\begin{equation}
    r(\STAB_0) < r_{\text{GB}},
\end{equation}
for all odd prime dimensions $d$ and $n \geq 2$.  For the $d=2$ case, this relation fails for $n\geq 3$.  That being said, the Gurvits-Barnum radius is only known to be optimal in the bipartite case~\cite{GurvitsBarnum2002}: larger balls have been found in the multipartite case~\cite{GurvitsBarnum2005, Hildebrand2007, WenKempf2025} and the optimal radius is unknown.  Hence it may still be true that all states in multi-qubit stabilizer inscribed ball are absolutely separable, though this is an open question.  

To summarize all of the above, together with results from previous sections, we have the following chain of inclusions of Hilbert-Schmidt balls (i.e.~rescaled purities):
\begin{widetext}
\begin{equation}\label{radius_comparisons}
\begin{cases}
    r(\ASTAB_0) = r(\STAB_0) = r(\AWP_0) = r(\WP_0) < r_{\text{GB}} < r_{\PSD} = R(\AWP_0) \quad &d=\text{odd prime} \\
    r(\ASTAB_0) = r(\STAB_0) < r_{\text{GB}} < r_{\PSD}, \quad &d=\text{2}, n=2 \\
    r_{\text{GB}} < r(\ASTAB_0) = r(\STAB_0) < r_{\PSD}, \quad &d=\text{2}, n\geq 3
\end{cases}
\end{equation}
\end{widetext}
Eq.~\eqref{radius_comparisons} is to be read as follows.  For odd-prime dimensions, the entire chain has been proven and does not rely on any conjectures.  In the case of a single qudit, simply remove the $r_{\mathrm{GB}}$ radius and the rest remains true.  Note that since the circumradius of $\AWP_0$ is greater than the Gurvits-Barnum ball, it is likely there exist entangled AWP states; see also~\cite{DenisMartin2024}.  For the $d=2$ cases, $r(\STAB_0)$ depends on our Conjecture~\ref{Conjecture:LambdaOutradius}, or, equivalently, Conjecture~1 from Ref.~\cite{LiuRoth2025}.

\section{Witnesses of multi-qubit magic and the CNC hierarchy}\label{Section:ACNC}

Here we discuss in a bit more detail our Conjecture~\ref{Conjecture:LambdaMajorization}, its consequences, and how it relates to previous work.  In doing so, let us first formalize and develop further the CNC-related structure that has arisen from the previous sections.  

\subsection{The CNC hierarchy, the Triangle Criterion, and single-copy stabilizer reductions}
Recall from the background Section~\ref{sec:CNC_background} that maximal CNC operators $A_{\Omega}^\gamma$ (i.e.\ those with maximal support in the Pauli basis) of type $m$ are vertices of the Lambda polytope.  This suggests the following definition, 

\begin{Definition}
    Consider a system of $n$ qubits and its associated Lambda polytope $\Lambda$.  Let $\Lambda_{\mathrm{CNC}} \subset \Lambda$ denote the convex hull of all CNC vertices.  Moreover, for $m\in\{1,...,n\}$, denote by $\Lambda_{\mathrm{CNC}}^{(m)} \subset \Lambda_{\mathrm{CNC}}$ the convex hull of CNC vertices of type $m$.  The set of \emph{CNC states of type $m$} is
\begin{equation}\label{CNC_state_individual_type}
    \mathrm{CNC}^{(m)} := \{ X \in \Herm_1 \, |\, \Tr[X^\dagger V] \geq 0 \quad \forall \, V \in \Lambda_{\mathrm{CNC}}^{(m)} \}.
\end{equation}
Relatedly, the set of \emph{CNC states of cumulative type $m$} is
\begin{equation}
    \CNC^{(\leq m)} := \bigcap_{r=1}^m \CNC^{(r)}.
\end{equation}
States in the highest cumulative level are CNC of all possible types and are here simply called \emph{CNC states}:
\begin{equation}
\begin{aligned}
    \CNC &:= \CNC^{(\leq n)}.
\end{aligned}
\end{equation}
\end{Definition} 
\noindent Eq.~\eqref{CNC_state_individual_type} is analogous to the relationship between STAB and $\Lambda$, but since the number of vertices / witnesses have decreased we have the inclusion
\begin{equation}\label{CNC-STAB_subset}
    \mathrm{vert}(\Lambda_{\mathrm{CNC}}^{(m)}) \subseteq \mathrm{vert}(\Lambda) \quad \Longrightarrow \quad \STAB \subseteq \mathrm{CNC}^{(m)}
\end{equation}
for each $m$, together with the hierarchy
\begin{equation}\label{CNC_hierarchy}
    \STAB \subseteq \CNC \subseteq \CNC^{(\leq n-1)} \subseteq ... \subseteq \CNC^{(\leq 1)}.
\end{equation}
 
On the other hand, recall that Ref.~\cite{LiuRoth2025} recently developed the Triangle Criterion, a PPT-like detection test in the resource theory of magic.  It can technically be seen as a collection of inequalities for witnessing a certain kind of magic state.  Specifically, an $n$-qubit state will be detected by this criterion if and only if it can be transformed into a single-qubit magic state using Clifford unitaries, Pauli measurements with post-selection, and stabilizer ancillas.  Such operations generate what could be called \textit{single-copy reduction protocols}, as they consist of the usual stabilizer operations but without the ability to use copies of the initial state.  Let us here denote TC as the set of states that do \textit{not} violate any of the triangle criterion inequalities.  We may now make the following identification between TC and the first level of the CNC hierarchy.

\begin{Theorem}\label{Theorem:TriangleCriterion}
    Consider a system of $n$ qubits.  Then 
    \begin{equation}
        \mathrm{TC} = \CNC^{(1)} = \CNC^{(\leq 1)}.
    \end{equation}
\end{Theorem}

\begin{proof}[Proof of Theorem~\ref{Theorem:TriangleCriterion}]
    TC can be defined as follows (See Appendix B1 of~\cite{LiuRoth2025}):
    \begin{equation}
        \mathrm{TC} := \{ \rho \in \Herm_1 \, | \, \Tr[\rho \, W_g ] \geq 0 \},
    \end{equation}
    for all witnesses $W_g$ of the form
    \begin{equation}\label{TC_witness_normal_form}
        W_g = g [ ( \ketbra{+}{+} + 
        \ketbra{+i}{+i} - \ketbra{0}{0}) \otimes \ketbra{0^{n-1}}{0^{n-1}} ] g^\dagger,
    \end{equation}
    where $g \in \Cl$ is an $n$-qubit Clifford unitary.  These witness operators are exactly the normal form of a CNC operator of type $m=1$, Eq.~\eqref{CNC_normal_form}, where the associated single-qubit CNC operator is
    \begin{equation}
        \ketbra{+}{+} + 
        \ketbra{+i}{+i} - \ketbra{0}{0} = \frac12(I + X + Y - Z),
    \end{equation}
    and the $(n-1)$-qubit stabilizer state is $\ketbra{0^{n-1}}{0^{n-1}}$.
    
    Conversely, every single-qubit Lambda vertex $\frac12(I \pm X \pm Y \pm Z)$ is Clifford-equivalent to the one above, Eq.~\eqref{TC_witness_normal_form}, and any $(n-1)$-qubit stabilizer ancilla is Clifford equivalent to $\ketbra{0^{n-1}}{0^{n-1}}$.  Hence any Lambda vertex of CNC type 1 can be transformed into the form Eq.~\eqref{TC_witness_normal_form} via a global Clifford unitary.  

    Finally, $\CNC^{(1)}=\CNC^{(\leq 1)}$ by definition.
\end{proof}

This establishes a direct connection between closed-and-non-contextual operators (of type $m=1$) and single-copy reductions to single-qubit magic states.  Such a connection suggests a deeper and broader relationship between the entire CNC hierarchy Eq.~\eqref{CNC_hierarchy} and single-copy reductions to multi-qubit magic states.  Indeed, both of the sets $\CNC^{(m)}$ and $\CNC^{(\leq m)}$ naturally generalize the Triangle Criterion in the following sense: any state outside of $\CNC^{(\leq m)}$ (or even just $\CNC^{(m)}$) is guaranteed to be single-copy-reducible to an $m$-qubit magic state.  The reasoning is similar to Theorem~2 of~\cite{LiuRoth2025} and depends on the normal form for $n$-qubit CNC Lambda vertices, $A_\Omega^\gamma = g ( A_{\tilde{\Omega}}^{\tilde{\gamma}} \otimes \ketbra{\sigma}{\sigma} )g^\dagger$, where $g$ is an $n$-qubit Clifford unitary, $\ket{\sigma}$ is an $(n-m)$-qubit stabilizer state, and $A_{\tilde{\Omega}}^{\tilde{\gamma}}$ is a CNC Lambda vertex for $m$ qubits; see Eq.~\eqref{CNC_normal_form}.  This normal form implies that evaluating the $n$-qubit magic witness $A_\Omega^\gamma \in \Lambda_{\text{CNC}}^{(m)}$ is operationally equivalent to performing a global Clifford gate, projecting $n-m$ qubits onto a stabilizer state $\ket{\sigma}$, then testing if the reduced $m$-qubit state is in the $m$-qubit $\CNC^{(m)} \supset \STAB$ set.  Hence if $\Tr[\rho \, A_\Omega^\gamma] <0$ for any $A_\Omega^\gamma \in \Lambda_{\text{CNC}}^{(m)}$ then $\rho$ is single-copy-reducible to an $m$-qubit magic state.

For a simple example, we can construct a witness of $2$-qubit magic as follows. Fix a set of Pauli projectors $\Pi_1,\Pi_2,\Pi_3,\Pi_4,\Pi_5$ corresponding to pair-wise anticommuting Pauli operators such that $\rank(\Pi_j)=2$ and $\Tr(\Pi_i\Pi_j)=1$ for $i\ne j$. As a concrete example, we can take
\begin{equation}
    \begin{cases}
        \Pi_1=(\mathbb{I}+X_1)/2,\\
        \Pi_2=(\mathbb{I}+Y_1)/2,\\
        \Pi_3=(\mathbb{I}+Z_1X_2)/2,\\
        \Pi_4=(\mathbb{I}+Z_1Y_2)/2,\\
        \Pi_5=(\mathbb{I}+Z_1Z_2)/2.
    \end{cases}
\end{equation}
Then for any $n$-qubit Clifford $g$,
\begin{equation}
    W_g=g\left[\frac{1}{2}(\Pi_1+\Pi_2+\Pi_3-\Pi_4-\Pi_5)\otimes\ketbra{0}{0}^{\otimes n-2}\right]g^\dagger
\end{equation}
witnesses $2$-qubit magic; by comparison, see the Triangle Criterion witnesses for $1$-qubit magic in Eq.~\eqref{TC_witness_normal_form}. Here $W_g$ is a CNC operator of type $m=2$, or, more generally, any CNC operator of type $1$ up to $m$ witnesses $m$-qubit magic.

On the other hand, membership of $\CNC^{(\leq m)}$ for general $m$ does not immediately imply single-copy-\textit{non}-reducibility to an $m$-qubit magic state.  This is because, unlike the $n=1$ single-qubit case, the inclusion $\CNC^{(\leq m)}\supset \STAB$ for $m$ qubits is now strict.  Hence even if a state $\rho$ was not detected by any witness in $\Lambda_{\text{CNC}}^{(\leq m)}$, it could still in principle be single-copy reducible to an $m$-qubit magic state. For example, consider the state
\begin{equation*}
    \rho=\frac{1}{2^n}\mathbb{I}-\epsilon A_\alpha
\end{equation*}
where
\begin{equation*}
    A_\alpha=A_\beta^{(2)}\otimes \ketbra{\sigma}{\sigma}
\end{equation*}
where $A_\beta^{(2)}$ is a non-CNC $2$-qubit vertex of $\Lambda$ and $\ketbra{\sigma}{\sigma}$ is an $n-2$-qubit stabilizer state. Then $\epsilon>0$ can always be made small enough such that $\rho$ does not violate any CNC inequalities, but it is still reducible to a $2$-qubit magic state, namely, a $2$-qubit state $\rho'$ such that $\Tr(\rho'A_\beta^{(2)})<0$.

In summary, for $n$-qubit systems
\begin{align*}
    \rho \notin \CNC^{(m)} \quad &\Rightarrow \quad \text{reducible to $m$-qubit  magic} \\
    \rho \in \CNC^{(m)} \quad &\nRightarrow \quad \text{not reducible to $m$-qubit  magic}.
\end{align*}

\subsection{The ACNC hierarchy}
We can also define the unitarily-invariant version of the CNC hierarchy.
\begin{Definition}
    Consider a system of $n$ qubits.  A state is \emph{absolutely CNC of type $m$} if
\begin{equation}
    U \rho U^\dagger \in \mathrm{CNC}^{(m)} ,\quad \forall \, U \in \mathrm{U}(2^n).
\end{equation}
The set of such states is denoted \emph{ACNC}$^{(m)}$ where $m\in\{1,...,n\}$.  Relatedly, the set of \emph{absolutely CNC states of cumulative type $m$} is
\begin{equation}
    \begin{aligned}
    \ACNC^{(\leq m)} &:= \bigcap_{r=1}^m \ACNC^{(r)}.
    \end{aligned}
\end{equation}
States in the highest cumulative level are absolutely CNC of all possible types and are here simply called \emph{absolutely CNC states}:
\begin{equation}
\begin{aligned}
    \ACNC &:= \ACNC^{(\leq n)}.
\end{aligned}
\end{equation}
\end{Definition}

These definitions produce a chain of inclusions analogous to Eq.~\eqref{CNC_hierarchy}:
\begin{equation}\label{ACNC_hierarchy}
    \ASTAB \subseteq \ACNC \subseteq \ACNC^{(\leq n-1)} \subseteq ... \subseteq \ACNC^{(\leq 1)}.
\end{equation}

With Lemma~\ref{Lemma:CNCEigenvaluesQubits} and the re-use of the machinery developed in the previous sections, a complete characterization of $\ACNC^{(\leq m)}$ is as the set of states $\rho$ with spectra obeying the $m$ inequalities
\begin{equation}\label{CNC-ratio-inequalities}
    \frac{\sum\limits_{k=1}^{2^{r-1}}\lambda_k^\uparrow(\rho)}{\sum\limits_{k=1}^{2^{r-1}}\lambda_k^\downarrow(\rho)} \geq \frac{\sqrt{2r+1}-1}{\sqrt{2r+1}+1}, \quad \forall\,r\in\{1,...,m\} .
\end{equation}
That is, for each $r$, the ratio between the sum of the smallest $2^{r-1}$ eigenvalues and the sum of the largest $2^{r-1}$ eigenvalues is bounded from below by a value that approaches unity as $r$ increases.  This is essentially just a re-expression of Theorem~\ref{Theorem:QubitConjecturedCharacterization} but with reference to the ACNC hierarchy Eq.~\eqref{ACNC_hierarchy} rather than to $\ASTAB$ or Conjecture~\ref{Conjecture:LambdaMajorization}.
Membership for the individual ACNC type-$r$ set is governed only by the individual type-$r$ inequality.
Note that due to the inclusion $\ASTAB \subseteq \ACNC$, the $n$ inequalities characterizing ACNC, Eq.~\eqref{CNC-ratio-inequalities}, already constitute a simple necessary condition for absolute stabilizerness.  Equivalently, they form a sufficient condition for ruling out absolute stabilizerness: if any of the $n$ inequalities are violated then there exists a basis in which the state is magic.

Similar to TC = CNC$^{(1)}$ having an operational interpretation of being closed under single-copy stabilizer reduction protocols, here we can investigate an analogous operational interpretation for the unitarily-invariant versions $\ACNC^{(m)}$.  In order to do this, first recall from Sec.~\ref{sec:absolutely_stabilizer_definition} that $\ASTAB$ was closed under destructive subsystem post-selection onto arbitrary pure states.  Now we show that both ACNC$^{(m)}$ and $\ACNC^{(\leq m)}$, which are \textit{a priori} larger than $\ASTAB$ for each type $m$, are in fact also closed under such post-selections.  This requires the following preliminary theorem:
\begin{Theorem*}(Cauchy interlacing theorem)~\cite{MarshallArnold2011, Hwang2004}.
\label{thm:Cauchy_interlace}
    Let $A$ be a $D\times D$ Hermitian matrix and $B = P^\dagger AP$ be a $d\times d$ matrix obtained by a $D \times d$ isometry $P$.  Let $\{\lambda^\downarrow_i\}_{i=1}^D$ and $\{ \mu^\downarrow_i \}_{i=1}^d$ be the decreasingly sorted eigenvalues of the two matrices.  Then
    \begin{equation}
        \lambda^\downarrow_i \geq \mu^\downarrow_i \geq \lambda^\downarrow_{i + (D - d)}, \quad \forall \, i \in \{1,...,d\}.
    \end{equation}
\end{Theorem*}
\noindent We may now state the result:
\begin{Theorem}\label{Theorem:ACNC-closed-under-PS}
    Consider a system of $n$ qubits. The set $\ACNC^{(m)}$ is closed under destructive subsystem post-selection onto arbitrary pure states of $(n-k)$ qubits, provided $m\leq k$.
\end{Theorem}
\begin{proof}[Proof of Theorem~\ref{Theorem:ACNC-closed-under-PS}]
    Let $\rho \in \ACNC^{(m)}$.  Hence the eigenvalues $\{\lambda_i^\downarrow\}_{i=1}^{2^n}$ of $\rho$ satisfy Eq.~\eqref{CNC-ratio-inequalities}.  Similar to Eq.~\eqref{stabilizer_isometry} let $\ket{\phi} \in \mathbb C^{2^{n-k}}$ be any $(n-k)$-qubit state and let $V: \mathbb C^{2^k} \rightarrow \mathbb C^{2^n}$ be the associated isometry, defined as $V\ket{\psi} = \ket{\psi}\otimes \ket{\phi}$.  The unnormalized post-selection of $\rho$ onto $\ket{\phi}$ is given by $V^\dagger \rho V$.  Denote this set of unnormalized eigenvalues $\{\mu_i^\downarrow\}_{i=1}^{2^k}$.
    
    By the Cauchy interlacing theorem, each of $2^k$ sorted eigenvalues of $V^\dagger \rho V$ must lie within a sliding window of the sorted $2^n$ eigenvalues of $\rho$.  As a result, the ordered spectral sums cannot grow, either from above or below:
\begin{equation}
    \sum_{k=1}^{N} \lambda_i^\downarrow \geq \sum_{k=1}^{N} \mu_i^\downarrow, \qquad \sum_{k=1}^N \mu_i^\uparrow \geq \sum_{k=1}^N \lambda_i^\uparrow,
\end{equation}
for any $N \in \{ 1,...,2^k \}$.  Hence Eq.~\eqref{CNC-ratio-inequalities} remains satisfied:
\begin{equation}
    \frac{\sum\limits_{k=1}^{2^{m-1}}\mu_k^\uparrow}{\sum\limits_{k=1}^{2^{m-1}}\mu_k^\downarrow} \geq \frac{\sum\limits_{k=1}^{2^{m-1}}\lambda_k^\uparrow}{\sum\limits_{k=1}^{2^{m-1}}\lambda_k^\downarrow} \geq \frac{\sqrt{2m+1}-1}{\sqrt{2m+1}+1}.
\end{equation}
And since this is a ratio, the renormalization $\Tr[V^\dagger \rho V]$ is canceled out, leading to the conclusion
\begin{equation}
\rho \in \ACNC^{(m)} \quad \Rightarrow \quad 
    \frac{V^\dagger \rho V}{\Tr[V^\dagger \rho V]} \in \ACNC^{(m)}.
\end{equation}
\end{proof}

It then follows by definition that $\ACNC^{(\leq m)}$ is also closed under post-selection as each element in the intersection is closed.  Moreover, similar to the argument for $\ASTAB$ in Sec.~\ref{sec:absolutely_stabilizer_definition}, both ACNC$^{(m)}$ and $\ACNC^{(\leq m)}$ are also closed under adaptive unitaries and partial trace.  Hence we may give the following operational interpretation analogous to the non-absolute case: states outside of $\ACNC^{(m)}$ are guaranteed to be reducible to an $m$-qubit magic state under protocols consisting of arbitrary unitaries, possibly conditioned on arbitrary destructive post-selections, and partial trace.  This is because if $\rho \notin \ACNC^{(m)}$ then by definition there exists a unitary $U$ such that $U\rho U^\dagger \notin \CNC^{(m)}$.  And then by the previous argument for the non-absolute case, such a state can be single-copy reduced to an $m$-qubit magic state.
And again similar to the non-absolute case, states within $\ACNC^{(m)}$ however may or not be reducible to an $m$-qubit magic state using such protocols.

\subsection{Consequences of Conjecture 1}

We now discuss the relationship between the notions developed in this section and Conjecture~\ref{Conjecture:LambdaMajorization}, with emphasis on the consequences that would follow should the conjecture hold true.  The most immediate would be that $\ASTAB = \mathrm{ACNC}$.  This is essentially a restatement of Theorem~\ref{Theorem:QubitConjecturedCharacterization}.  Practically, this would provide a dramatically more efficient way of characterizing absolute stabilizerness: $n$ inequalities Eq.~\eqref{CNC-ratio-inequalities} for $n$ qubits.

Another would be that our Conjecture~\ref{Conjecture:LambdaOutradius} on the shared inradius of STAB and $\ASTAB$ (equivalently Conjecture 1 in~\cite{LiuRoth2025}) would naturally follow.  This can be quickly seen as follows.  If the set of CNC vertices $\{A_\Omega^\gamma \in \Lambda_{\mathrm{CNC}}\}$ is spectrally generating, then for any other vertex $V \in \Lambda$, there must exist a probability distribution $p(\Omega,\gamma)$ such that
\begin{equation}
    \lambda^\downarrow(V) \prec \sum_{(\Omega,\gamma)}p(\Omega,\gamma)\lambda^\downarrow(A_\Omega^\gamma).
\end{equation}
The purity functional $||\cdot ||_2^2$ on spectra is Schur-convex (i.e.\ is a monotone on the majorization pre-order) and so we have
\begin{align}
    ||\lambda^\downarrow(V)||_2^2 &\leq \bigg|\bigg|\sum_{(\Omega,\gamma)} p(\Omega,\gamma) \lambda^\downarrow(A_\Omega^\gamma)\bigg|\bigg|_2^2 \\
    &\leq \sum_{(\Omega,\gamma)} p(\Omega,\gamma) ||\lambda^\downarrow(A_\Omega^\gamma)||_2^2 \leq 2,
\end{align}
where in the second line we used Jensen's inequality and in the third line we used Lemma~\ref{Lemma:CNCRadius}, a consequence of which says that all CNC vertices $A \in \Lambda_{\mathrm{CNC}}$ have a purity less than or equal to 2.  Hence, if Conjecture~\ref{Conjecture:LambdaMajorization} holds, then Conjecture~\ref{Conjecture:LambdaOutradius} holds as well, which by polar duality implies that the inradius of STAB and $\ASTAB$ is $1/\sqrt{2^n(2^{n+1}-1)}$; see Theorem~\ref{Theorem:StabilizerInradius}.

A third motivation for finding a proof of Conjecture~\ref{Conjecture:LambdaMajorization} is to provide more evidence towards a deeper connection between CNC operators, magic witnessing in mixed states, and single-copy stabilizer reductions to magic states.  Indeed, the previous subsections have argued that CNC structures, be it either the individual type $\CNC^{(m)}$ or cumulative type $\CNC^{(\leq m)}$, are natural generalizations of the Triangle Criterion, and that states outside these outer approximations to STAB are always reducible to a magic state via single-copy stabilizer protocols.  The unitarily-invariant versions have analogous interpretations, and Conjecture~\ref{Conjecture:LambdaMajorization} would give credence to the idea that the CNC and ACNC hierarchies have both operational and foundational significance.  Moreover, such a connection between magic witnessing and CNC operators would be conceptually satisfying as such operators are inherently based on non-contextual value assignments of Pauli observables, a well-established notion of classicality in quantum computation~\cite{Raussendorf2013,HowardEmerson2014,BermejoVegaRaussendorf2017,DelfosseRaussendorf2017}.

We end by noting that if, on the other hand, there was a strict inclusion of $\ASTAB$ within ACNC, this would imply the existence of a new and somewhat strange type of bound magic.  In particular, it would imply the existence of states that are not absolutely stabilizer, including therefore certain magic states, that cannot be reduced to a magic state outside of ACNC via arbitrary unitaries and arbitrary subsystem post-selection.  And since we have already shown that $\ASTAB = \ACNC$ for one and two qubits, such states can only exist in systems of three qubits or higher.  Moreover, if one started with one of these hypothetical magic states and further post-selected all but one or two qubits onto a pure magic state, this would counter-intuitively result in an absolutely stabilizer state.  This is again because ACNC is closed under post-selection and because $\ACNC$ = $\ASTAB$ for these dimensions.

\section{Discussion}\label{Section:Discussion}

Here we characterized the set of absolutely stabilizer states for $n$ qudits as a polytope in the simplex of state spectra.  This was done by analyzing the vertices of $\Lambda$, the polar dual of the stabilizer polytope, to create a set of linear constraints on the eigenvalues of a general mixed state.  Restricting to prime local dimension, we then focused on the closed and noncontextual (CNC) vertices of $\Lambda$ to develop Conjecture~\ref{Conjecture:LambdaMajorization} on finding a facet-minimal description of the multi-qubit muggle polytope.  This conjecture essentially says that only the CNC vertices, which are a tiny subset of $\mathrm{vert}(\Lambda)$, are required.  This was explicitly proved in low-dimensional cases, and evidence for its validity in larger systems was presented in the appendix.  We then applied our framework to absolute Wigner positivity to introduce the AWP polytopes in the simplex of spectra.  These had a simpler structure due to the smaller number of facets to the Wigner polytope relative to the stabilizer polytope.  The inradius of these balls, which amount to purity-based sufficient conditions for state inclusion, were then studied.  This resulted in complete results in odd dimensions and in Conjecture~\ref{Conjecture:LambdaOutradius} for qubits. We then made a comparison to other known balls, namely those of absolute separability and positive semi-definiteness, which also served as a preliminary comparison between magic and entanglement in the unitarily-invariant setting.  Finally, we ended by further developing the CNC related structures discussed throughout. This led to a natural hierarchy thereof and a generalization of the Triangle Criterion \cite{LiuRoth2025}, with precisely the same structure used being central to our Conjecture \ref{Conjecture:LambdaMajorization}.   

This work opens up several lines of future research, the most immediate of which being to advance Conjecture~\ref{Conjecture:LambdaMajorization} on the sufficiency of the CNC operators in the characterization of the multi-qubit absolutely stabilizer spectral polytope.  This will require a better understanding of how the CNC$^{(m)}$ polytopes (i.e.~the dual of the convex hull of the CNC vertices of type $m$) relate to the remaining vertices.  As shown in Section~\ref{Section:ACNC}, such operators are intimately linked to the single-copy reduction of general mixed states to $m$-qubit magic states via stabilizer operations.  We showed that such CNC polytopes can be seen as natural generalizations of the Triangle Criterion~\cite{LiuRoth2025}, but only half of the operational interpretation survives due the higher structure not present in the case of a single qubit.  Namely, it remains open the question whether there is a simple criterion that decides which states within $\CNC^{(\leq m)}$ can and cannot be single-copy-reduced to an $m$-qubit magic state. A related avenue is to explore the feasibility of this conjecture for qudits with odd-prime dimensions.  This appears to be even more difficult than the qubit case due to the globally non-linear value assignments unique to this setting.

The other primary follow-up is to advance Conjecture~\ref{Conjecture:LambdaOutradius} on the circumradius of $\Lambda$ for qubits, or, equivalently, the inradius of the qubit stabilizer polytope.  Such a result would amount to a sufficient and tight test for absolute stabilizerness based only on purity.  This conjecture is moreover equivalent to Conjecture 1 of Ref.~\cite{LiuRoth2025}, where the same purity threshold is derived via the Triangle Criterion, a magic analogue of the PPT test.  This suggests that the structure of the $\Lambda$ polytope may provide insight to related protocols on magic state witnessing and single-copy stabilizer reductions. Finding the circumradius of the set of absolutely stabilizer states is also an open problem, which would provide a tight necessary condition for absolute stabilizerness. Relatedly, finding non-optimal but more accessible necessary and sufficient conditions for $\ASTAB$ or AWP membership would be desirable.  This could for example be explored using higher moments $\Tr[\rho^{\ell}]$ of a quantum state, beyond the case of purity, i.e. $\ell >2$. In any case, it would moreover be of theoretical interest to give a detailed comparison to the set of absolutely separable states.  For example, if the multi-qubit $\ASTAB$ circumradius is bigger than the absolutely stabilizer inball, this would give strong evidence towards there existing entangled absolutely stabilizer multi-qubit states. Conclusively relating entanglement to absolute stabilizerness or absolute Wigner-positivity is left open for future work.

We also mention the possibility of extending our framework to quantum channels via the Choi matrix representation~\cite{MariEisert2012, WangSu2019, SeddonCampbell2019, HeimendahlGross2022}.  This would entail the characterization of ``completely muggle-preserving'' channels as those that preserve the set of absolutely stabilizer states regardless of any coupling to ancilla; see also~\cite{PatraSenDe2024}.  Associated to this would be a resource theory of unitarily-accessible magic, analogous to~\cite{SalazarJunior2022, PatraSenDe2023}. Such a resource theory could have interesting features as $\ASTAB$ is additionally closed under partial trace, and arbitrary unitaries conditioned on outcomes of arbitrary destructive measurements.

\section*{Acknowledgments}
We thank Ernesto Galv\~{a}o, Mark Howard, and Ulysse Chabaud for insightful discussions.  We also thank Zhenhuan Liu, Tobias Haug, and Ingo Roth for discussions on Ref.~\cite{LiuRoth2025}.  M.Z.~is funded by the Natural Sciences and Engineering Research Council of Canada (NSERC). J.D.~acknowledges funding from the European Union’s Horizon Europe Framework Programme (EIC Pathfinder Challenge project Veriqub) under Grant Agreement No.~101114899. This project grew out of discussions that took place at the Algebraic Structures in Quantum Computation Workshop (ASQC6) at Leibniz University Hannover in 2025. We thank Robert Raussendorf and the rest of the organizing committee for the invitation.

\bibliographystyle{apsrev4-2}
\bibliography{biblio}

@article{KyFan1949,
    title = {{On a Theorem of Weyl Concerning Eigenvalues of Linear Transformations I}},
    author = {Ky Fan},
    journal = {Proceedings of the National Academy of Sciences},
    volume = {35},
    number = {11},
    pages = {652--655},
    year = {1949},
    doi = {10.1073/pnas.35.11.652}
}

@article{KyFan1950,
    title = {{On a Theorem of Weyl Concerning Eigenvalues of Linear Transformations II}},
    author = {Ky Fan},
    journal = {Proceedings of the National Academy of Sciences},
    volume = {36},
    number = {1},
    pages = {31--35},
    year = {1950},
    doi = {10.1073/pnas.36.1.31}
}

@article{Harriman1978,
    title = {Geometry of density matrices. I. Definitions, $N$ matrices and 1 matrices},
    author = {Harriman, John E.},
    journal = {Physical Review A},
    volume = {17},
    number = {4},
    pages = {1249--1256},
    year = {1978},
    doi = {10.1103/PhysRevA.17.1249}
}

@article{Wootters1987,
    title = {{A Wigner-function formulation of finite-state quantum mechanics}},
    author = {William K. Wootters},
    journal = {Annals of Physics},
    volume = {176},
    number = {1},
    pages = {1--21},
    year = {1987},
    doi = {10.1016/0003-4916(87)90176-X}
}

@article{GritzmannKlee1992,
    title = {{Inner and outer j-radii of convex bodies in finite-dimensional normed spaces}},
    author = {Peter Gritzmann and Victor Klee},
    journal = {Discrete \& Computational Geometry},
    volume = {7},
    pages = {255--280},
    year = {1992},
    doi = {10.1007/BF02187841}
}

@article{Mermin1993,
    title = {{Hidden variables and the two theorems of John Bell}},
    author = {N. David Mermin},
    journal = {Reviews of Modern Physics},
    volume = {65},
    number = {3},
    pages = {803--815},
    year = {1993},
    doi = {10.1103/RevModPhys.65.803},
    eprint = {arXiv:1802.10119}
}

@book{Ziegler1995,
    title = {{Lectures on Polytopes}},
    author = {G\"unter M. Ziegler},
    publisher = {Springer},
    address = {New York},
    year = {1995},
    doi = {10.1007/978-1-4613-8431-1}
}

@inproceedings{FukudaProdon1996,
    title = {{Double description method revisited}},
    author = {Komei Fukuda and Alain Prodon},
    booktitle = {Combinatorics and Computer Science},
    editor = {Michel Deza and Reinhardt Euler and Ioannis Manoussakis},
    publisher = {Springer Berlin Heidelberg},
    address = {Berlin, Heidelberg},
    pages = {91--111},
    year = {1996}
}

@article{HorodeckiHorodecki1996,
    title = {Separability of mixed states: necessary and sufficient conditions},
    author = {Horodecki, Micha\l{} and Horodecki, Pawe\l{} and Horodecki, Ryszard},
    journal = {Physics Letters A},
    volume = {223},
    number = {1–2},
    pages = {1--8},
    year = {1996},
    doi = {10.1016/S0375-9601(96)00706-2},
    eprint = {arXiv:quant-ph/9605038}
}

@article{Peres1996,
    title = {Separability Criterion for Density Matrices},
    author = {Peres, Asher},
    journal = {Physical Review Letters},
    volume = {77},
    number = {8},
    pages = {1413--1415},
    year = {1996},
    doi = {10.1103/PhysRevLett.77.1413},
    eprint = {arXiv:quant-ph/9604005}
}

@phdthesis{Gottesman1997,
	title = {{Stabilizer codes and quantum error correction}},
	author = {Daniel Gottesman},
	school = {California Institute of Technology},
	year = {1997},
	eprint = {arXiv:quant-ph/9705052}
}

@article{ZyczkowskiLewenstein1998,
    title = {Volume of the set of separable states},
    author = {\.{Z}yczkowski, Karol and Horodecki, Pawe\l{} and Sanpera, Anna and Lewenstein, Maciej},
    journal = {Physical Review A},
    volume = {58},
    number = {2},
    pages = {883--892},
    year = {1998},
    doi = {10.1103/PhysRevA.58.883},
    eprint = {arXiv:quant-ph/9804024}
}

@inproceedings{Gottesman1999,
    title = {{Fault-Tolerant Quantum Computation with Higher-Dimensional Systems}},
    author = {Daniel Gottesman},
    booktitle = {NASA International Conference on Quantum Computing and Quantum Communications},
    editor = {Colin P. Williams},
    year = {1999},
    publisher = {Springer Berlin Heidelberg},
    address = {Berlin, Heidelberg},
    pages = {302--313},
    doi = {10.1007/3-540-49208-9_27},
    eprint = {arXiv:quant-ph/9802007}
}

@article{GottesmanChuang1999,
    title = {{Demonstrating the viability of universal quantum computation using teleportation and single-qubit operations}},
    author = {Daniel Gottesman and Isaac L. Chuang},
    journal = {Nature},
    volume = {402},
    number = {6760},
    pages = {390--393},
    year={1999},
    doi = {10.1038/46503},
    eprint = {arXiv:quant-ph/9908010}
}

@article{KusZyczkowski2001,
    title = {{Geometry of entangled states}},
    author = {Marek Kus and Karol \.Zyczkowski},
    journal = {Physical Review A},
    volume = {63},
    pages = {032307},
    year = {2001},
    doi = {10.1103/PhysRevA.63.032307},
    eprint = {arXiv:quant-ph/0006068}
}

@article{GurvitsBarnum2002,
    title = {{Largest separable balls around the maximally mixed bipartite quantum state}},
    author = {Leonid Gurvits and Howard Barnum},
    journal = {Physical Review A},
    volume = {66},
    pages = {062311},
    year = {2002},
    doi = {10.1103/PhysRevA.66.062311},
    eprint = {arXiv:quant-ph/0204159}
}

@book{Grunbaum2003,
    title = {Convex Polytopes},
    author = {Branko Gr\"unbaum},
    publisher = {Springer},
    address = {New York},
    year = {2003},
    doi = {10.1007/978-1-4613-0019-9}
}

@article{GurvitsBarnum2003,
    title = {{Separable balls around the maximally mixed multipartite quantum states}},
    author = {Leonid Gurvits and Howard Barnum},
    journal = {Physical Review A},
    volume = {68},
    pages = {042312},
    year = {2003},
    doi = {10.1103/PhysRevA.68.042312},
    eprint = {arXiv:quant-ph/0302102}
}

@article{AaronsonGottesman2004,
    title = {{Improved simulation of stabilizer circuits}},
    author = {Scott Aaronson and Daniel Gottesman},
    journal = {Physical Review A},
    volume = {70},
    pages = {052328},
    year = {2004},
    doi = {10.1103/PhysRevA.70.052328},
    eprint = {arXiv:quant-ph/0406196}
}

@article{Hwang2004,
    title = {Cauchy’s Interlace Theorem for Eigenvalues of Hermitian Matrices},
    author = {Hwang, Suk-Geun},
    journal = {The American Mathematical Monthly},
    volume = {111},
    number = {2},
    pages = {157--159},
    year = {2004},
    doi = {10.1080/00029890.2004.11920060}
}

@article{Appleby2005,
    title = {{Symmetric informationally complete–positive operator valued measures and the extended Clifford group}},
    author = {D. M. Appleby},
    journal = {Journal of Mathematical Physics},
    volume = {46},
    number = {5},
    pages = {052107},
    year = {2005},
    doi = {10.1063/1.1896384},
    eprint = {arXiv:quant-ph/0412001}
}

@article{BravyiKitaev2005,
    title = {{Universal quantum computation with ideal {Clifford} gates and noisy ancillas}},
    author = {Sergey Bravyi and Alexei Kitaev},
    journal = {Physical Review A},
    volume = {71},
    pages = {022316},
    year = {2005},
    doi = {10.1103/PhysRevA.71.022316},
    eprint = {arXiv:quant-ph/0403025}
}

@article{Chin2005,
    title = {{On non-commuting sets in an extraspecial p-group}},
    author = {Angelina Y. M. Chin},
    journal = {Journal of group theory},
    volume = {8},
    number = {2},
    pages = {189--194},
    year = {2005},
    doi = {10.1515/jgth.2005.8.2.189}
}

@article{GurvitsBarnum2005,
    title = {{Better bound on the exponent of the radius of the multipartite separable ball}},
    author = {Leonid Gurvits and Howard Barnum},
    journal = {Physical Review A},
    volume = {72},
    pages = {032322},
    year = {2005},
    doi = {10.1103/PhysRevA.72.032322},
    eprint = {arXiv:quant-ph/0409095}
}

@book{BengtssonZyczkowski2006,
    title = {{Geometry of Quantum States}},
    author = {Ingemar Bengtsson and Karol \.Zyczkowski},
    publisher = {Cambridge University Press},
    address = {Cambridge},
    year = {2006},
    doi = {10.1017/CBO9780511535048}
}

@article{Gross2006,
    title = {{Hudson’s theorem for finite-dimensional quantum systems}},
    author = {David Gross},
    journal = {Journal of Mathematical Physics},
    volume = {47},
    pages = {122107},
    year = {2006},
    doi = {10.1063/1.2393152},
    eprint = {arXiv:quant-ph/0602001}
}

@book{NocedalWright2006,
    title = {Numerical Optimization},
    author = {Jorge Nocedal and Stephen J. Wright},
    publisher = {Springer},
    address = {New York},
    year = {2006},
    doi = {10.1007/978-0-387-40065-5}
}

@article{GrossEisert2007,
    title = {{Evenly distributed unitaries: On the structure of unitary designs}},
    author = {D. Gross and K. Audenaert and J. Eisert},
    journal = {Journal of Mathematical Physics},
    volume = {48},
    number = {5},
    pages = {052104},
    year = {2007},
    doi = {10.1063/1.2716992},
    eprint = {arXiv:quant-ph/0611002},
}

@article{Hildebrand2007,
    title = {{Entangled states close to the maximally mixed state}},
    author = {Roland Hildebrand},
    journal = {Physical Review A},
    volume = {75},
    pages = {062330},
    year = {2007},
    doi = {10.1103/PhysRevA.75.062330},
    eprint = {arXiv:quant-ph/0702040}
}

@article{ApplebyChaturvedi2008,
    title = {{Spectra of phase point operators in odd prime dimensions and the extended Clifford group}},
    author = {D. M. Appleby and Ingemar Bengtsson and S. Chaturvedi},
    journal = {Journal of Mathematical Physics},
    volume = {49},
    pages = {012102},
    year = {2008},
    doi = {10.1063/1.2824479},
    eprint = {arXiv:0710.3013}
}

@phdthesis{Gross2008,
    title = {{Computational power of quantum many-body states and some results on discrete phase spaces}},
    author = {David Gross},
    school = {Imperial College London},
    year = {2008},
    url = {https://www.thp.uni-koeln.de/gross/files/diss.pdf}
}

@inbook{CampbellBrowne2009,
    title = {{On the Structure of Protocols for Magic State Distillation}},
    author = {Campbell, Earl T. and Browne, Dan E.},
    booktitle = {Theory of Quantum Computation, Communication, and Cryptography},
    editor = {Childs, Andrew and Mosca, Michele},
    publisher = {Springer Berlin Heidelberg},
    address = {Berlin, Heidelberg},
    pages = {20--32},
    year = {2009},
    doi = {10.1007/978-3-642-10698-9_3},
    eprint = {arXiv:0908.0838}
}

@article{EastinKnill2009,
    title = {{Restrictions on Transversal Encoded Quantum Gate Sets}},
    author = {Bryan Eastin and Emanuel Knill},
    journal = {Physical Review Letters},
    volume = {102},
    pages = {110502},
    year = {2009},
    doi = {10.1103/PhysRevLett.102.110502},
    eprint = {arXiv:0811.4262}
}

@article{Reichardt2009,
    title = {{Quantum universality by state distillation}},
    author = {Ben W. Reichardt},
    journal = {{Quantum Information and Computation}},
    volume = {9},
    number = {11\&12},
    pages = {1030--1052},
    year = {2009},
    doi = {10.26421/QIC9.11-12-7},
    eprint = {arXiv:quant-ph/0608085}
}

@article{CampbellBrowne2010,
    title = {{Bound States for Magic State Distillation in Fault-Tolerant Quantum Computation}},
    author = {Earl T. Campbell and Dan E. Browne},
    journal = {Physical Review Letters},
    volume = {104},
    pages = {030503},
    year = {2010},
    doi = {10.1103/PhysRevLett.104.030503},
    eprint = {arXiv:0908.0836}
}

@book{NariciBeckenstein2010,
    title = {{Topological Vector Spaces}},
    author = {Lawrence Narici and Edward Beckenstein},
    publisher = {Chapman \& Hall / CRC},
    address = {London},
    year = {2010}
}

@book{MarshallArnold2011,
    title = {{Inequalities: Theory of Majorization and Its Applications}},
    author = {Albert W. Marshall and Ingram Olkin and Barry C. Arnold},
    publisher = {Springer},
    address = {New York},
    edition = {2nd},
    year = {2011},
    doi = {10.1007/978-0-387-68276-1}
}

@article{MariEisert2012,
    title = {{Positive Wigner Functions Render Classical Simulation of Quantum Computation Efficient}},
    author = {A. Mari and J. Eisert},
    journal = {Physical Review Letters},
    volume = {109},
    pages = {230503},
    year = {2012},
    doi = {10.1103/PhysRevLett.109.230503},
    eprint = {arXiv:1208.3660}
}

@article{VeitchEmerson2012,
    title = {{Negative quasi-probability as a resource for quantum computation}},
    author = {Victor Veitch and Christopher Ferrie and David Gross and Joseph Emerson},
    journal = {New Journal of Physics},
    volume = {14},
    pages = {113011},
    year = {2012},
    doi = {10.1088/1367-2630/14/11/113011},
    eprint = {arXiv:1201.1256}
}

@article{DeBeaudrap2013, 
    title = {{A linearized stabilizer formalism for systems of finite dimension}},
    author = {Niel De Beaudrap},
    journal = {Quantum Information and Computation},
    volume = {13},
    number = {1\&2},
    pages = {73--115},
    year = {2013},
    doi = {10.26421/QIC13.1-2-6},
    eprint = {arXiv:1102.3354}
}

@article{GrossWalter2013, 
    title = {{Stabilizer information inequalities from phase space distributions}},
    author = {David Gross and Michael Walter},
    journal = {Journal of Mathematical Physics},
    volume = {54},
    number = {8},
    pages = {082201},
    year = {2013},
    doi = {10.1063/1.4818950},
    eprint = {arXiv:1302.6990}
}

@article{Raussendorf2013,
    title = {Contextuality in measurement-based quantum computation},
    author = {Raussendorf, Robert},
    journal = {Physical Review A},
    volume = {88},
    pages = {022322},
    year = {2013},
    doi = {10.1103/PhysRevA.88.022322},
    eprint = {arXiv:0907.5449}
}

@article{HowardEmerson2014,
    title = {Contextuality supplies the {`magic'} for quantum computation},
    author = {Mark Howard and Joel Wallman and Victor Veitch and Joseph Emerson},
    journal = {Nature},
    volume = {510},
    number = {7505},
    pages = {351--355},
    year = {2014},
    doi = {10.1038/nature13460},
    eprint = {arXiv:1401.4174}
}

@article{VeitchEmerson2014,
    title = {{The resource theory of stabilizer quantum computation}},
    author = {Victor Veitch and SA Hamed Mousavian and Daniel Gottesman and Joseph Emerson},
    journal = {New Journal of Physics},
    volume = {16},
    pages = {013009},
    year = {2014},
    doi = {10.1088/1367-2630/16/1/013009},
    eprint = {arXiv:1307.7171}
}

@article{PashayanBartlett2015,
    title = {Estimating outcome probabilities of quantum circuits using quasiprobabilities},
    author = {Hakop Pashayan and Joel J Wallman and Stephen D Bartlett},
    journal = {Physical Review Letters},
    volume = {115},
    number = {7},
    pages = {070501},
    year = {2015},
    doi = {10.1103/PhysRevLett.115.070501},
    eprint = {arXiv:1503.07525}
}

@article{BravyiGosset2016,
    title = {{Improved Classical Simulation of Quantum Circuits Dominated by Clifford Gates}},
    author = {Sergey Bravyi and David Gosset},
    journal = {Physical Review Letters},
    volume = {116},
    pages = {250501},
    year = {2016},
    doi = {10.1103/PhysRevLett.116.250501},
    eprint = {arXiv:1601.07601}
}

@article{BravyiSmolin2016,
    title = {{Trading Classical and Quantum Computational Resources}},
    author = {Sergey Bravyi and Graeme Smith and John A. Smolin},
    journal = {Physical Review X},
    volume = {6},
    pages = {021043},
    year = {2016},
    doi = {10.1103/PhysRevX.6.021043},
    eprint = {arXiv:1506.01396}
}

@article{Webb2016,
    title = {{The Clifford group forms a unitary 3-design}},
    author = {Zak Webb},
    journal = {Quantum Information and Computation},
    volume = {16},
    pages = {1379--1400},
    year = {2016},
    doi = {10.26421/QIC16.15-16-8},
    eprint = {arXiv:1510.02769}
}

@article{BermejoVegaRaussendorf2017,
    title = {Contextuality as a Resource for Models of Quantum Computation with Qubits},
    author = {Bermejo-Vega, Juan and Delfosse, Nicolas and Browne, Dan E. and Okay, Cihan and Raussendorf, Robert},
    journal = {Physical Review Letters},
    volume = {119},
    pages = {120505},
    year = {2017},
    doi = {10.1103/PhysRevLett.119.120505},
    eprint = {arXiv:1610.08529}
}

@article{DelfosseRaussendorf2017,
    title = {{Equivalence between contextuality and negativity of the Wigner function for qudits}},
    author = {Nicolas Delfosse and Cihan Okay and Juan Bermejo-Vega and Dan E Browne and Robert Raussendorf},
    journal = {New Journal of Physics},
    volume = {19},
    pages = {123024},
    year = {2017},
    doi = {10.1088/1367-2630/aa8fe3},
    eprint = {arXiv:1610.07093}
}

@article{HowardCampbell2017,
    title = {{Application of a Resource Theory for Magic States to Fault-Tolerant Quantum Computing}},
    author = {Mark Howard and Earl Campbell},
    journal = {Physical Review Letters},
    volume = {118},
    pages = {090501},
    year = {2017},
    doi = {10.1103/PhysRevLett.118.090501},
    eprint = {arXiv:1609.07488}
}

@article{OkayRaussendorf2017,
    title = {{Topological proofs of contextuality in quantum mechanics}},
    author = {Cihan Okay and Sam Roberts and Stephen D. Bartlett and Robert Raussendorf},
    journal = {Quantum Information and Computation},
    volume = {17},
    number = {13\&14},
    pages = {1135--1166},
    year = {2017},
    doi = {10.26421/QIC17.13-14-5},
    eprint = {arXiv:1701.01888}
}

@article{Zhu2017,
    title = {{Multiqubit Clifford groups are unitary 3-designs}},
    author = {Huangjun Zhu},
    journal = {Physical Review A},
    volume = {96},
    pages = {062336},
    year = {2017},
    doi = {10.1103/PhysRevA.96.062336},
    eprint = {arXiv:1510.02619}
}

@article{BravyiHoward2019,
    title = {{Simulation of quantum circuits by low-rank stabilizer decompositions}},
    author = {Sergey Bravyi and Dan Browne and Padraic Calpin and Earl Campbell and David Gosset and Mark Howard},
    journal = {{Quantum}},
    volume = {3},
    pages = {181},
    year = {2019},
    doi = {10.22331/q-2019-09-02-181},
    eprint = {arXiv:1808.00128}
}

@mastersthesis{Heimendahl2019,
    title = {{The stabilizer polytope and contextuality for qubit systems}},
    author = {Arne Heimendahl},
    school = {University of Cologne},
    year = {2019},
    url = {https://www.mi.uni-koeln.de/opt/wp-content/uploads/2020/07/MT_Arne_Heimendahl.pdf}
}

@article{SeddonCampbell2019,
    title = {{Quantifying magic for multi-qubit operations}},
    author = {James R. Seddon and Earl T. Campbell},
    journal = {Proceedings of the Royal Society A},
    volume = {475},
    number = {2227},
    pages = {20190251},
    year = {2019},
    doi = {10.1098/rspa.2019.0251},
    eprint = {arXiv:1901.03322}
}

@article{WangSu2019,
    title = {Quantifying the magic of quantum channels},
    author = {Xin Wang and Mark M Wilde and Yuan Su},
    journal = {New Journal of Physics},
    volume = {21},
    number = {10},
    pages = {103002},
    year = {2019},
    doi = {10.1088/1367-2630/ab451d},
    eprint = {arXiv:1903.04483}
}

@article{AbgaryanTorosyan2020,
    title = {{The Global Indicator of Classicality of an Arbitrary $N$-Level Quantum System}},
    author = {Vahagn Abgaryan and Arsen Khvedelidze and Astghik Torosyan},
    journal = {Journal of Mathematical Sciences},
    volume = {251},
    number = {3},
    pages = {301--314},
    year = {2020},
    doi = {10.1007/s10958-020-05092-6},
    eprint = {arXiv:1910.11220}
}

@article{PrakashGupta2020,
    title = {{Contextual bound states for qudit magic state distillation}},
    author = {Shiroman Prakash and Aashi Gupta},
    journal = {Physical Review A},
    volume = {101},
    pages = {010303},
    year = {2020},
    doi = {10.1103/PhysRevA.101.010303},
    eprint = {arXiv:1905.00392}
}

@article{RaussendorfZurel2020,
    title = {{Phase-space-simulation method for quantum computation with magic states on qubits}},
    author = {Robert Raussendorf and Juani Bermejo-Vega and Emily Tyhurst and Cihan Okay and Michael Zurel},
    journal = {Physical Review A},
    volume = {101},
    pages = {012350},
    year = {2020},
    doi = {10.1103/PhysRevA.101.012350},
    eprint = {arXiv:1905.05374}
}

@article{ZurelRaussendorf2020,
    title = {{Hidden Variable Model for Universal Quantum Computation with Magic States on Qubits}},
    author = {Michael Zurel and Cihan Okay and Robert Raussendorf},
    journal = {Physical Review Letters},
    volume = {125},
    pages = {260404},
    year = {2020},
    doi = {10.1103/PhysRevLett.125.260404},
    eprint = {arXiv:2004.01992}
}

@article{SeddonCampbell2021,
    title = {{Quantifying Quantum Speedups: Improved Classical Simulation From Tighter Magic Monotones}},
    author = {James R. Seddon and Bartosz Regula and Hakop Pashayan and Yingkai Ouyang and Earl T. Campbell},
    journal = {PRX Quantum},
    volume = {2},
    pages = {010345},
    year = {2021},
    doi = {10.1103/PRXQuantum.2.010345},
    eprint = {arXiv:2002.06181}
}

@article{HeimendahlGross2022, 
    title = {{The axiomatic and the operational approaches to resource theories of magic do not coincide}},
    author = {Arne Heimendahl and Markus Heinrich and David Gross},
    journal = {Journal of Mathematical Physics},
    volume = {63},
    number = {11},
    pages = {112201},
    year = {2022},
    doi = {10.1063/5.0085774},
    eprint = {arXiv:2011.11651}
}

@misc{SalazarJunior2022,
    title = {{Resource theory of Absolute Negativity}},
    author = {Roberto Salazar and Jakub Czartowski and A. de Oliveira Junior},
    year = {2022},
    eprint = {arXiv:2205.13480}
}

@misc{KimAbramsky2023,
    title = {{State-independent all-versus-nothing arguments}},
    author = {Boseong Kim and Samson Abramsky},
    year = {2023},
    eprint = {arXiv:2311.11218}
}

@article{PatraSenDe2023,
    title = {Resource theory of nonabsolute separability},
    author = {Patra, Ayan and Maity, Arghya and Sen(De), Aditi},
    journal = {Physical Review A},
    volume = {108},
    pages = {042402},
    year = {2023},
    doi = {10.1103/PhysRevA.108.042402},
    eprint={arXiv:2212.11105}
}

@article{RaussendorfFeldmann2023,
    title = {{The role of cohomology in quantum computation with magic states}},
    author = {Robert Raussendorf and Cihan Okay and Michael Zurel and Polina Feldmann},
    journal = {{Quantum}},
    volume = {7},
    pages = {979},
    year = {2023},
    doi = {10.22331/q-2023-04-13-979},
    eprint = {arXiv:2110.11631}
}

@article{SingalHsieh2023,
    title = {{Counting stabiliser codes for arbitrary dimension}},
    author = {Tanmay Singal and Che Chiang and Eugene Hsu and Eunsang Kim and Hsi-Sheng Goan and Min-Hsiu Hsieh},
    journal = {{Quantum}},
    volume = {7},
    pages = {1048},
    year = {2023},
    doi = {10.22331/q-2023-07-06-1048},
    eprint = {arXiv:2209.01449}
}

@article{DenisMartin2024,
    title = {{Polytopes of absolutely Wigner bounded spin states}},
    author = {J\'er\^ome Denis and Jack Davis and Robert B. Mann and John Martin},
    journal = {{Quantum}},
    volume = {8},
    pages = {1550},
    year = {2024},
    doi = {10.22331/q-2024-12-04-1550},
    eprint = {arXiv:2304.09006}
}

@misc{PatraSenDe2024,
    title = {{Qubit magic-breaking channels}},
    author = {Ayan Patra and Rivu Gupta and Alessandro Ferraro and Aditi Sen(De)},
    year = {2024},
    eprint = {arXiv:2409.04425}
}

@article{ZurelHeimendahl2024,
    title = {{Hidden variable model for quantum computation with magic states on qudits of any dimension}},
    author = {Michael Zurel and Cihan Okay and Robert Raussendorf and Arne Heimendahl},
    journal = {{Quantum}},
    volume = {8},
    pages = {1323},
    year = {2024},
    doi = {10.22331/q-2024-04-30-1323},
    eprint = {arXiv:2110.12318}
}

@misc{ZurelHeimendahl2024b,
    title = {{Efficient classical simulation of quantum computation beyond Wigner positivity}},
    author = {Michael Zurel and Arne Heimendahl},
    year = {2024},
    eprint = {arXiv:2407.10349}
}

@article{WenKempf2025,
    title = {{Separable ellipsoids around multipartite states}},
    author = {Robin Y. Wen and Gilles Parez and Liuke Lyu and William Witczak-Krempa and Achim Kempf},
    journal = {Physical Review A},
    volume = {112},
    pages = {012426},
    year = {2025},
    doi = {10.1103/tyhj-1wlq},
    eprint = {arXiv:2410.05400}
}

@misc{LiuRoth2025,
    title = {{A magic criterion (almost) as nice as PPT, with applications in distillation and detection}},
    author = {Zhenhuan Liu and Tobias Haug and Qi Ye and Zi-Wen Liu and Ingo Roth},
    year = {2025},
    eprint = {arXiv:2512.16777}
}

@article{WagnerGalvao2025, 
    title = {Unitary-invariant method for witnessing nonstabilizerness in quantum processors},
    author = {Wagner, Rafael and Peres, Filipa C R and Zambrini Cruzeiro, Emmanuel and Galv\~{a}o, Ernesto F},
    journal = {Journal of Physics A},
    volume = {58},
    number = {28},
    pages = {285302},
    year = {2025},
    doi = {10.1088/1751-8121/ade9ff},
    eprint = {arXiv:2404.16107}
}

\onecolumngrid
\appendix

\section{Evidence for Conjectures\texorpdfstring{~\ref{Conjecture:LambdaMajorization} and~\ref{Conjecture:LambdaOutradius}}{ }}\label{Section:ConjectureEvidence}

We can verify Conjecture~\ref{Conjecture:LambdaMajorization} and Conjecture~\ref{Conjecture:LambdaOutradius} directly for $n=1$ and $n=2$ qubits. For $n=1$, $\Lambda$ is a cube with vertices $A_{\pm\pm\pm}=(I\pm X\pm Y\pm Z)/2$~\cite{ZurelRaussendorf2020}. These are all CNC operators, and so Conjecture~\ref{Conjecture:LambdaMajorization} holds trivially. Also, $\Tr(A_{\pm\pm\pm}^2)=2$, and so Conjecture~\ref{Conjecture:LambdaOutradius} also holds in this case.

For two qubits, we can enumerate all vertices of $\Lambda$. We find there are $22{,}320$ vertices, which fall into $8$ orbits under the action of the Clifford group with $8$ distinct spectra. These spectra are listed in Table~\ref{Table:TwoQubitLambdaEigenvalues}.

\begin{table}[ht]
    \centering
    \begin{tabular}{|l|cccc|}
        \hline
        Orbit & Exact eigenvalues & & & \\
        \hline\hline
        CNC ($m=1$) & $\frac{1}{2} - \frac{\sqrt{3}}{2}$ & $0$ & $0$ & $\frac{1}{2} + \frac{\sqrt{3}}{2}$ \\
        \hline
        CNC ($m=2$) & $\frac{1}{4} - \frac{\sqrt{5}}{4}$ & $\frac{1}{4} - \frac{\sqrt{5}}{4}$ & $\frac{1}{4} + \frac{\sqrt{5}}{4}$ & $\frac{1}{4} + \frac{\sqrt{5}}{4}$ \\
        \hline
        Row 2 & $-\sqrt{\frac{\sqrt{3}}{8} + \frac{15}{64}} + \frac{\sqrt{3}}{8} + \frac{1}{4}$ & $-\frac{\sqrt{3}}{8} - \sqrt{\frac{15}{64} - \frac{\sqrt{3}}{8}} + \frac{1}{4}$ & $-\frac{\sqrt{3}}{8} + \sqrt{\frac{15}{64} - \frac{\sqrt{3}}{8}} + \frac{1}{4}$ & $\frac{\sqrt{3}}{8} + \frac{1}{4} + \sqrt{\frac{\sqrt{3}}{8} + \frac{15}{64}}$ \\
        \hline
        Row 3 & $-\sqrt{\frac{\sqrt{2}}{8} + \frac{7}{32}} + \frac{\sqrt{2}}{8} + \frac{1}{4}$ & $-\sqrt{\frac{7}{32} - \frac{\sqrt{2}}{8}} - \frac{\sqrt{2}}{8} + \frac{1}{4}$ & $-\frac{\sqrt{2}}{8} + \sqrt{\frac{7}{32} - \frac{\sqrt{2}}{8}} + \frac{1}{4}$ & $\frac{\sqrt{2}}{8} + \frac{1}{4} + \sqrt{\frac{\sqrt{2}}{8} + \frac{7}{32}}$ \\
        \hline
        Row 4 & $\frac{1}{2} - \frac{\sqrt{7}}{4}$ & $0$ & $0$ & $\frac{1}{2} + \frac{\sqrt{7}}{4}$ \\
        \hline
        Row 5 & $-\frac{\sqrt{2}\sqrt{5-\sqrt{5}}}{8} + \frac{1}{4}$ & $\frac{1}{4} + \frac{\sqrt{2}\sqrt{5-\sqrt{5}}}{8}$ & $-\frac{\sqrt{2}\sqrt{5+\sqrt{5}}}{8} + \frac{1}{4}$ & $\frac{1}{4} + \frac{\sqrt{2}\sqrt{5+\sqrt{5}}}{8}$ \\
        \hline
        Row 6 & $\frac{1}{2} - \frac{\sqrt{15}}{6}$ & $0$ & $0$ & $\frac{1}{2} + \frac{\sqrt{15}}{6}$ \\
        \hline
        Row 7 & $-\sqrt{\frac{11\sqrt{5}}{128} + \frac{25}{128}} + \frac{1}{4} + \frac{\sqrt{5}}{8}$ & $-\frac{\sqrt{5}}{8} - \sqrt{\frac{25}{128} - \frac{11\sqrt{5}}{128}} + \frac{1}{4}$ \,& $-\frac{\sqrt{5}}{8} + \sqrt{\frac{25}{128} - \frac{11\sqrt{5}}{128}} + \frac{1}{4}$ \,& $\frac{1}{4} + \frac{\sqrt{5}}{8} + \sqrt{\frac{11\sqrt{5}}{128} + \frac{25}{128}}$ \\
        \hline\hline
        \hline\hline
        Orbit & Approximate decimal expansions & & & \\
        \hline\hline
        CNC ($m=1$) & $-0.36603$ & $\;\;\;0.00000$ & $0.00000$ & $1.36603$\\
        \hline
        CNC ($m=2$) & $-0.30902$ & $-0.30902$ & $0.80902$ & $0.80902$\\
        \hline
        Row 2 & $-0.20497$ & $-0.10018$ & $0.16717$ & $1.13798$\\
        \hline
        Row 3 & $-0.20213$ & $-0.13165$ & $0.27810$ & $1.05569$\\
        \hline
        Row 4 & $-0.16144$ & $\;\;\;0.00000$ & $0.00000$ & $1.16144$\\
        \hline
        Row 5 & $-0.22553$ & $-0.04389$ & $0.54389$ & $0.72553$\\
        \hline
        Row 6 & $-0.14550$ & $\;\;\;0.00000$ & $0.00000$ & $1.14550$\\
        \hline
        Row 7 & $-0.09297$ & $-0.08564$ & $0.02662$ & $1.15198$\\
        \hline
    \end{tabular}
    \caption{Eigenvalues of the eight unitary group orbits (or Clifford group orbits) of the vertices of the $2$-qubit $\Lambda$ polytope. The first two rows are the CNC-type vertices with paramters $m=1$ and $m=2$ according to the classification of Theorem~\ref{Theorem:CNCClassification}. The remaining rows are labeled according to the row in which they appear in Table~2 of Ref.~\cite{Reichardt2009}.}
    \label{Table:TwoQubitLambdaEigenvalues}
\end{table}

The partial sums, $\sum_{\ell=0}^k\lambda_\ell^\downarrow(A_\Omega^\gamma)$, of the two-qubit CNC operators are plotted in Figure~\ref{Figure:TwoQubitLorenzCurves}, along with the corresponding partial sums of eigenvalues, $\sum_{\ell=0}^k\lambda_\ell^\downarrow(A)$ of the remaining two-qubit $\Lambda$-polytope vertex spectra. From this plot we can see that all non-CNC spectra are majorized by the CNC spectrum with parameter $m=1$. Therefore, Conjecture~\ref{Conjecture:LambdaMajorization} holds for $n=2$.

\begin{figure}
    \centering
    \includegraphics[width=0.8\linewidth]{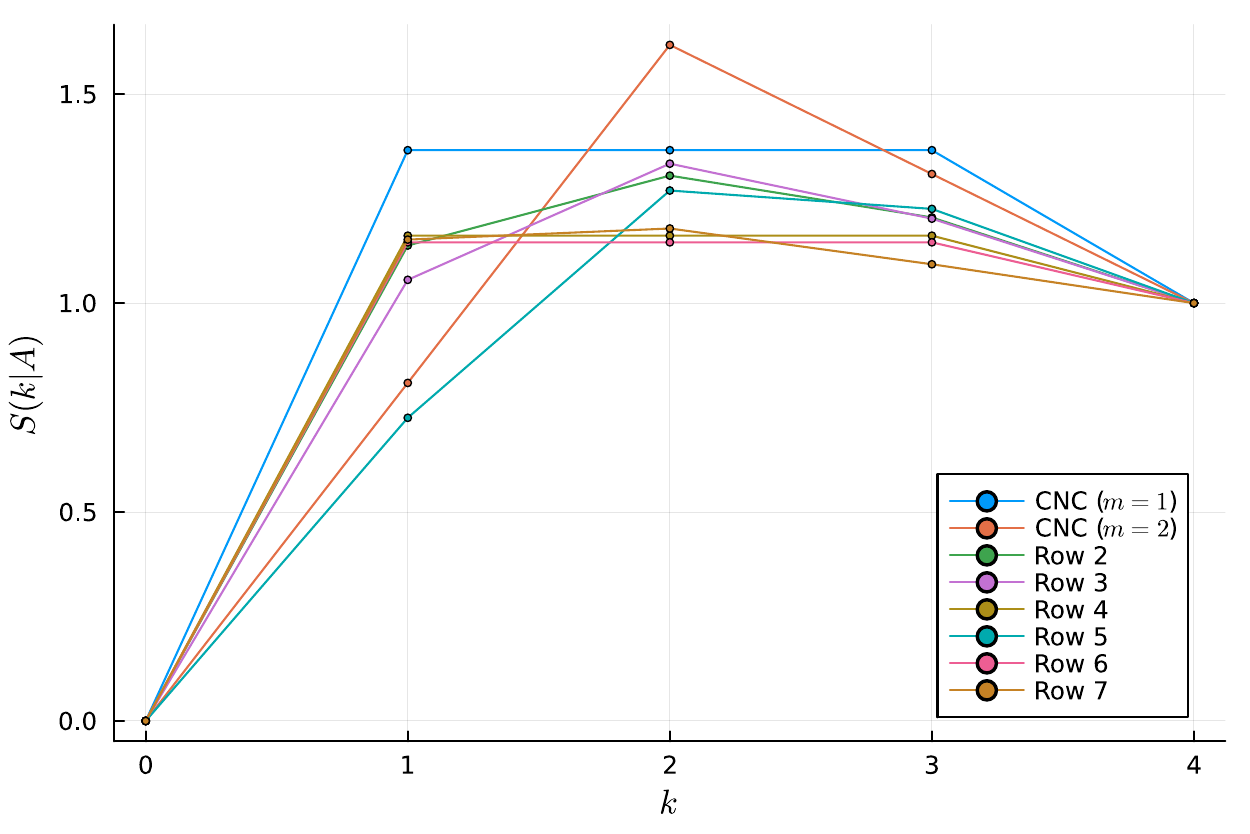}
    \caption{Lorenz curves, $S(k\,|\,A):=\sum_{\ell=1}^k\lambda_k^\downarrow(A)$, of the spectra of two-qubit $\Lambda$ polytope vertices $A$. The first two lines correspond to the eigenvalues of the type $m=1$ and $m=2$ CNC type vertices. The remaining lines are the spectra of the remaining six Clifford group orbits of $2$-qubit $\Lambda$ polytope vertices. The legend refers to the rows of Table~2 in Ref.~\cite{Reichardt2009}.}
    \label{Figure:TwoQubitLorenzCurves}
\end{figure}

The squares of the Hilbert-Schmidt norms of these orbits are listed in Table~\ref{Table:TwoQubitLambdaNorms}. Here we can see that the largest Hilbert-Schmidt norm is achieved by the CNC-type orbit of $\Lambda$-polytope vertices with parameter $m=1$. Therefore, Conjecture~\ref{Conjecture:LambdaOutradius} also holds for $n=2$.

\begin{table}[ht]
    \centering
    \begin{tabular}{|l|c|}
        \hline
        Orbit & $\Tr(A^2)$ \\
        \hline\hline
        CNC ($m=1$) & $2.0$ \\
        CNC ($m=2$) & $1.5$ \\
        Row 2 & $1.375$ \\
        Row 3 & $1.25$ \\
        Row 4 & $1.375$ \\
        Row 5 & $0.875$ \\
        Row 6 & $1.\bar{3}$ \\
        Row 7 & $1.34375$ \\
        \hline
    \end{tabular}
    \caption{The square of the Hilbert-Schmidt norm $\Tr(A^2)$ of vertices $A$ of the $2$-qubit $\Lambda$ polytope. The first two rows are CNC-type vertices. In the remaining rows, the label ``Row X'' refers to the rows in which they appear in Table~II of Ref.~\cite{Reichardt2009}.}
    \label{Table:TwoQubitLambdaNorms}
\end{table}

For $n=3$, a complete enumeration of vertices of $\Lambda$ is out of reach. We can gain evidence for the conjectures by generating random vertices of $\Lambda$ and checking that they satisfy the conjectures. We can generate random vertices of $\Lambda$ by choosing a random pure quantum state $\ket{\psi}$ and solving the linear program
\begin{equation}
    \max\{\Tr(\ketbra{\psi}{\psi}X)\;|\;X\in\Lambda\}.
\end{equation}
By the fundamental theorem of linear programming, the optimal value of this program will be achieved at a vertex of the feasible set, namely, $\Lambda$.

Using this method, we generated $10{,}000{,}000$ random $3$-qubit vertices which fell into over $4{,}285{,}000$ distinct orbits under the action of the unitary group. All of the spectra of the of the vertices generated are majorized by the type $m=1$ CNC vertices, except for those of the other CNC-type vertices. Therefore, Conjecture~\ref{Conjecture:LambdaMajorization} holds for all of the vertices generated. Some representative examples are shown in Figure~\ref{Figure:ThreeQubitLorenzCurves}. Similarly, all of the vertices generated have Hilbert-Schmidt norm less than or equal to the type $m=1$ CNC vertices. A histogram of the Hilbert-Schmidt norm of these vertices is shown in Figure~\ref{Figure:ThreeQubitHSNormHistogram}.

\begin{figure}
    \centering
    \includegraphics[width=0.8\linewidth]{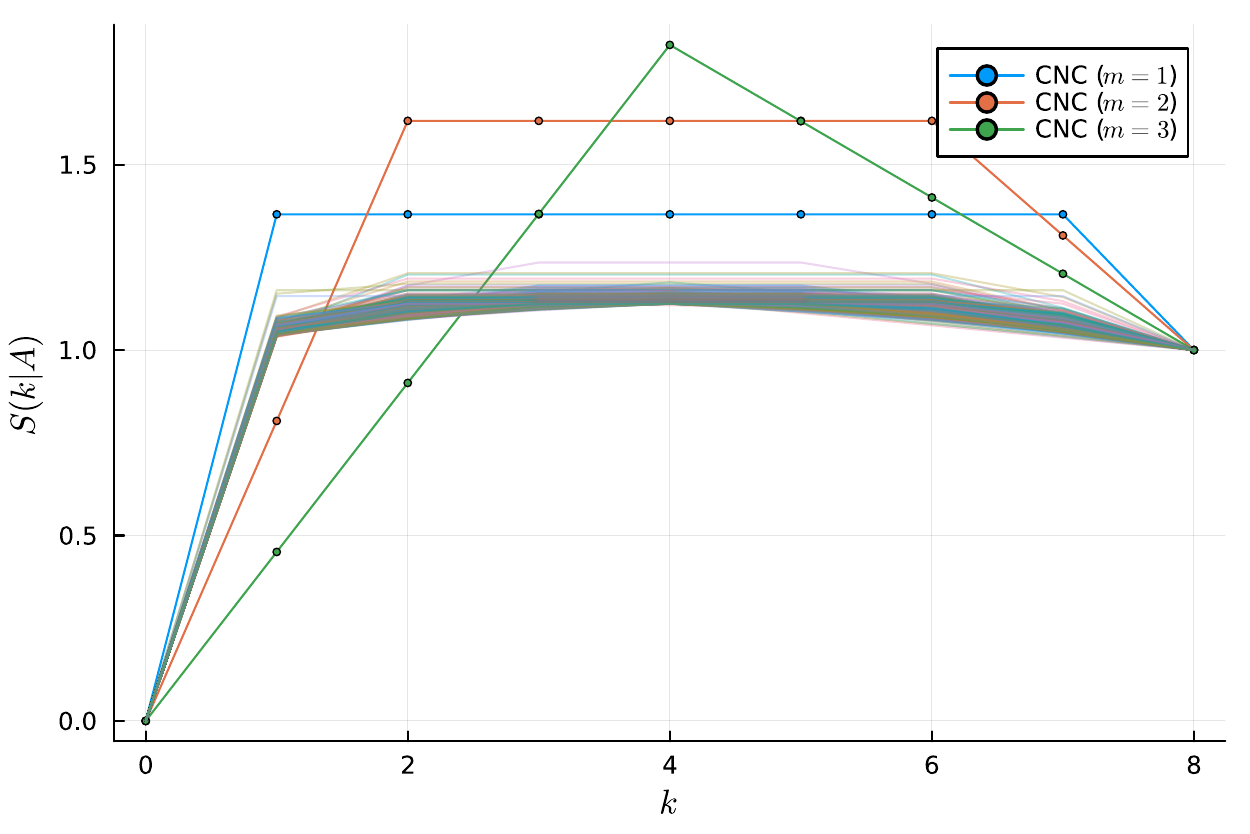}
    \caption{Lorenz curves, $S(k\,|\,A):=\sum_{\ell=1}^k\lambda_\ell^\downarrow(A)$, of three qubit $\Lambda$ polytope vertices $A$. The CNC-type vertices are highlighted. Also shown are the top $1000$ out of $10{,}000{,}000$ randomly generated $3$-qubit vertices that come closest to surpassing the Lorenz curve of the $m=1$ CNC vertex.}
    \label{Figure:ThreeQubitLorenzCurves}
\end{figure}

\begin{figure}
    \centering
    \includegraphics[width=0.8\linewidth]{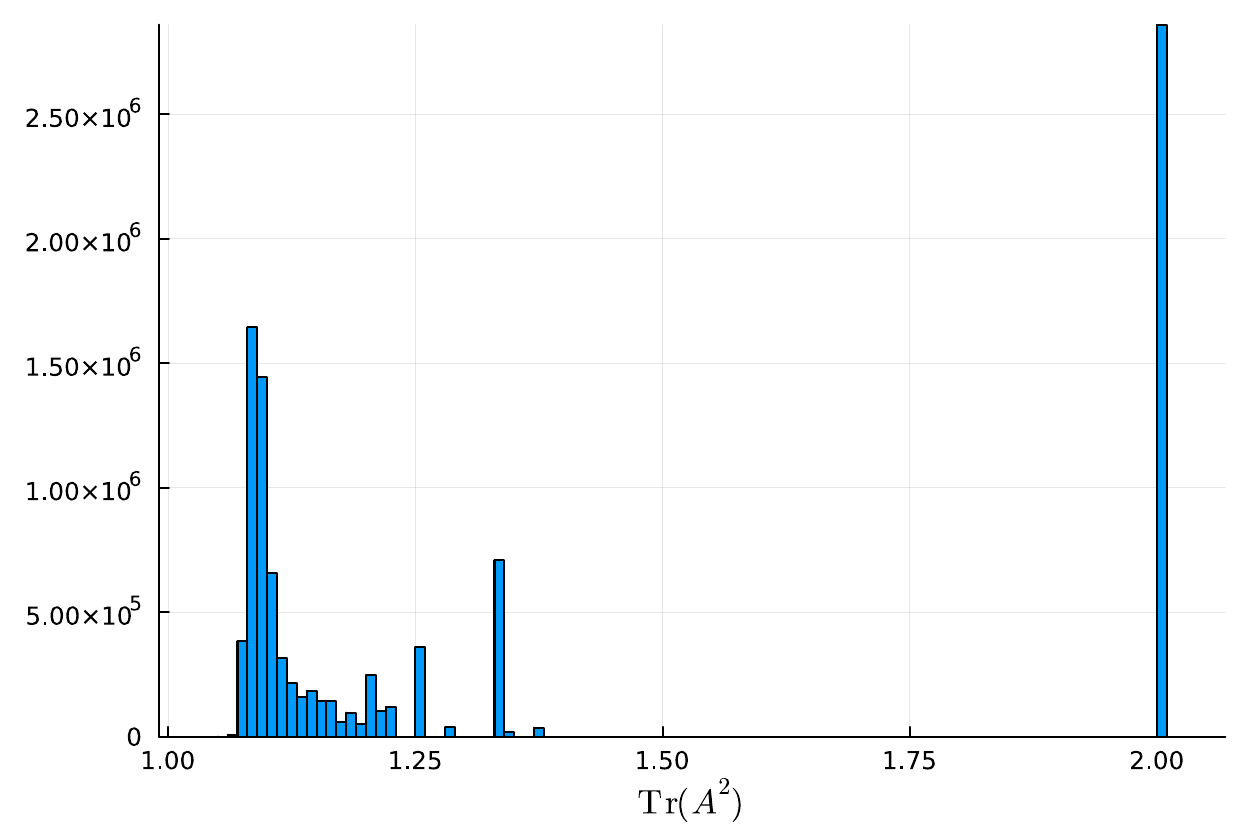}
    \caption{Histogram showing the square of the Hilbert-Schmidt norm for $10{,}000{,}000$ (non-unique) randomly generated vertices of the $3$-qubit $\Lambda$ polytope. The spike at $2$ results from the fact that the method of generating random vertices of $\Lambda$ is biased towards generating CNC-type vertices over the much more common non-CNC vertices.}
    \label{Figure:ThreeQubitHSNormHistogram}
\end{figure}

Generating $100{,}000$ random $4$-qubit vertices yielded similar results. All non-CNC vertices generated had spectra majorized by the $m=1$ CNC spectrum, and all had Hilbert-Schmidt norm less than $\sqrt{2}$ (see Figures~\ref{Figure:FourQubitLorenzCurves} and~\ref{Figure:FourQubitHSNormHistogram}). Therefore, the conjectures hold for of these randomly generated vertices.

\begin{figure}
    \centering
    \includegraphics[width=0.8\linewidth]{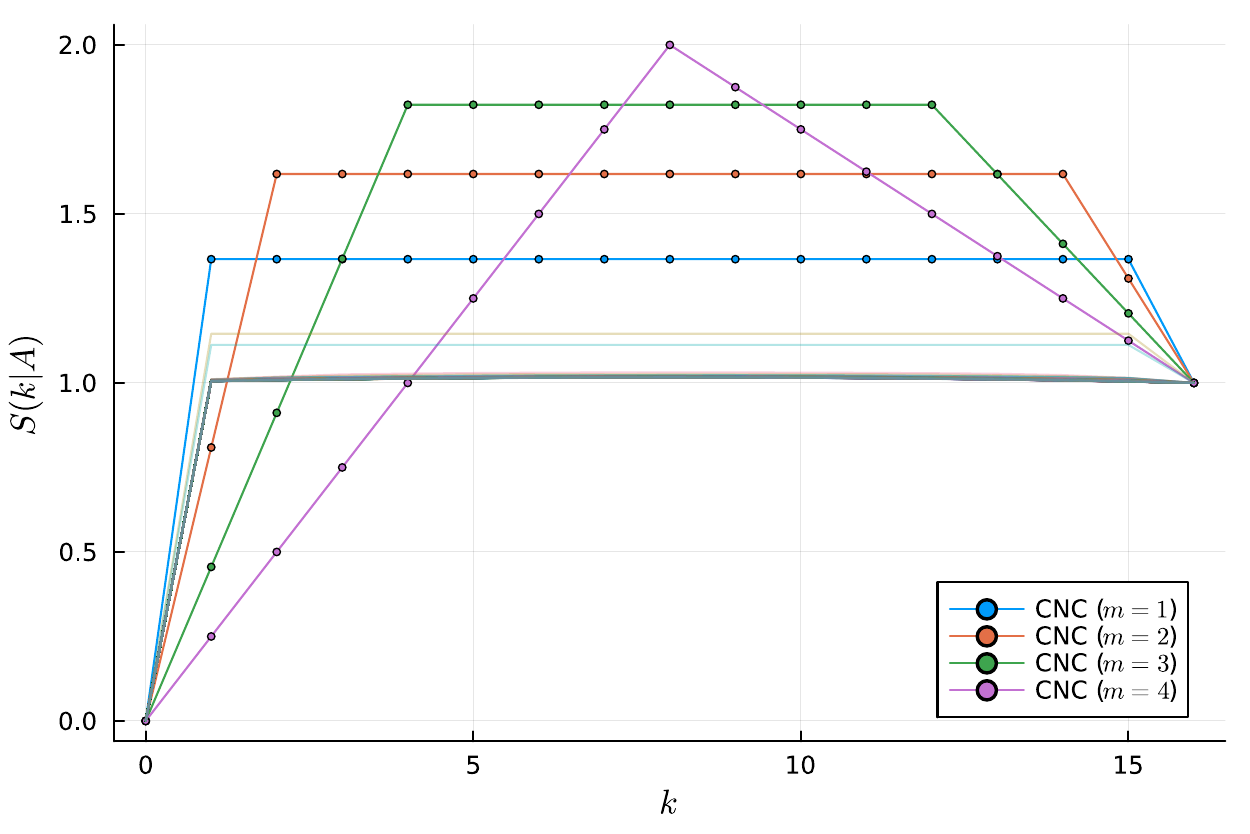}
    \caption{Lorenz curves, $S(k\,|\,a):=\sum_{\ell=1}^k\lambda_\ell^\downarrow(A)$, of four qubit $\Lambda$ polytope vertices $A$. The CNC-type vertices are highlighted. Also shown are the top $1000$ out of $100{,}000$ randomly generated $4$-qubit vertices that come closest to surpassing the Lorenz curve of the $m=1$ CNC vertex.}
    \label{Figure:FourQubitLorenzCurves}
\end{figure}

\begin{figure}
    \centering
    \includegraphics[width=0.8\linewidth]{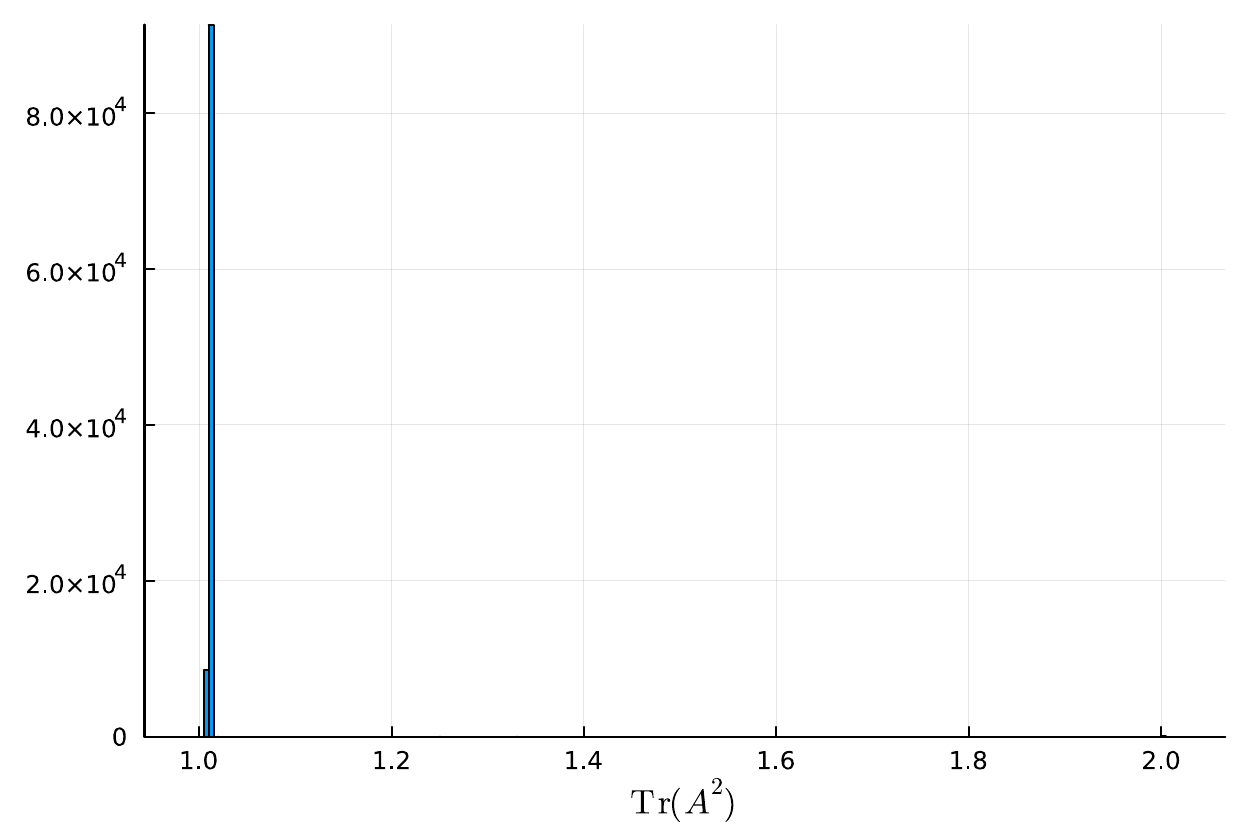}
    \caption{Histogram showing the square of the Hilbert-Schmidt norm for $100{,}000$ (non-unique) randomly generated vertices of the $4$-qubit $\Lambda$ polytope.}
    \label{Figure:FourQubitHSNormHistogram}
\end{figure}

In addition to numerical evidence, we can give a heuristic argument for why Conjecture~\ref{Conjecture:LambdaOutradius} may be true. First note that since the function $X\mapsto\Tr(X^2)$ is convex, its maximum over $\Lambda$ is achieved at a vertex of $\Lambda$. Any element $X$ of $\Lambda$ can be expanded in the Pauli basis as
\begin{equation}
    X=\frac{1}{2^n}\sum_{u\in\mathbb{Z}_d^{2n}}x_uT_u,
\end{equation}
with $x_0=1$ and $|x_u|\le 1$ for all $u\in\mathbb{Z}_d^{2n}$. The purity of $X$ can be written in terms of these coefficients as $\Tr(X^2)=\frac{1}{2^n}\sum_u|x_u|^2$. The type $m=1$ CNC vertices have coefficients in $\{0,\pm1\}$, and have maximal Pauli basis support among all integer-valued vertices of $\Lambda$~\cite{Heimendahl2019,RaussendorfZurel2020}. Since every integer-valued vertex is of CNC type, any vertex $A$ of $\Lambda$ with $\Tr(A^2)>2$ would require support strictly larger than that of the $m=1$ CNC vertices, and therefore would also require fractional Pauli basis coefficients.

This fact can give some intuition for why such vertices may not exist. Let $A_\Omega^\gamma$ be a type $m=1$ CNC vertex, and let $X(t)=A_\Omega^\gamma+tB$ be a feasible movement away from $A_\Omega^\gamma$, that is, $B$ is chosen such that $X(t)\in\Lambda$ for all sufficiently small $t>0$. Then,
\begin{equation}
    \Tr(X(t)^2)-\Tr({A_\Omega^\gamma}^2)=\frac{2t}{2^n}\sum_{u\in\mathbb{Z}_d^{2n}}a_ub_u+\frac{t^2}{2^n}\sum_{u\in\mathbb{Z}_d^{2n}}b_u^2,
\end{equation}
where $a_u$ and $b_u$ are the Pauli basis coefficients of $A_\Omega^\gamma$ and $B$ respectively. The coefficients of $X(t)$ initially in $\{\pm1\}$ at $t=0$ must move towards $0$ as $t$ increases if $X(t)$ is to remain in $\Lambda$, whereas coefficients initially equal to zero contribute to the purity only at second order as $t$ increases. The KKT conditions~\cite{NocedalWright2006} show that the first term is strictly negative along every nonzero feasible direction for sufficiently small $t$. Intuitively, enlarging the Pauli basis support from an $m=1$ CNC vertex requires initially decreasing the purity. This argument can be made more rigorous using the KKT conditions~\cite{NocedalWright2006} to show that the type $m=1$ CNC vertices are local maximizers of the function $X\mapsto\Tr(X^2)$ in $\Lambda$, but it fails to resolve Conjecture~\ref{Conjecture:LambdaOutradius} by showing that they are also global maximizers.

\end{document}